\documentclass[12pt, a4paper]{amsart}
\usepackage{amsmath, amsfonts, amsthm, amssymb, bm, bbm}
\usepackage{graphicx}
\usepackage{enumerate}
\usepackage{enumitem}
\usepackage{multirow}
\usepackage{color}
\usepackage{natbib}
\usepackage[colorlinks,citecolor=blue,urlcolor=blue]{hyperref}
\usepackage{amsaddr}
\usepackage{longtable}
\usepackage{multirow}
\usepackage{arydshln}
\newcommand{\bY}{\bm{Y}}
\newcommand{\tY}{\tilde{Y}}
\newcommand{\by}{\bm{y}}
\newcommand{\ty}{\tilde{y}}

\newcommand{\data}{\text{data}}

\newcommand{\E}{\text{E}}
\newcommand{\textS}{\text{S}}
\newcommand{\textD}{\text{D}}

\newcommand{\bp}{\bm{p}}

\newcommand{\bxi}{\bm{\xi}}

\newcommand{\bv}{\bm{v}}
\newcommand{\bz}{\bm{z}}
\newcommand{\Dir}{\text{Dir}}
\newcommand{\Beta}{\text{Beta}}
\newcommand{\TBeta}{\text{TBeta}}

\newcommand{\Weibull}{\text{Weibull}}
\newcommand{\Unif}{\text{Unif}}
\newcommand{\obs}{\text{obs}}
\newcommand{\mis}{\text{mis}}

\newcommand{\HH}{\mathcal{H}}

\newcommand{\A}{\mathcal{A}}

\newcommand{\C}{\mathcal{C}}
\newcommand{\Z}{\mathcal{Z}}

\newcommand{\dd}{\text{d}}
\newcommand{\bb}{\bm{b}}
\newcommand{\tL}{\text{L}}
\newcommand{\tR}{\text{R}}
\newcommand{\bone}{\mathbbm{1}}

\DeclareMathOperator*{\argmax}{arg\,max}
\DeclareMathOperator*{\argmin}{arg\,min}

\definecolor{brown}{rgb}{0.8, 0.33, 0.1}

\numberwithin{equation}{section}
\numberwithin{table}{section}
\numberwithin{figure}{section}

\theoremstyle{plain}
\newtheorem{thm}{Theorem}[section]
\newtheorem{lemma}{Lemma}[section]
\newtheorem{corollary}{Corollary}[section]
\newtheorem{proposition}{Proposition}[section]
\theoremstyle{definition}
\newtheorem{definition}{Definition}[section]
\theoremstyle{remark}
\newtheorem{remark}{Remark}[section]

\makeatletter
\newcommand{\addresseshere}{%
  \enddoc@text\let\enddoc@text\relax
}
\makeatother

\makeatletter
\def\paragraph{\@startsection{paragraph}{4}%
  \z@{1ex \@plus1ex \@minus.2ex}{-\fontdimen2\font}%
  {\normalfont\itshape}}
\makeatother

\begin{document}

\def\spacingset#1{\renewcommand{\baselinestretch}%
{#1}\small\normalsize} \spacingset{1}


\title[Dose-Finding Designs with Late-Onset Toxicities]{Statistical Frameworks for Oncology Dose-Finding Designs with Late-Onset Toxicities: A Review}

\thanks{\mbox{} \\ Submitted to \emph{Statistical Science} in December 2019}

\author[T. Zhou]{Tianjian Zhou$^1$}
\address{Department of Statistics, Colorado State University \\}
\author[Y. Ji]{Yuan Ji$^2$}
\address{Department of Public Health Sciences, University of Chicago}
\email{$^1$tianjian.zhou@colostate.edu, $^2$yji@health.bsd.uchicago.edu}


\begin{abstract}
In oncology dose-finding trials, due to staggered enrollment, it might be desirable to make dose-assignment decisions in real-time in the presence of pending toxicity outcomes, for example, when the dose-limiting toxicity is late-onset.
Patients' time-to-event information may be utilized to facilitate such decisions.
We review statistical frameworks for time-to-event modeling in dose-finding trials and summarize existing designs into two classes: TITE designs and POD designs. TITE designs are based on inference on toxicity probabilities, while POD designs are based on inference on dose-finding decisions.  These two classes of designs contain existing individual designs as special cases and also give rise to new designs.
We discuss and study the theoretical properties of these designs, including large-sample convergence properties, coherence principles, and the underlying decision rules.
To facilitate the use of these designs in practice, we introduce efficient computational algorithms and review common practical considerations, such as safety rules and suspension rules.
Finally, the operating characteristics of several designs are evaluated and compared through computer simulations.
\end{abstract}


\spacingset{1.45}

\maketitle

\section{Introduction}
\label{sec:intro}

Phase I dose-finding studies are usually first-in-human trials for a new drug to go through before it can be approved for general use in the public.  Oncology  dose-finding trials aim to identify an efficacious dose with a tolerable level of toxicity,  known as the maximum tolerated dose (MTD),  which may be selected as the chosen dose for treating thousands of patients in subsequent clinical trials of later phases. The failure of determining an appropriate dose could lead to failure of later-phase trials and the final drug approval.
Therefore, to increase the chance of identifying a safe and efficacious dose, dose-finding trials rely on efficient statistical designs and motivate several interesting statistical problems and challenges.
First,  oncology  dose-finding trials are usually conducted with relatively small sample sizes (10 to 50 patients), which may not support statistical models of great complexity. Therefore, mainstream dose-finding designs are typically built upon relatively simple statistical models.
Second, in dose-finding trials, safety concerns and ethical constraints often outplay statistical optimality in decision-making.
Third, to reduce the enormous cost of new drug development, there has always been a need to shorten trial duration with advanced statistical tools. 
Dose-finding is a broad topic, and in this paper, we focus on the third problem of shortening trial duration by utilizing time-to-event information.
It is estimated that a phase I clinical trial typically costs several thousand dollars a day to maintain its regular operation (see, e.g., \citealp{sertkaya2016key}). Therefore, financial savings on trial duration is very desirable. 
In addition, due to the competition in drug development, the race to being the first-in-class drug is also important, and shorter trials give a higher chance of being the first to go to market.
For all these reasons and others that will be more clear next, novel statistical designs have been developed to speed up dose-finding trials.


Formally, the MTD is defined as the highest dose with toxicity probability close to or lower than a pre-specified target rate $p^*$. 
The type of toxicity is usually severe, like organ failure, and is called dose-limiting toxicity (DLT).
The premise behind dose-finding trials is that both the toxicity and efficacy of a treatment monotonically increase with the dose level. 
A dose level that is too low can not provide needed efficacy, e.g. anti-tumor activity, while a dose level that is too high might induce severe toxicity.
Therefore, it is crucial to find an appropriate dose that has the highest possible efficacy while maintains tolerable toxicity.
Usually, a grid of discrete dose levels are investigated, and cohorts of patients are sequentially enrolled and adaptively treated at the ascending dose levels based on the previously observed data. 
The trial objectives include the identification of the MTD and the estimation of the dose-toxicity relationship, as well as maximizing the chance of treating patients at safe and efficacious doses.

The evaluation of DLT is conducted by following patients post-treatment within a time window, during which DLT outcomes and times to DLTs are recorded. 
If a patient does not experience any DLT during the follow-up window, the patient is declared having no DLT.
Most existing designs require the DLT evaluation of all the previously enrolled patients to be completed before they can make a treatment assignment for the next cohort of patients.
Consequently, we refer to this type of designs as \emph{complete-data designs}.
Examples of complete-data designs include the 3+3 design \citep{storer1989design}, continual reassessment method (CRM, \citealp{o1990continual, goodman1995some, shen1996consistency, o1996continual}), 
escalation with overdose control (EWOC, \citealp{babb1998cancer}),
product-of-beta prior method (PBP, \citealp{gasparini2000curve}),
cumulative cohort design (CCD, \citealp{ivanova2007cumulative}),
Bayesian logistic regression model (BLRM, \citealp{neuenschwander2008critical}),
modified toxicity probability interval design (mTPI, \citealp{ji2010modified}; \citealp{ji2013modified}), 
product of independent beta probabilities escalation design (PIPE, \citealp{mander2015product}),
Bayesian optimal interval design (BOIN, \citealp{liu2015bayesian};
 \citealp{yuan2016bayesian}), mTPI-2 design (\citealp{guo2017bayesian}), keyboard design (\citealp{yan2017keyboard}), semiparametric dose finding method (SPM, \citealp{clertant2017semiparametric}; \citealp{clertant2018semiparametric}),
 Bayesian uncertainty-directed design (BUD, \citealp{domenicano2019bayesian}),
surface-free design (SFD, \citealp{mozgunov2020surface}),
and i3+3 design (\citealp{liu2019design}), among many others. 
The majority of these designs are for \emph{single-agent} dose-finding trials (i.e., a single investigational drug is studied), which are the focus of this article. 
For therapies of which the toxicity is acute and can be ascertained relatively quickly, such as cytotoxic chemotherapies,
waiting for the DLT evaluation of previous patients may not be a concern, as the DLT assessment window can be short.
However, for therapies that usually have late-onset toxicity, such as immunotherapies \citep{weber2015toxicities, kanjanapan2019delayed},  it is more sensible to use a relatively long assessment window. 
This may cause difficulty for complete-data designs to operate, since patient enrollment needs to be frequently suspended until the previous patients have finished their assessment.
The same difficulty arises when patient accrual is relatively fast compared to the length of the assessment window.
For example, in Figure \ref{fig:dose_finding}(a), while waiting for the DLT outcomes of the first 3 patients, new patient enrollment needs to be suspended, and 3 eligible patients have to be turned away.
Trial suspension is undesirable in practice for two reasons.
First, trial duration is prolonged, which delays scientific research and drug development. Second, subsequent patients that are available for enrollment need to be turned away, which results in a delay in their cancer care. Many patients participating in the trial do not have alternative choices for treatment, and the trial may be their last treatment option.
Their diseases may also be in rapid deterioration, thus they are in need of immediate treatment.

\begin{figure}[h!]
\center
\begin{tabular}{ccc}
\includegraphics[width = 0.45\textwidth]{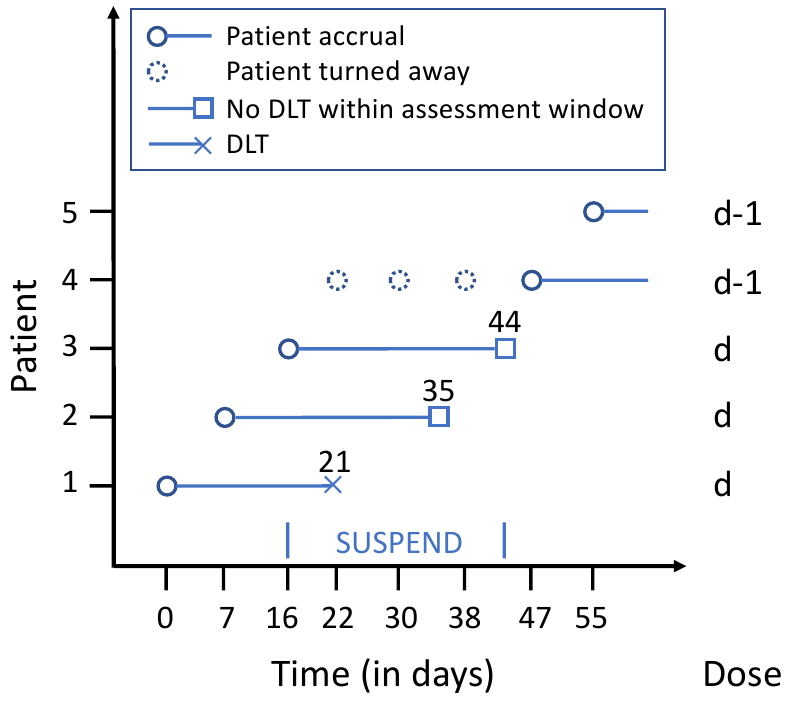} 
& \hspace{2mm} &
\includegraphics[width = 0.45\textwidth]{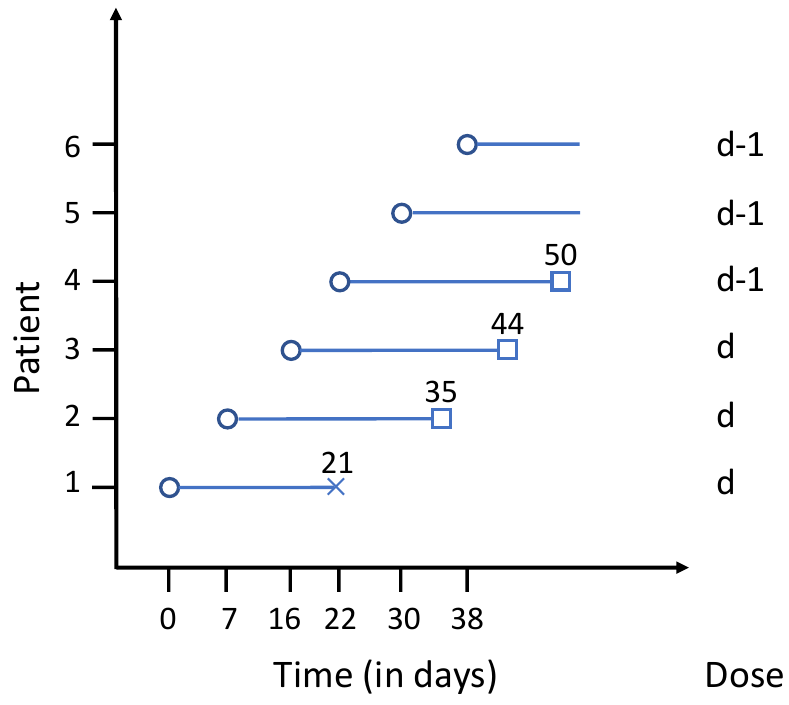} \\
(a) Complete-data designs & & (b) Time-to-event designs
\end{tabular}
\caption{Illustration of complete-data designs and time-to-event designs. Suppose the target DLT rate is 0.17, the length of the DLT assessment window is 28 days, and patients are enrolled in cohorts of 3. Using complete-data designs (a), the trial needs to be suspended while waiting for the DLT outcomes of the first 3 patients. Using time-to-event designs (b), the suspension may be avoided.}
\label{fig:dose_finding}
\end{figure}


To address these practical concerns, several designs have been proposed to allow consecutive patient accrual even if some enrolled patients are still pending for DLT assessment. 
These include the time-to-event CRM (TITE-CRM, \citealp{cheung2000sequential, normolle2006designing}), rolling six design (R6, \citealp{skolnik2008shortening}), 
expectation-maximization CRM (EM-CRM, \citealp{yuan2011robust}), data augmentation CRM (DA-CRM, \citealp{liu2013bayesian}),
time-to-event BOIN design (TITE-BOIN, \citealp{yuan2018time}), 
time-to-event PIPE design (TITE-PIPE, \citealp{wheeler2019bayesian}),
rolling TPI design (R-TPI, \citealp{guo2019rtpi}), time-to-event keyboard design (TITE-keyboard, \citealp{lin2018time}), and probability-of-decision TPI design (POD-TPI, \citealp{zhou2019pod}).
Except for R6 and R-TPI,  these designs utilize time-to-event information to make treatment assignments thus are referred to as \emph{time-to-event designs}.
As an example, in Figure \ref{fig:dose_finding}(b), when the 4th patient is available for enrollment, patients 2 and 3 are still being followed without definitive outcomes. Based on the DLT outcome of patient 1, time-to-DLT information of patient 1 and follow-up time information of patients 2 and 3, a time-to-event design may enroll the patient and de-escalate the dose level, which avoids the trial suspension.

In this article, we summarize the vast literature of time-to-event dose-finding designs into two general classes. 
See Figure \ref{fig:illustration}.
The key component is the construction of the likelihood function with time-to-event data, and
the primary interest is inference on toxicity probabilities. Specifically, two equivalent modeling approaches can be taken for the likelihood construction.
The statistical frameworks give rise to two classes of time-to-event designs, which contain the existing designs as special cases and also lead to new designs.
The first class of time-to-event designs, called TITE designs, make dose-finding decisions based on inference on toxicity probabilities.
The second class of time-to-event designs, called POD (probability of decision) designs, is a new type of designs that directly make inference on dose-finding decisions when DLT outcomes may be pending.
The POD designs directly reflect the confidence of possible decisions and offer the investigators and regulators a way to properly assess and control the chance of making incompatible decisions when not all patients have been completely followed.

\begin{figure}[h!]
\center
\includegraphics[width = 0.9\textwidth]{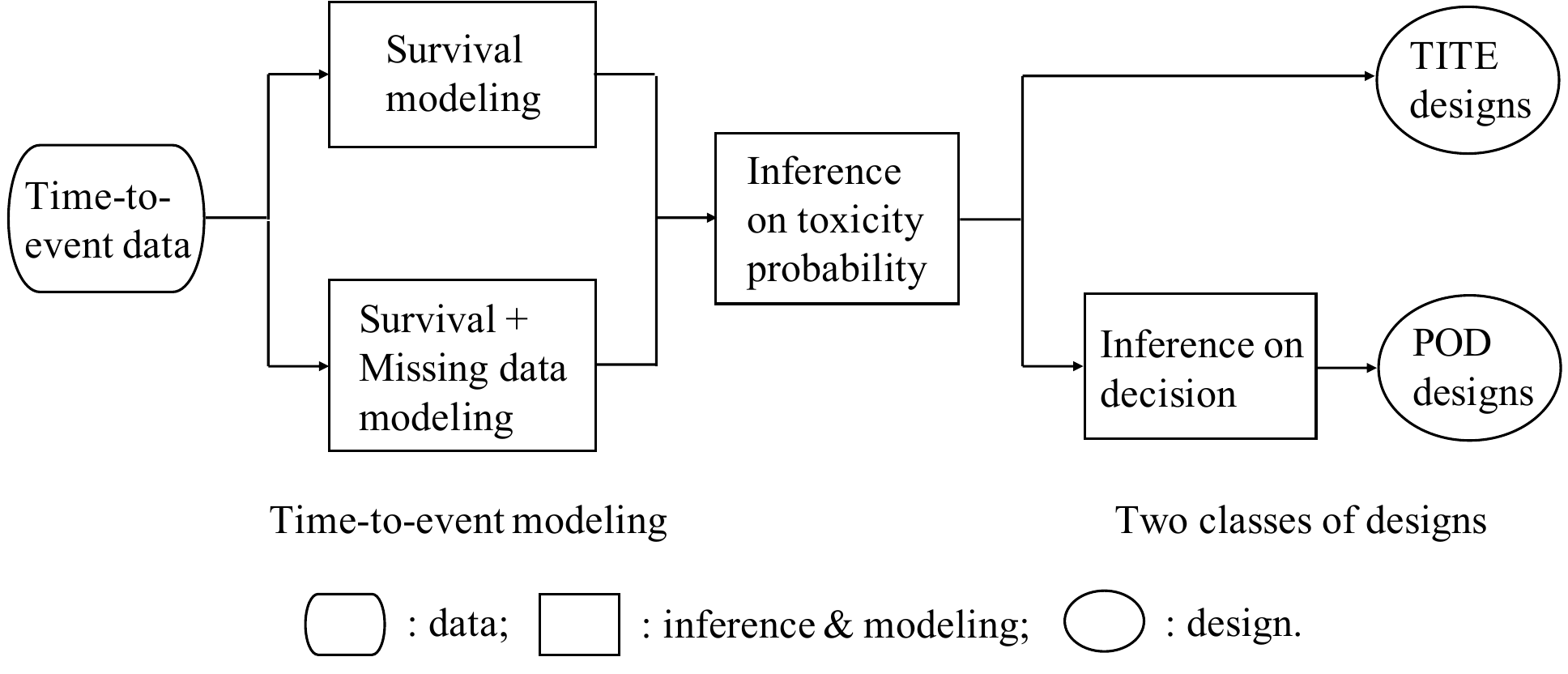}
\caption{Illustration of statistical frameworks for time-to-event modeling in dose-finding trials and two classes of time-to-event designs.}
\label{fig:illustration}
\end{figure}

Along with the statistical frameworks, we introduce several computational algorithms to facilitate the use of time-to-event designs in practice. 
We discuss and study theoretical properties of time-to-event designs, such as large-sample convergence properties, coherence principles and the underlying decision rules, with a focus on interval-based and curve-free designs.
We also review practical considerations of time-to-event designs, which are important in the execution of clinical trials.
Usually, ad-hoc rules need to be imposed to ensure that the designs satisfy safety concerns and ethical constraints.
Lastly, we examine finite-sample operating characteristics of some designs through computer simulations.

The remainder of the paper is structured as follows.
In Section \ref{sec:complete_data_design}, we give a brief review of complete-data designs. 
In Section \ref{sec:framework}, we review statistical frameworks for time-to-event modeling in dose-finding trials. 
In Section \ref{sec:two_class_designs}, we introduce the two  summarized classes of time-to-event designs: TITE and POD designs.
In Section \ref{sec:property}, we study theoretical properties of time-to-event designs.
In Section \ref{sec:practical}, we discuss practical considerations. 
We assess the operating characteristics of several existing and newly proposed dose-finding designs under the TITE and POD frameworks via simulation studies in Section \ref{sec:simulation}. Finally, we conclude with a discussion in Section \ref{sec:discussion}.
Technical details, including the computational algorithms and proof of the theoretical results, are provided in the appendix.
The R source code for implementing several time-to-event designs and replicating the simulation studies is also provided as a supplementary file.

\section{Review of Complete-Data Designs}
\label{sec:complete_data_design}

We start with a brief review of complete-data designs.
At a given moment in a dose-finding trial, suppose $N$ patients have been treated, and the $(N+1)$th patient is eligible for enrollment. 
Let $Z_i \in \{1, \ldots, J \}$ denote the dose assigned to patient $i$, where $J$ is the number of available doses in the trial.
Each patient is supposed to be followed for a fixed period of time $W$, and we use $Y_i = 1$ or $0$ to represent the binary outcome of whether or not patient $i$ experiences DLT within the time window, respectively. 
For example, in many oncology trials, $W = 28$ days.
Let $p_z$ denote the probability of DLT at dose $z \in \{1, \ldots, J \}$ within the assessment window. 
The conditional distribution of $Y_i$ given $Z_i$ and $p_z$ is commonly modeled with a Bernoulli distribution, 
\begin{align}
\Pr(Y_i = y \mid Z_i = z, p_z) = p_z^y (1 - p_z)^{1-y}, \quad 
y \in \{ 0, 1 \}.
\label{eq:bernoulli_dist}
\end{align}
A widely recognized assumption is that the DLT probability is monotone with the dose level, i.e. $p_1 \leq p_2 \leq \cdots \leq p_J$.

Suppose $Y_i$'s are fully observed for the first $N$ patients, and denote by $\HH_{N}^* = \{ (Y_i, Z_i) :  i \leq N \}$ the previous history of observations.
A complete-data design $\A^*$ can be viewed as a function of $\HH_{N}^*$, which prescribes a dose $\A^*(\HH_{N}^*)$ for the new patient through two steps: (1) making inference about $p_z$'s, and (2) translating such inference to a dose-finding decision. 
As we shall see later, the monotonicity assumption of the DLT probabilities can play a role in both of these two steps.
Inference on $p_z$'s can be based on the likelihood,
\begin{align}
L( \bp \mid \HH_N^* ) = \prod_{i = 1}^N p_{z_i}^{y_i} (1 - p_{z_i})^{1 - y_i},
\label{eq:likelihood_complete}
\end{align}
where $\by = (y_1, \ldots, y_N)$ and $\bz = (z_1, \ldots, z_N)$ are the observed outcomes and dose assignments for the $N$ patients, respectively, and $\bp = (p_1, \ldots, p_J)$ is the vector of toxicity probabilities.
As mentioned in Section \ref{sec:intro}, there is a rich literature on complete-data dose-finding designs. 
The existing literature can be roughly divided into model-based
and rule-based (i.e., model-free) designs \citep{zhou2020emerging}. Two classes of
model-based designs, one using parametric dose-toxicity response curves
and the other curve-free, accommodate the monotone dose-toxicity
relationship with different approaches. The curve-based designs
directly build the monotonicity assumption into the proposed curves,
e.g., using a positive coefficient, while curve-free designs 
implicitly accommodate the assumption by an up-and-down decision
framework. For example, when the observed toxicity data suggest that a
dose is below or above the MTD target, the decision is to escalate (go up) or de-escalate (go down) the dose level, respectively.
Below, we provide a brief review of six main-stream complete-data dose-finding designs: CRM, BOIN, mTPI-2, keyboard, SPM and i3+3.
The 3+3 design is excluded from the discussion, as it does not allow the specification of a particular DLT target $p^*$ and a maximum sample size.
It is widely recognized that 3+3 has worse performance than, e.g., mTPI \citep{ ji2013modified}.
We denote by $N_z = \sum_{i = 1}^N \bone(Z_i = z)$, $n_z = \sum_{i = 1}^N \bone(Z_i = z, Y_i = 1)$, and $m_z = \sum_{i = 1}^N \bone(Z_i = z, Y_i = 0)$ the total numbers of patients, DLTs, and non-DLTs at dose $z$, respectively.

\begin{description}[leftmargin = 0pt, topsep=2mm, listparindent=\parindent, itemsep=3mm]

\item[Continual Reassessment Method (CRM)]
The CRM design assumes a dose-toxicity response curve $p_z = \phi(z, \alpha)$, where $\alpha$ is an unknown parameter.
Based on the monotonicity assumption of the DLT probabilities, $\phi$ should monotonically increase with $z$. 
For example, a commonly used dose-toxicity curve is 
$\phi(z, \alpha) = p_{0z}^{\exp(\alpha)}$, where $p_{0z}$'s are pre-specified constants satisfying $p_{01} < \cdots < p_{0D}$. The likelihood becomes
\begin{align*}
L( \alpha \mid \HH_{N}^* ) = \prod_{i = 1}^N \phi(z_i, \alpha)^{y_i} [1 - \phi(z_i, \alpha)]^{1 - y_i}.
\end{align*}
Inference on $\alpha$ can be Bayesian \citep{o1990continual} or based on maximum likelihood \citep{o1996continual}. 
From a Bayesian perspective, a prior distribution $\pi_0(\alpha)$ is specified for $\alpha$ (for example, $\alpha \sim \text{N}(0, 1.34^2)$), leading to the posterior $\pi(\alpha \mid \by, \bz) \propto \pi_0(\alpha) L(\alpha \mid \by, \bz)$. 
The DLT probabilities can thus be estimated by $\hat{p}_z = \int \phi(z, \alpha) \pi(\alpha \mid \by, \bz) d \alpha$.
On the other hand, the maximum likelihood estimate (MLE) for $\alpha$ is $\hat{\alpha} = \argmax_{\alpha} L(\alpha \mid \by, \bz)$, and $p_z$ can be estimated by $\hat{p}_z = \phi(z, \hat{\alpha})$. In both cases, the dose $d^* = \argmin_{z} | \hat{p}_{z} - p^* |$ is recommended for the next patient, subject to some practical safety restrictions \citep{goodman1995some, cheung2005coherence}.

\item[Bayesian Optimal Interval (BOIN) Design]  

We refer to the local BOIN design \citep{liu2015bayesian}, which considers a statistical test of three hypotheses:
\begin{align}
H_0: p_d = p^*, \quad 
H_1: p_d = p^{\tL}, \quad
H_2: p_d = p^{\tR}.
\label{eq:boin_hypothesis}
\end{align}
Here $d$ is the current dose level, $p^{\tL}$ denotes the highest toxicity probability that is deemed subtherapeutic such that dose escalation should be made, and $p^{\tR}$ denotes the lowest toxicity probability that is deemed overly toxic such that dose de-escalation is required. The quantities $p^{\tL}$ and $p^{\tR}$ need to be pre-specified by physicians. 
Assuming equal prior weights on the three hypotheses,
the optimal decision boundaries $\lambda^{\tL}(p^*, p^{\tL})$ and 
$\lambda^{\tR}(p^*, p^{\tR})$ minimizing the decision error rate are,
\begin{align*}
\begin{split}
\lambda^{\tL} &= \left. \log \left( \frac{1 - p^{\tL}}{1 - p^*} \right) 
\middle/
\log \left[ \frac{ p^* (1 - p^{\tL}) }{ p^{\tL} (1 - p^*) } \right] \right., \\
\lambda^{\tR} &= \left. \log \left( \frac{1 - p^*}{1 - p^{\tR}} \right) 
\middle/
\log \left[ \frac{ p^{\tR} (1 - p^*) }{ p^* (1 - p^{\tR}) } \right] \right..
\end{split}
\end{align*}
Let $\hat{p}_d = n_d / N_d$ denote the MLE for $p_d$.
If $\hat{p}_d \leq \lambda^{\tL}$, the dose is escalated for the next patient; if $\hat{p}_d \geq \lambda^{\tR}$, the dose is de-escalated; otherwise, the same dose level is retained.
Note that $\lambda^{\tL}$ and $\lambda^{\tR}$ can be pre-specified without the optimization procedure in \cite{liu2015bayesian}. When they are pre-specified, BOIN uses essentially the same up-and-down rules as the cumulative cohort design \citep{ivanova2007cumulative}.
Although the monotonicity assumption of the DLT probabilities is not imposed on the inference of $p_d$, it plays a role in the dose-finding decision. For example, if $p_d$ is deemed lower than $p^*$, then a dose-escalation decision is made, because we believe the next higher dose level will have higher DLT and efficacy probabilities and may be closer to the MTD.

\item[mTPI-2 Design]

The mTPI-2 design considers a partition of the $[0, 1]$ interval into an equivalence interval $I_{\textS} = [p^* - \epsilon_1, p^* + \epsilon_2]$, an underdosing interval $I_{\E} = [0, p^* - \epsilon_1)$, and an overdosing interval $I_{\textD} = (p^* + \epsilon_2, 1]$. Here, the subscripts ``S'', ``E'', and ``D'' stand for the corresponding dose-finding decisions: \textbf{S}tay, \textbf{E}scalation, and \textbf{D}e-escalation.
Any dose with toxicity probability inside $I_{\textS}$ is considered a true MTD and  corresponds to a decision of ``S''.
The doses in $I_{\E}$ (or $I_{\textD})$ are considered subtherapeutic (or overly toxic) and lower (or higher) than the MTD thus correspond to a decision of ``E'' (or ``D'').
The values $\epsilon_1$ and $\epsilon_2$ need to be specified by physicians.
The intervals $I_{\E}$ and $I_{\textD}$ are further divided into several sub-intervals, $I_{\E_0}, \ldots, I_{\E_{K_1}}$ and $I_{\textD_0}, \ldots, I_{\textD_{K_2}}$, such that all the sub-intervals have  the same length $(\epsilon_1 + \epsilon_2)$ except for the two intervals ($I_{\E_0}$ and $I_{\textD_0}$) reaching the boundary of $[0, 1]$.
Let $d$ denote the current dose level, and let model $\{ \mathcal{M}_d = k \}$ represent $\{ p_d \in I_k \}$, $k = \E_0, \ldots, \E_{K_1}, \textS, \textD_0, \ldots, \textD_{K_2}$. 
The mTPI-2 design is based on the following hierarchical prior model for $\mathcal{M}_d$ and $p_d$,
\begin{align}
\begin{split}
&\Pr(\mathcal{M}_d = k) = 1/(K_1 + K_2 + 3), \; \text{for $k = \E_0, \ldots, \E_{K_1}, \textS, \textD_0, \ldots, \textD_{K_2}$}; \\
&p_d \mid \mathcal{M}_d \sim \TBeta(1, 1; I_{\mathcal{M}_d}),
\end{split}
\label{eq:mtpi2}
\end{align}
where $\TBeta(\cdot, \cdot; I)$ represents a truncated beta distribution restricted to interval $I$.
The dose-assignment decision for the next patient is ``E'', ``S'', or ``D'', if $\arg \max_{k}$ $\Pr(\mathcal{M}_d = k  \mid n_d, m_d)$ belongs to $\{ \E_0, \ldots, \E_{K_1} \}$, equals S, or belongs to $\{ \textD_0, \ldots, \textD_{K_2} \}$, respectively. This is shown to be the Bayes' rule under a 0-1 loss \cite{guo2017bayesian}.

\item[Keyboard Design]

The keyboard design, similar to the mTPI-2 design, considers a partition of the $[0, 1]$ interval into sub-intervals of equal length (except for the two boundary intervals $I_{\E_0}$ and $I_{\textD_0}$).
The sub-intervals $I_{\E_1},$ $\ldots,$ $I_{\E_{K_1}},$ $I_{\textS},$ $I_{\textD_1},$ $\ldots,$ $I_{\textD_{K_2}}$ are referred to as keys, and the equivalence interval $I_{\textS}$ is referred to as the target key.
The two boundary intervals $I_{\E_0}$ and $I_{\textD_0}$ may not be long enough to form a key.
Instead of using a hierarchical prior for $p_d$, keyboard considers a simple prior $p_d \sim \Beta(1, 1)$, leading to the posterior 
\begin{align*}
p_d \mid n_d, m_d \sim \Beta(n_d + 1, m_d + 1).
\end{align*}
The dose-assignment decision for the next patient is ``E'', ``S'', or ``D'', if $\arg \max_{k}$ $\Pr(p_d \in I_k  \mid n_d, m_d)$ belongs to $\{ \E_1, \ldots, \E_{K_1} \}$, equals S, or belongs to $\{ \textD_1,$ $\ldots,$ $\textD_{K_2} \}$, respectively. Under the default prior settings as described above, the mTPI-2 and keyboard designs yield identical inference, although the priors may be changed to produce different inferences.


\item[Semiparametric Dose Finding Method (SPM)]

The SPM directly models the the location of the MTD $\gamma$, $1 \leq \gamma \leq J$. Conditional on $\gamma$ being the MTD, the support of $p_z$ is restricted to 
\begin{align*}
\text{supp}(p_z) = 
\begin{cases} 
I_{\E} = [0, p^* - \epsilon_1), & \text{if }  z < \gamma; \\
I_{\textS} = [p^* - \epsilon_1, p^* + \epsilon_2],  & \text{if } z = \gamma; \\
I_{\textD} = (p^* + \epsilon_2, 1],  & \text{if } z > \gamma.
\end{cases}
\end{align*}
This restriction guarantees the partial ordering of the $p_z$'s. 
The partition of the $[0, 1]$ interval in the SPM coincides with the mTPI and mTPI-2 designs, while the center interval is interpreted differently as an indifference interval \citep{cheung2002simple}.
The priors on $\gamma$ and $p_z$'s can be specified as follows,
\begin{align}
\begin{split}
&\Pr(\gamma = z^*) = \kappa_{z^*}, \\
&p_z \mid \gamma \sim \TBeta(c \theta_z^{\gamma} + 1, c(1 - \theta_z^{\gamma}) + 1; I_{z}^{\gamma}).
\end{split}
\label{eq:spm}
\end{align}
Here $\kappa_{z^*}$, $c$ and $\theta_{z}^{\gamma}$ are hyperparameters, and $I_z^{\gamma} = I_{\E}, I_{\textS}$ or $I_{\textD}$ for $z < \gamma$, $z = \gamma$ or $z > \gamma$, respectively.
The hyperparameter $\theta_{z}^{\gamma}$ is the prior mode of $p_z$ if $\gamma$ is the assumed MTD, and is specified in a similar fashion as CRM.
The posterior $\pi(\gamma, \bp \mid \by, \bz) \propto \pi_0(\gamma) \pi_0(\bp \mid \gamma) L( \bp \mid \by, \bz )$, and the dose $\hat{\gamma} = \argmax_{\gamma} \pi(\gamma \mid \by, \bz)$ is recommended for the next patient, again subject to some restrictions seen in the CRM design.

\item[i3+3 Design]

The i3+3 design consists of a set of algorithmic decision rules, that is, model free. Similar to mTPI-2, it considers a partition of the $[0, 1]$ interval into $I_{\E}$, $I_{\textS}$ and $I_{\textD}$. Suppose the current dose is $d$.
If $n_d / N_d \in I_{\E}$, the decision is escalation.
If $n_d / N_d \in I_{\textS}$, the decision is stay.
If $n_d / N_d \in I_{\textD}$, the decision is stay when $(n_d - 1) / N_d \in I_{\E}$ and is de-escalation otherwise.

\end{description}

For the following discussion, it is helpful to elaborate on the categorization of complete-data designs.
First, we categorize the designs 
according to how they make inference about $\bp$.
A \emph{curve-based} design (e.g., CRM and BLRM) models the toxicity probabilities with a parametric curve $p_z = \phi(z, \bm \alpha)$, which is monotonically increasing in $z$. A \emph{curve-free} design (e.g., BOIN, mTPI-2 and keyboard) does not use a parametric dose-toxicity curve to estimate $\bp$ but instead estimate each $p_z$ separately.
A \emph{semiparametric} design (e.g., SPM) does not use a parametric curve to model $\bp$ but imposes some constraint on $\bp$ to ensure its (partial) ordering.
Second, we can also categorize the designs based on how they translate inference on $\bp$ to a dose-assignment decision.
Generally, a design starts at a low dose.
At each subsequent step, a \emph{point-based} design (e.g., CRM) allocates the next cohort to $d^* = \argmin_{z} | \hat{p}_{z} - p^* |$, where $\hat{p}_{z}$ is a point estimate of $p_z$ (e.g., MLE or posterior mean) based on \eqref{eq:likelihood_complete}.
On the other hand, suppose the currently-administered dose is $d$, and $\epsilon_1 > 0$ and $\epsilon_2 > 0$ are pre-determined constants. 
\emph{Interval-based} designs make dose-finding decisions based on the interval $I_{\textS} = [p^* - \epsilon_1, p^* + \epsilon_2]$.
As discussed above, one class of interval-based designs (e.g., CCD and BOIN) make stay (at $d$), escalation (to $d+1$) or de-escalation (to $d-1$) decisions based on whether $\hat{p}_{d}$ is within, below or above $I_{\textS}$, respectively.
Another class of interval-based designs (e.g., mTPI-2 and keyboard) consider a partition of the $[0, 1]$ interval into sub-intervals (from left to right)
$\{ I_{\E_0}, \ldots, I_{\E_{K_1}} \}$, $I_{\textS}$ and $\{ I_{\textD_0}, \ldots, I_{\textD_{K_2}} \}$, where 
$I_{\textS} = [p^* - \epsilon_1, p^* + \epsilon_2]$ is the only sub-interval that contains $p^*$. 
The dose-finding decision is stay (at $d$), escalation (to $d+1$) or de-escalation (to $d-1$) if 
$ \arg \max_{k}$ $\Pr(p_d \in I_k  \mid \HH_N^*)$ equals $\textS$, belongs to $\{ \E_0, \ldots, \E_{K_1} \}$ or belongs to $\{ \textD_0, \ldots, \textD_{K_2} \}$, respectively. Here, $\Pr(p_d \in I_k  \mid \HH_N^*)$ is the posterior probability of $p_d$ falling within the interval $I_k$.
Although the existing complete-data designs differ in many aspects, 
they can be extended to time-to-event designs based on the same strategy, which will be elaborated in the next sections.

\section{Frameworks for Time-to-Event Modeling in Dose-Finding Trials}
\label{sec:framework}

\subsection{Setup}
\label{sec:notation}

Since patients enter clinical trials sequentially at random time, it is often the case that when a new patient is eligible for enrollment, some previously enrolled patients are still being followed without DLT events; thus their DLT outcomes by the end of the assessment window, $Y_i$'s, remain unknown. 
As discussed in Section \ref{sec:intro}, even when some DLT outcomes are pending, it is still desirable to enroll the patient and assign an appropriate dose.
Complete-data designs do not allow this, and time-to-event designs attempt to address this problem.
The key is to develop inference on $\bp$ and a decision rule.
With pending outcomes, inference on $\bp$ becomes less straightforward and ideally requires modeling time-to-event data, because these data provide information regarding the likelihood of the pending patients experiencing DLT in the future \citep{cheung2000sequential, yuan2018time}.
For example, a patient followed for 21 days without DLT provides different information from another followed for 2 days without DLT.
Such difference can be exploited for better inference and decision making.

Define \emph{trial time} as the number of days since the enrollment of the
first patient. 
Let $\tau_i^*$ denote the trial time when patient $i$ is enrolled. By
definition, $\tau_1^* = 0$. A patient will be followed for a duration
of $W$ days. Call $W$ the follow-up window.   We denote by $(\tau_i^* + T_i)$ the trial time when patient $i$ experiences DLT, where $T_i$ is the time-to-DLT for patient $i$. Note that $T_i$ can be greater than $W$ in reality. At any trial
time $\tau$,  patient $i$ may or may not have experienced DLT. If s/he has experienced
DLT, then $(\tau_i^* + T_i) \le \tau$. If s/he has not experienced DLT, s/he
either is still being followed or has completed $W$ days of follow-up,
but without experiencing DLT in either case, and we call the patient is
``censored''. At trial time $\tau$, a patient $i$ who is censored  has
a censoring time $U_i (\tau) = \min \{ \max(\tau - \tau^*_i, 0),
  W \}$. 
Let $Y_i$ be the indicator of whether patient $i$ experiences
DLT within the follow-up window $W$, i.e., $Y_i
= \bone (T_i \leq W)$.
We do not observe $T_i$ at trial time $\tau$ if $T_i > U_i(\tau)$, but we always observe the follow-up time of the patient, given by $V_i =  T_i \wedge U_i(\tau)$.
Similarly, we do not observe indicator $Y_i$ at trial time $\tau$ if $T_i > U_i(\tau)$ and $U_i(\tau) < W$, but we always know the current DLT status of the patient, given by $\tilde{Y}_i = \bone[T_i \leq U_i(\tau)]$.
For example, in Figure \ref{fig:dose_finding}(b), we have $\tau_1^* = 0$ and $\tau_2^* = 7$ for patients 1 and 2, respectively.
On day $\tau = 22$ since trial start, for patient 1, we have $T_1 = 21$, $U_1 = 22$, $V_1 = 21$, and $Y_1 = \tY_1 = 1$; for patient 2, we have $U_2 = 15$, $V_2 = 15$, $\tY_2 = 0$, and $T_2$ and $Y_2$ are unknown.
The available information at study time $\tau$ can be summarized by $\HH(\tau) = \{ (\tilde{Y}_i(\tau), V_i(\tau), Z_i) : i \leq N(\tau) \}$, where $N(\tau)$ is the total number of treated patients just prior to $\tau$.
A time-to-event design $\A$ can be viewed as a function of $\HH(\tau)$. That is, if a new patient is enrolled at time $\tau$, the design would assign a dose $\A[\HH(\tau)]$ for the patient.

We introduce some more notation to facilitate the upcoming discussion. 
Denote by
\begin{align*}
B_i(\tau) = 
\begin{cases} 
0, & \text{if }  \tY_i(\tau) = 0 \text{ and } V_i(\tau) < W; \\
1,  & \text{if } \tY_i(\tau) = 1 \text{, or } \tY_i(\tau) = 0 \text{ and } V_i(\tau) = W.
\end{cases}
\end{align*}
In words, $B_i(\tau) = 1$ or 0 represents that patient $i$'s DLT outcome $Y_i$ has or has not been fully assessed, respectively.
Therefore, $Y_i = \tY_i(\tau)$ if $B_i(\tau) = 1$. 
Following the convention in the missing data literature, we use $\bY_{\obs} (\tau) = \{ Y_i : B_i(\tau) = 1, i \leq N(\tau)\}$ or $\bY_{\mis} (\tau) = \{ Y_i : B_i(\tau) = 0, i \leq N(\tau) \}$ to represent the sets of DLT outcomes that have been observed or are pending at time $\tau$, respectively.
Lastly, let $N_z(\tau) = \sum_{i = 1}^{N(\tau)} \bone(Z_i = z)$ denote the number of patients that have been treated at dose $z$ just prior to $\tau$. Among the $N_z(\tau)$ patients, let $n_z(\tau)$, $m_z(\tau)$ and $r_z(\tau)$ denote the number of patients having DLT, non-DLT and pending outcomes, respectively. 
Mathematically, these can be written as 
\begin{align*}
n_z(\tau) &= \sum_{i = 1}^{N(\tau)} \bone[Z_i = z, Y_i = 1, B_i (\tau) = 1], \\
m_z(\tau) &= \sum_{i = 1}^{N(\tau)} \bone[Z_i = z, Y_i = 0, B_i (\tau) = 1], \;\; \text{and} \\
r_z(\tau) &= \sum_{i = 1}^{N(\tau)} \bone[Z_i = z, B_i (\tau) = 0].
\end{align*}

In the next sections, we will summarize existing methods that use the observed data $\HH(\tau)$  to make inference on $\bp$.

\subsection{Modeling Time-to-Toxicity Data}
\label{sec:model_tite}

The first step is to specify a model for the time-to-toxicity data.
For the following discussion, patient index $i$ is suppressed from the subscript to simplify notation when no confusion is likely. 
A straightforward modeling choice is to model $(T \mid Z)$ with a proportional hazards model, treating the dose level $Z$ as a covariate. The DLT probabilities can then be inferred based on $p_z = \Pr(T \leq W \mid Z = z)$. However, existing time-to-event designs resort to an alternative (and generally simpler) approach due to the typically small sample size at each dose level and the focus on estimation of $\bp$. 

Specifically, most existing time-to-event designs are built upon the corresponding complete-data designs. They adopt the strategies in complete-data designs for inference on $\bp$ and add a layer of time-to-toxicity modeling. 
In this way, the resulting model for $(T \mid Z)$ is relatively simple and resembles its complete-data counterpart,
and inference and assumptions on $\bp$ are explicit rather than hidden behind those on $T$. Recall that $\{ Y= 1 \}$ is equivalent to $\{ T \leq W \}$. Existing approaches still model $Y$ with a Bernoulli distribution as in Equation \eqref{eq:bernoulli_dist}, which is equivalent to assuming $\Pr(T \leq W \mid Z = z, p_z) = p_z$.
Next, write $f_{T \mid Z}(t \mid z, \bp, \bxi )$ the probability density function (pdf) of $T$ at dose level $Z = z$, where $\bxi$ denotes additional and nuisance parameters that characterize the distribution of $T$.
For $t \leq W$, 
\begin{align}
\begin{split}
f_{T \mid Z}(t \mid z, \bp, \bxi ) 
&= \Pr(Y = 1 \mid Z = z, \bp, \bxi) \cdot f_{T \mid Z, Y}(t \mid z, Y = 1, \bp, \bxi )  \\
&= p_z \cdot f_{T \mid Z, Y}(t \mid z, Y = 1, \bxi ).
\end{split}
\label{eq:model_t}
\end{align}
The first equation is true since $f_{T \mid Z, Y}(t \mid z, Y = 0, \bp, \bxi ) = 0$ for $t \leq W$.
The second equation assumes that  the conditional distribution $f_{T \mid Z, Y}(t \mid z, Y = 1, \bp, \bxi )$ does not depend on $\bp$. That is, given that a patient experiences DLT within the assessment window at a dose $z$, when the patient experiences the DLT does not depend on the toxicity probability $p_z$ of the dose. 
This is a simplifying assumption implicitly made by all existing methods (both complete-data designs and time-to-event designs). 
Note that if $f_{T \mid Z, Y}(t \mid z, Y = 1, \bp, \bxi )$ does depend on $\bp$, it should be included in the complete data likelihood \eqref{eq:likelihood_complete} according to the likelihood principle. 

To specify $f_{T \mid Z}(t \mid z, \bp, \bxi )$,
it suffices to specify $f_{T \mid Z, Y}(t \mid z, Y = 1, \bxi)$, which can be any probability density with support in $(0, W]$. 
Examples of possible specifications of $f_{T \mid Z, Y}(t \mid z, Y = 1, \bxi)$ include a uniform distribution, a piecewise uniform distribution, a discrete hazard model, a piecewise constant hazard model, and a rescaled beta distribution. 
See Appendix \ref{supp:sec:model_tite} 
for details.

The survival function of $T$ is given by
\begin{align*}
S_{T \mid Z} (t \mid z, \bp, \bxi) = \Pr(T  > t \mid Z = z, \bp, \bxi) = \int_{t}^{\infty} f_{T \mid Z}(v \mid z, \bp, \bxi ) \dd v.
\end{align*}
The survival function must satisfy $S_{T \mid Z} (W \mid z, \bp, \bxi) = 1 - p_z$. This is important, since $p_z$ only represents the probability of DLT within time window $(0, W]$.
For $t < W$, 
\begin{align*}
S_{T \mid Z} (t \mid z, \bp, \bxi) = 1 - p_z \int_{0}^t f_{T \mid Z, Y}(v \mid z, Y = 1, \bxi ) \dd v \triangleq 1 - p_z \rho(t \mid z, \bxi),
\end{align*}
where we denote by 
\begin{align}
\rho(t \mid z, \bxi) = \int_{0}^t f_{T \mid Z, Y}(v \mid z, Y = 1, \bxi ) \dd v.
\label{eq:rho}
\end{align}

\begin{remark}
\label{rmk:rho}
The function $\rho(t \mid z, \bxi)$ must satisfy: 
(1) $\rho(t \mid z, \bxi) \in [0, 1]$ for $t \in (0, W]$, and
(2) $\rho(W \mid z, \bxi)  = 1$.
Also, by definition, $\rho(t \mid z, \bxi)$ is non-decreasing in $t$.
\end{remark}

\subsection{Survival Likelihood}
\label{sec:likelihood_tite}

The likelihood of $\bp$ and $\bxi$ at time $\tau$
can be constructed based on survival modeling by treating the unknown times-to-toxicities (i.e., $\{ T_i: \tY_i(\tau) = 0, i \leq N(\tau) \}$) as censored observations.
See, for example, Section 3.5 in \cite{klein2006survival}.
To simplify notation, we omit the time index $\tau$ in the following discussion.
For the patients with observed toxicities ($\tilde{Y}_i = 1$), their contributions to the likelihood are the pdfs of the times-to-toxicities. For the patients without observed toxicities ($\tilde{Y}_i = 0$), their times-to-toxicities are right censored, and their contributions to the likelihood are the survival functions at the censoring times. In particular,
\begin{align*}
L( \bp, \bxi \mid \HH )
=
\prod_{i = 1}^N \Big[ f_{T \mid Z}(v_i \mid z_i, \bp, \bxi)^{\bone(\ty_i = 1)} 
S_{T \mid Z}(v_i \mid z_i, \bp, \bxi)^{\bone(\ty_i = 0)} \Big].
\end{align*}
The likelihood can be further written as
\begin{multline}
L ( \bp, \bxi \mid \HH )
= \prod_{i = 1}^N \Big\{ p_{z_i}^{\bone(\ty_i = 1)}   f_{T \mid Z, Y}(v_i \mid z_i, Y = 1, \bxi)^{\bone(\ty_i = 1)}  \times \\
\left[ 1 - \rho(v_i \mid z_i, \bxi) p_{z_i} \right]^{\bone(\ty_i = 0)}
 \Big\}.
\label{eq:likelihood_tite}
\end{multline}

Here, due to Remark \ref{rmk:rho}, $\rho_i \triangleq \rho(v_i \mid z_i, \bxi)$ can be interpreted as the weight of a patient who is still being followed within the assessment window. The likelihood \eqref{eq:likelihood_tite} can be considered as a weighted likelihood as in \cite{cheung2000sequential}, where a patient with a complete outcome ($\ty_i = 1$, or $\ty_i = 0$ and $v_i = W$) receives full weight, and a patient with a pending outcome ($\ty_i = 0$ and $v_i < W$) receives a weight of $\rho_i$. 
The longer the follow-up time, the larger the weight.
The term $f_{T \mid Z, Y}(v_i \mid z_i, Y = 1, \bxi)$ is not related to $\bp$ but provides information about the time-to-toxicity.

From a Bayesian perspective, with the likelihood \eqref{eq:likelihood_tite} and prior distributions $\pi_0(\bp)$ and $\pi_0(\bxi)$, inference on $\bp$ and $\bxi$ is realized using the posterior distribution, 
\begin{align*}
\pi(\bp, \bxi \mid \HH) \propto \pi_0(\bp) \pi_0(\bxi) L ( \bp, \bxi \mid \HH ).
\end{align*}
Note that $\bxi$ is a set of nuisance parameters and is of no interest to the statistical design.
In general, the posterior is not available in closed form, and Monte Carlo simulation is applied to approximate the posterior. We provide a simple computational algorithm in Appendix \ref{supp:sec:inference_sur}.
From a frequentist perspective, the maximum likelihood estimate (MLE) $\hat{\bp}$ can be used as an estimate for $\bp$. 
One can calculate $(\hat{\bp}, \hat{\bxi}) = \argmax_{\bp, \bxi} L( \bp, \bxi \mid \HH )$ by taking partial derivatives of the log-likelihood with respect to all the parameters or using other optimization techniques. Again, see more details in Appendix \ref{supp:sec:inference_sur}.
In the CRM and BLRM designs, $\bp$ is modeled by a parametric curve, $p_z = \phi(z, \bm \alpha)$, where $\bm \alpha$ denotes unknown parameters. In such cases, the likelihood \eqref{eq:likelihood_tite} is re-parameterized with respect to $\bm \alpha$, and a prior distribution $\pi_0(\bm \alpha)$ would be specified for $\bm \alpha$ (instead of $\bp$).

\subsection{Augmented Likelihood with Missing Data}
\label{sec:likelihood_missing}

The likelihood of $\bp$ and $\bxi$ can be alternatively constructed based on modeling of missing data, by treating the pending DLT outcomes (i.e., $\bY_{\mis}$) as missing and augmenting the likelihood function that incorporates the unknown $\bY_{\mis}$ as a vector of latent variables. 
Specifically, a patient having an observed toxic outcome ($Y_i = 1$ and $B_i = 1$) and a known DLT time $v_i$ contributes $p_{z_i} f_{T \mid Z, Y}(v_i \mid z_i, Y = 1, \bxi)$ to the likelihood. 
A patient having a latent toxic outcome ($Y_i = 1$ and $B_i = 0$) and a follow-up time $v_i$ contributes $p_{z_i} \int_{v_i}^W f_{T \mid Z, Y}(t \mid z_i, Y = 1, \bxi) \dd t$ to the likelihood, because the DLT will occur in the interval $(v_i, W]$. Finally, a patient with an observed or latent non-DLT outcome ($Y_i = 0$) contributes $(1 - p_{z_i})$ to the likelihood. Therefore, using \eqref{eq:rho}, the augmented likelihood is given by
\begin{multline}
L( \bp, \bxi, \by_{\mis} \mid \HH ) 
= \prod_{i = 1}^N \Big \{ p_{z_i}^{\bone(y_i = 1)} (1 - p_{z_i})^{\bone(y_i = 0)}  \times \\
f_{T \mid Z, Y}(v_i \mid z_i, Y = 1, \bxi)^{\bone(y_i = 1, b_i = 1)} \left[1 - \rho(v_i \mid z_i, \bxi) \right]^{\bone(y_i = 1, b_i = 0)} \Big \}.
\label{eq:likelihood_missing}
\end{multline}

Although the augmented likelihood involves additional parameters compared to the survival likelihood, the following proposition shows inference under both approaches is the same. 

\begin{proposition}
\label{prop:eq_mis_sur}
The derived likelihood by marginalizing \eqref{eq:likelihood_missing} over $\by_{\mis}$ is the same as the survival likelihood \eqref{eq:likelihood_tite} for $\bp$ and $\bxi$.
\end{proposition}

The proof is given in Appendix \ref{supp:sec:eq_sur_mis}. 
The theoretical result in Proposition \ref{prop:eq_mis_sur} is also implied in \cite{liu2013bayesian} but is presented more explicitly here. 
The augmented likelihood opens the door to a set of flexible computational algorithms for making inference on $\bp$. For example, the posterior distribution $\pi(\bp \mid \HH)$ can be simulated using the data augmentation method \citep{tanner1987calculation}. The MLE of $\bp$ can be calculated through the expectation-maximization algorithm \citep{dempster1977maximum}.
We elaborate these methods in Appendix \ref{supp:sec:inference_mis}. 
In the upcoming review of designs based on missing data modeling, a key component in the inference is the conditional probability,
\begin{align}
&{} \Pr(Y_{\mis, i} = 1 \mid Z_i = z_i, T_i > v_i, \bp, \bxi) \label{eq:dist_missing} \\
= &{} \frac{\Pr(T_i > v_i \mid Z_i = z_i, Y_{\mis, i} = 1, \bxi) \cdot \Pr(Y_{\mis, i} = 1 \mid Z_i = z_i, \bp)}{\sum_{y \in \{0, 1\}} \Pr(T_i > v_i \mid Z_i = z_i, Y_{\mis, i} = y, \bxi) \cdot \Pr(Y_{\mis, i} = y \mid Z_i = z_i, \bp)} \nonumber \\
= &{} \frac{ [1 - \rho(v_i \mid z_i, \bxi)] \cdot p_{z_i} }{ [1 - \rho(v_i \mid z_i, \bxi)] \cdot p_{z_i} + (1 - p_{z_i})}. \nonumber
\end{align}
That is, the probability of a pending patient experiencing DLT within the assessment window given the patient is treated at dose $z_i$ and has been followed for $v_i$ units of time.

 Instead of using the exact likelihood function \eqref{eq:likelihood_tite} or \eqref{eq:likelihood_missing}, it is possible to use an approximate likelihood to make inference about $\bp$. Based on a Taylor expansion, $( 1 - \rho_i p_{z_i} ) \approx (1 -  p_{z_i})^{\rho_i}$, thus the terms involving $\bp$ in Equation \eqref{eq:likelihood_tite} can be approximated by a Bernoulli likelihood. See \cite{lin2018time} for more details.

\section{Two Classes of Time-to-Event Designs}
\label{sec:two_class_designs}

\subsection{TITE Designs}
\label{sec:tite_design}

In Section \ref{sec:framework}, we have reviewed and summarized statistical frameworks for time-to-event modeling in dose-finding trials. 
Based on these frameworks and using each of the complete-data designs (Section \ref{sec:complete_data_design}), one can easily generate a corresponding time-to-event design. Following the literature, we call this class of designs TITE designs. 
Below, we illustrate this idea through reviewing five existing TITE designs: TITE-CRM, EM-CRM, DA-CRM, TITE-BOIN  and TITE-keyboard, and proposing three new TITE designs: TITE-TPI, TITE-SPM and TITE-i3.

\subsubsection{Existing TITE Designs}

\begin{description}[leftmargin = 0pt, topsep=2mm, listparindent=\parindent, itemsep=3mm]
\item[TITE-CRM \& EM-CRM \& DA-CRM]
The TITE-CRM design \citep{cheung2000sequential} is a TITE extension of the CRM design.
It assumes a dose-toxicity  curve $p_z = \phi(z, \alpha)$,
such that $\phi$ monotonically increases with $z$, and $\alpha$ is an unknown parameter.
The likelihood \eqref{eq:likelihood_tite} is re-parameterized with respect to $\alpha$ and becomes
\begin{multline}
L ( \alpha, \bxi \mid \HH )
= \prod_{i = 1}^N \Big\{ \phi(z_i, \alpha)^{\bone(\ty_i = 1)}  \left[ 1 - \rho(v_i \mid z_i, \bxi) \phi(z_i, \alpha) \right]^{\bone(\ty_i = 0)}  \times \\
 f_{T \mid Z, Y}(v_i \mid z_i, Y = 1, \bxi)^{\bone(\ty_i = 1)} \Big\}.
\label{eq:titecrm}
\end{multline}
By default, the conditional distribution of $[T \mid Z, Y = 1]$ is modeled by a uniform distribution, thus $\rho(v_i \mid z_i, \bxi)  \equiv v_i / W$ and $f_{T \mid Z, Y}(v_i \mid z_i, Y = 1, \bxi) \equiv 1 / W$.
Alternatively, one can model $[T \mid Z, Y = 1]$ with a piecewise-uniform distribution, which can be calibrated to match the adaptive weighting scheme described in \cite{cheung2000sequential}.
Inference on $\bp$ can be Bayesian or based on MLE. 
The dose $d^* = \argmin_{z} | \hat{p}_{z} - p^* |$ is recommended for the next patient, subject to some practical safety restrictions \citep{goodman1995some, cheung2005coherence}. Here $\hat{p}_{z}$ is an appropriate point estimate of $p_z$.

The EM-CRM \citep{yuan2011robust} and DA-CRM \citep{liu2013bayesian} are two alternative TITE extensions of the CRM design. They take the missing data modeling approach and consider the augmented likelihood \eqref{eq:likelihood_missing}. EM-CRM models $[T \mid Z, Y = 1]$ with a discrete hazard model and estimates $\bp$ using the expectation-maximization algorithm. It also has an additional layer of model averaging over different choices of skeletons to improve its robustness.
DA-CRM models $[T \mid Z, Y = 1]$ with a piecewise constant hazard model and estimates $\bp$ using the data augmentation method.
Note that the original TITE-CRM \citep{cheung2000sequential} takes a two-step approach in estimating $\bp$: first use the data to estimate $\bxi$, and then plug in the estimated $\bxi$ in the likelihood to estimate $\alpha$ and $\bp$. This would lead to different inference compared to the EM-CRM and DA-CRM. According to Proposition \ref{prop:eq_mis_sur}, if $\bxi$ and $\alpha$ were estimated jointly under the likelihood \eqref{eq:titecrm}, and if the same specification of $f_{T \mid Z, Y}(t \mid z, Y = 1, \bxi)$ was used, the three methods would yield identical inference. This connection shows the underlying similarities of the three designs, although these designs are proposed independently with different statistical models.

\item[TITE-BOIN]
The TITE-BOIN design \citep{yuan2018time} is a TITE extension of the BOIN design.  
To maintain the transparent and simple decision rules in BOIN, TITE-BOIN uses single imputation, substituting $\bY_{\mis}$ with their expected values $\hat{\by}_{\mis}$. Specifically, $\hat{y}_{\mis, i} = \text{E}(Y_{\mis, i} = 1 \mid Z_i = z_i, T_i > v_i, \bp, \bxi) = \Pr(Y_{\mis, i} = 1 \mid Z_i = z_i, T_i > v_i, \bp, \bxi)$, given in \eqref{eq:dist_missing}. A uniform distribution is assumed for $[T \mid Z, Y = 1]$, thus $\rho(v_i \mid z_i, \bxi)  \equiv v_i / W$.
The imputation involves the unknown parameter $\bp$, for which an estimate based on an approximation procedure is plugged in.
Finally, the decision rule of BOIN is applied to the imputed dataset $(\by_{\obs}, \hat{\by}_{\mis})$.
We note that another way of extending the BOIN design is to consider the BOIN hypothesis test (Equation \ref{eq:boin_hypothesis}) directly under the likelihood \eqref{eq:likelihood_tite}, although this would be more complicated than the single imputation approach.

\item[TITE-keyboard]
The TITE-keyboard design \citep{lin2018time} is a TITE extension of the keyboard design.
Similar to the keyboard design, TITE-keyboard considers non-overlapping and equal-lengthed sub-intervals (keys) of the $[0, 1]$ interval
(from left to right): $\{ I_{\E_1}, \ldots, I_{\E_{K_1}} \}$, $I_{\textS}$, and $\{ I_{\textD_1}, \ldots, I_{\textD_{K_2}} \}$. 
Here, $I_{\textS} = [p^* - \epsilon_1, p^* + \epsilon_2]$ is referred to as the target key and is the only sub-interval that contains $p^*$.
By default, TITE-keyboard assumes independent $\Beta(1, 1)$ priors on $p_z$'s, $\pi_0(p_z) = \bone_{[0, 1]}$.
Suppose the current dose is $d$. With a model for $[T \mid Z, Y = 1]$ and an additional prior $\pi_0(\bxi)$, the posterior $\pi (p_d, \bxi \mid \HH) \propto \pi_0 (p_d) \pi_0(\bxi) L( \bp, \bxi \mid \HH )$, with $L( \bp, \bxi \mid \HH )$ given by \eqref{eq:likelihood_tite}.
Let $k^* = \argmax_k \Pr(p_d \in I_k \mid \HH)$.
The dose-assignment decision for the next patient follows the keyboard design. That is, to escalate, stay or de-escalate, if $k^*$ belongs to $\{ \E_1, \ldots, \E_{K_1} \}$, equals S or belongs to $\{ \textD_1, \ldots, \textD_{K_2} \}$, respectively.

\item[TITE-TPI]
 \cite{lin2018time} proposed a TITE extension of the mTPI design.  We note that the same extension can be applied to mTPI-2, and we refer to this extension as the TITE-TPI design.
Consider a partition of the $[0, 1]$ interval into sub-intervals (from left to right): $\{ I_{\E_0}, \ldots, I_{\E_{K_1}} \}$, $I_{\textS}$, and $\{ I_{\textD_0}, \ldots, I_{\textD_{K_2}} \}$. Except for the two boundary intervals, all these sub-intervals have the same length, and $I_{\textS} = [p^* - \epsilon_1, p^* + \epsilon_2]$ is the only sub-interval that contains $p^*$.
Next, let $d$ denote the current dose level, and let model $\{ \mathcal{M}_d = k \}$ represent $\{ p_d \in I_k \}$, $k = \E_0, \ldots, \E_{K_1}, \textS, \textD_0, \ldots, \textD_{K_2}$.
We consider the same hierarchical prior models for $\mathcal{M}_d$ and $p_d$ as in Equation \eqref{eq:mtpi2}.
With a model for $[T \mid Z, Y = 1]$ and an additional prior $\pi_0(\bxi)$, the posterior $\pi (\mathcal{M}_d, p_d, \bxi \mid \HH) \propto \pi_0 (\mathcal{M}_d) \pi_0 (p_d \mid \mathcal{M}_d) \pi_0(\bxi) L( \bp, \bxi \mid \HH )$.
The dose-assignment decision for TITE-TPI follows the mTPI-2 design. That is, to escalate, stay or de-escalate, if $\arg \max_{k} \Pr(\mathcal{M}_d = k  \mid \HH )$ belongs to $\{ \E_0, \ldots, \E_{K_1} \}$, equals S or belongs to $\{ \textD_0, \ldots, \textD_{K_2} \}$, respectively.
When the prior probabilities of $\{ \mathcal{M}_d = k \}$ are the same for all $k$'s, TITE-TPI results in the same inference as TITE-keyboard, although the prior probabilities need not be the same.

\end{description}


\subsubsection{New TITE Designs}
\label{sec:tite_new}

\begin{description}[leftmargin = 0pt, topsep=2mm, listparindent=\parindent, itemsep=3mm]

\item[TITE-SPM]
We propose the TITE-SPM design as a TITE extension of the SPM design.
Recall that the SPM directly models the the location of the MTD $\gamma$, $\gamma \in \{1, \ldots, J\}$.
The same hierarchical models for $\gamma$ and $\bp$, denoted by $\pi_0(\gamma)$ and $\pi_0(\bp \mid \gamma)$, can be used in TITE-SPM as in Equation \eqref{eq:spm}.
With a model for $[T \mid Z, Y = 1]$ and an additional prior $\pi_0(\bxi)$, the posterior $\pi(\gamma, \bp, \bxi \mid \HH) \propto \pi_0(\gamma) \pi_0(\bp \mid \gamma)  \pi_0(\bxi) L( \bp, \bxi \mid \HH )$.
The dose-assignment decision for TITE-SPM follows the SPM design. That is, to assign the dose $\hat{\gamma} = \argmax_{\gamma} \pi(\gamma \mid \HH)$ to the next patient, subject to some safety restrictions.

\item[TITE-i3]
We propose the TITE-i3 design as a TITE extension of the i3+3 design. Recall that the i3+3 design makes dose-finding decisions by comparing $n_d / N_d$ with the boundaries of $I_{\textS}$.
For TITE-i3, we replace $n_d$ in the i3+3 design with $N_d \hat{p}_d$, where $\hat{p}_d$ is the MLE under the likelihood \eqref{eq:likelihood_tite}. 
To maintain the simple algorithmic rules in i3+3, we model $[T \mid Z, Y = 1]$ with a uniform distribution, and the MLE is easy to solve.

\end{description}

According to how a design makes inference about $\bp$ and translates such inference to a dose-finding decision, we can categorize the TITE designs in the same way as the complete-data designs (Section \ref{sec:complete_data_design}).
For example, TITE-CRM is a point-based and curve-based TITE design, and TITE-BOIN,  TITE-keyboard, TITE-TPI and TITE-i3 are interval-based and curve-free TITE designs.

\subsection{POD Designs}
\label{sec:pod_design}

Taking one step further of the TITE designs, one can  directly make inference on possible dose-finding decisions when some DLT outcomes are pending (Figure \ref{fig:illustration}). This leads to a new class of POD (probability-of-decision) designs. We discuss the details next.

As mentioned in Section \ref{sec:complete_data_design}, the dose-assignment decision for any complete-data design can be written as a deterministic function of the previous (complete) DLT outcomes $\by$ and dose assignments $\bz$, denoted by $\A^*(\by, \bz)$.
In other words, the dose level $\A^*(\by, \bz) \in \{ 1, \ldots, J \}$ will be used to treat the next patient.
For example, when $p^* = 0.2$, for CRM with default prior hyperparameters, $\A_{\text{CRM}}^*[(0, 0, 0, 0, 0, 1), (1, 1, 1 , 2, 2, 2)] = 2$; for mTPI-2 with $\epsilon_1 = \epsilon_2 = 0.05$, $\A_{\text{mTPI-2}}^*[(0, 0,$ $0, 0, 0, 1), (1, 1, 1 , 2, 2, 2)] = 1$.

In the presence of pending outcomes, let $A = \A^*[(\by_{\obs}, \bY_{\mis}), \bz]$ denote the dose-assignment decision. Since $\bY_{\mis}$ is a vector of latent variables and $A$ is a function of $\bY_{\mis}$, $A$ is essentially a random variable.
Under the Bayesian paradigm, the posterior distribution of $A$ is given by
\begin{align}
\Pr(A = a \mid \HH) = \sum_{\by_{\mis}: \A^*[(\by_{\obs}, \by_{\mis}), \bz] = a} \Pr(\bY_{\mis} = \by_{\mis} \mid \HH).
\label{eq:post_decision}
\end{align}
Here,
\begin{align}
\Pr(\bY_{\mis} = \by_{\mis} \mid \HH) =  \int_{\bxi} \int_{\bp} \Pr(\bY_{\mis} = \by_{\mis} \mid \HH, \bp, \bxi) \pi(\bp, \bxi \mid \HH) \dd \bp \dd \bxi
\label{eq:pred_ymis}
\end{align}
is the posterior predictive distribution of $\bY_{\mis}$, and  $\Pr(\bY_{\mis} = \by_{\mis} \mid \HH, \bp, \bxi)$ is given in \eqref{eq:dist_missing}.
From a frequentist perspective, instead of marginalizing over the posterior distribution of $\bp$ and $\bxi$, one could plug in the MLE of $\bp$ and $\bxi$.
The probability \eqref{eq:post_decision} is referred to as the POD, which accounts for the variability in the missing data and directly reflects the confidence of every possible decision. 
The dose assignment for the next patient can be guided by the POD. 
For example, one may make the decision with the highest POD, $a^* = \argmax_a \Pr(A = a \mid \HH)$. 


Similar to ``probability of success'' (POS), POD is based on the
posterior predictive distribution $ P(\bY_{\mis} \mid
\HH) $ and describes the uncertainty of making a dosing decision due
to the uncertainty associated with $\bY_{\mis}$.  The use of the posterior
predictive distribution, rather than the posterior distribution
$\pi(\bp \mid \HH)$, highlights the main difference between POD and
TITE designs. This difference is analogous to that between the uses of posterior probabilities and posterior predictive probabilities for interim decision making \citep{lee2008predictive, saville2014utility}. 
Specifically, using predictive distributions, POD designs incorporate information about the number of patients with pending outcomes into decision considerations.
As an illustration, consider the hypothetical scenario shown in Figure \ref{fig:POD_vs_TITE}. 
When a new patient $(4 + r_d)$ arrives, $(3 + r_d)$ patients have been treated at dose $d$.
Among these patients, patients 1--3 have finished DLT assessment with 2 non-DLTs and 1 DLT; the other patients were treated just prior to the arrival of patient $(4 + r_d)$ thus having zero follow-up times.
Then, the value of $r_d$ does not affect the inference on $\bp$.
This is because patients with zero follow-up times receive zero weights in the likelihood function of $\bp$ (see Equations
\ref{eq:rho} and \ref{eq:likelihood_tite}, and note that $\rho(0 \mid d, \bm \xi) = 0$).
For example, assuming a $\Beta(1, 1)$ prior on $p_d$, the posterior distribution $p_d \mid \HH \sim \Beta(2, 3)$ for any $r_d = 0, 1, \ldots.$
Therefore, using a TITE design that is based only on the posterior distribution of $\bp$, the decision for the new patient $(4 + r_d)$ will be the same for any $r_d$. In other words, the pending outcomes with zero follow-up times play no role in the decision rule of TITE designs.
In contrast, POD designs utilize additional information from the number of pending outcomes (Equation \ref{eq:post_decision}).
Different values of $r_d$ lead to different probabilities of decisions despite the same inference on $\bp$.
For example, suppose the target DLT rate is $p^* = 0.3$, and let $\A^*$ be the dose-assignment decision of mTPI-2.
When $r_d = 0$, the complete-data decision of mTPI-2 (stay) receives 100\% of the posterior mass. 
When $r_d = 1$,  with a $\Beta(1, 1)$ prior on $p_d$, $\Pr(Y_{\mis, 4} = 1 \mid \HH) = \E(p_d \mid \HH) = 0.4$ according to Equations \eqref{eq:dist_missing} and \eqref{eq:pred_ymis},  and the posterior probabilities of stay (when $Y_{\mis, 4} = 0$) and de-escalation (when $Y_{\mis, 4} = 1$) are 0.6 and 0.4, respectively.
Similarly, for other values of $r_d$, the probabilities of decisions vary accordingly.  See Figure \ref{fig:POD_vs_TITE} for an illustration.

\begin{figure}[h!]
\begin{center}
\includegraphics[width = 0.98\textwidth]{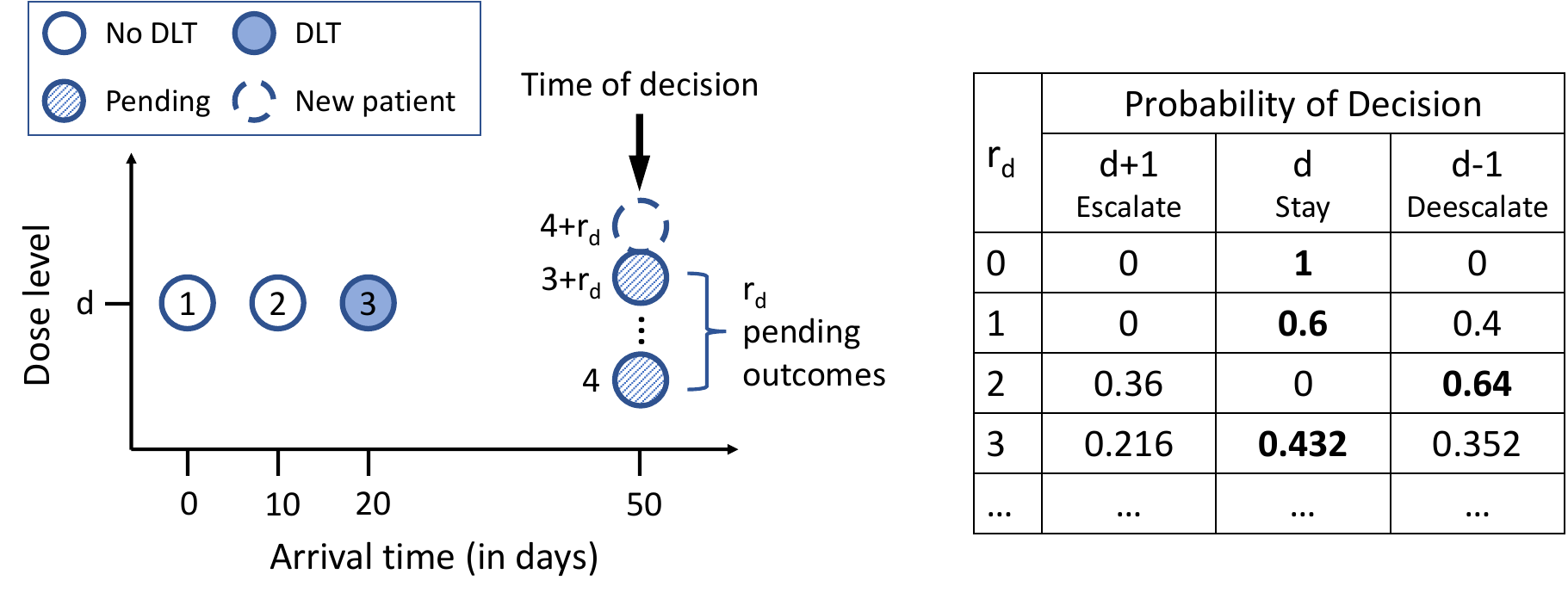} 
\end{center}
\caption{A hypothetical scenario illustrating the difference between TITE and POD designs. 
When a new patient $(4 + r_d)$ arrives, $(3 + r_d)$ patients have been treated at dose $d$, among which 2 are non-DLTs, 1 is DLT, and $r_d$ are pending with zero follow-up times.
Assuming a $\Beta(1, 1)$ prior on $p_d$. By the time the new patient arrives, the posterior distribution of $p_d$ is $\Beta(2, 3)$ for any $r_d$.
However, using POD designs, different values of $r_d$ lead to different predictive probabilities of decisions, shown in the right panel of the figure. }
\label{fig:POD_vs_TITE}
\end{figure}

\subsubsection{Existing and New POD designs}
\label{sec:pod_example}

\begin{description}[leftmargin = 0pt, topsep=2mm, listparindent=\parindent, itemsep=3mm]
\item[POD-TPI]
We illustrate the POD designs through reviewing the POD-TPI design \citep{zhou2019pod}.
POD-TPI assumes independent $\Beta(1, 1)$ priors on $p_z$'s. 
With a model for $[T \mid Z, Y = 1]$ and an additional prior $\pi_0(\bxi)$, the posterior $\pi (\bp, \bxi \mid \HH) \propto \pi_0 (\bp) \pi_0(\bxi) L( \bp, \bxi \mid \HH )$. By default, POD-TPI models $[T \mid Z, Y = 1]$ with a piecewise uniform distribution.
The posterior predictive distribution of $\bY_{\mis}$ is then computed. Finally, suppose the current dose is $d$, and let $A = \A^*[(\by_{\obs}, \bY_{\mis}), \bz] \in \{ d-1, d, d+1 \}$, where $\A^*$ is the decision function of mTPI-2.
The PODs $\Pr(A = a \mid \HH)$ of possible decision $a$'s are calculated, and the decision $a^*$ with the highest POD is executed, subject to additional safety restrictions.
\end{description}

Apparently, the decision function $\A^*$ can be based on any complete-data design. 
The model for $[T \mid Z, Y = 1]$ and priors on $\bp$ and $\bxi$ can also be adjusted if desired.
In this way, we obtain a new class of POD designs, such as POD-CRM, POD-BOIN, POD-keyboard, POD-SPM and POD-i3. As an example, we present POD-CRM below. 

\begin{description}[leftmargin = 0pt, topsep=2mm, listparindent=\parindent, itemsep=3mm]
\item[POD-CRM]
We propose the POD-CRM design as a POD extension of the CRM design. Similar to the CRM, we consider a parametric dose-toxicity curve $p_z = \phi(z, \alpha)$. Based on the likelihood \eqref{eq:titecrm}, the posterior distribution of $\alpha$ and $\bxi$ is given by $\pi(\alpha, \bxi \mid \HH) \propto \pi_0(\alpha) \pi_0(\bxi) L ( \alpha, \bxi \mid \HH )$. The posterior distribution of $\pi(\bp, \bxi \mid \HH)$ can then be computed, and the posterior predictive distribution of $\bY_{\mis}$ can be obtained. Finally, let $\A^*$ denote the decision function of CRM, and let $A = \A^*[(\by_{\obs}, \bY_{\mis}), \bz] \in \{ 1, \ldots, J \}$. The PODs of possible decisions can be calculated based on Equation \eqref{eq:post_decision}, and the decision with the highest POD is executed, subject to additional safety restrictions.
\end{description}

We can categorize the POD designs according to the corresponding complete-data designs $\A^*$. For example, POD-TPI is an interval-based and curve-free POD design.

\section{Design Properties}
\label{sec:property}

In this section, we study large- and finite-sample properties of the aforementioned time-to-event designs, with an emphasis on interval-based and curve-free designs.


\subsection{Large-Sample Convergence Properties}
\label{sec:convergence}

Dose-finding studies are usually carried out with relatively small sample sizes (10 to 50 subjects). Still, as noted in \cite{oron2011dose}, large-sample convergence properties should be viewed as a necessary quality criterion for dose-finding designs. 
In general, the large-sample properties for a particular complete-data design should also hold for its time-to-event version, as long as the DLT assessment window $W$ and the patient accrual rate are both finite. Intuitively, at time $\tau$, all patients enrolled before $(\tau - W)$ have finished their DLT assessments, and only the patients enrolled within $(\tau - W, \tau]$ can have pending outcomes. As $\tau \rightarrow \infty$, the number of complete outcomes goes to infinity too, and the number of pending outcomes is finite with probability one, making the contribution of the pending outcomes negligible in the likelihood \eqref{eq:likelihood_tite}.

In what follows, we present some general large-sample results for interval-based and curve-free time-to-event designs.
First, the following lemma establishes the consistency of the posterior distribution and MLE of $p_z$ in a time-to-event setting when (1) the number of patients treated by dose $z$, $N_z$, goes to infinity, and (2) the number of pending outcomes at dose $z$, $r_z$, is small compared to $N_z$. 

\begin{lemma}[consistency]
\label{lem:consistency}
Suppose the true distribution of the DLT outcome is $\Pr(Y_i = 1 \mid Z_i = z) = p_{0z}$, and $r_z = o(N_z)$. 
(1) Let $\C_{\varepsilon} = \{ p_z : | p_z - p_{0z} | < \varepsilon \}$. Let $\pi_0(p_z)$ be a prior distribution for $p_z$ such that $\pi_0(p_z \in \C_{\varepsilon}) > 0$ for every $\varepsilon > 0$, and the likelihood of $p_z$ is as in \eqref{eq:likelihood_tite}.
Then, for every $\varepsilon > 0$, the posterior distribution $\pi(p_z \in \C_{\varepsilon} \mid \HH) \rightarrow 1$ almost surely as $N_z \rightarrow \infty$.
(2) The maximum likelihood estimator $\hat{p}_z \rightarrow p_{0z}$  almost surely as $N_z \rightarrow \infty$.
\end{lemma}

The proof is given in Appendix \ref{supp:sec:convergence}.
As a consequence of Lemma \ref{lem:consistency}, we have the following convergence theorem for interval-based and curve-free TITE designs.


\begin{thm}[convergence]
\label{thm:convergence}
Suppose the conditions in Lemma \ref{lem:consistency} are met. If there is a dose $d^*$ satisfying $p_{0d^*} \in (p^* - \epsilon_1, p^* + \epsilon_2)$, and $d^*$ is also the only dose such that $p_{0d^*} \in [p^* - \epsilon_1, p^* + \epsilon_2]$, then dose allocations in interval-based and curve-free TITE designs converge almost surely to $d^*$.
\end{thm}

The convergence theorem has  been established in \cite{lin2018time} with a heuristic proof according to the authors. We provide a full-scale mathematical elaboration in our proof, given in Appendix \ref{supp:sec:convergence}, 
which may provide more theoretical insight for interested readers.
When the condition about $d^*$ is violated, the proof's logic immediately leads to the following results (see also \citealp{oron2011dose}).

\begin{corollary}
\label{corollary:convergence}
(1) If no dose level $d^*$ satisfies $p_{0d^*} \in [p^* - \epsilon_1, p^* + \epsilon_2]$, and $p^* \in [p_{01}, p_{0J}]$, an interval-based and curve-free TITE design will eventually oscillate almost surely between the two doses whose true DLT probabilities straddle the target interval. (2) If there are multiple doses whose true DLT probabilities lie within the interval $(p^* - \epsilon_1, p^* + \epsilon_2)$, an interval-based and curve-free TITE design will converge almost surely to one of these dose levels. However, convergence to the dose level closest to $p^*$  is not guaranteed.
\end{corollary}

Next, the following lemma establishes the consistency of the dose-finding decisions in interval-based and curve-free POD designs.

\begin{lemma}[consistency]
\label{lem:consistency_POD}
Suppose $d$ is the current dose, which is neither the lowest dose nor the highest dose.
Suppose $\A^*$ is the dose decision function of an  interval-based and curve-free complete-data design, and the conditions in Lemma \ref{lem:consistency} are met.
(1) If $p_{0d} \in (p^* - \epsilon_1, p^* + \epsilon_2)$, then $\exists N_{0d} > 0$, when $N_d > N_{0d}$, $\Pr(A = d \mid \HH) = 1$ almost surely.
 (2) If $p_{0d} < p^* - \epsilon_1$, then $\exists N_{0d} > 0$, when $N_d > N_{0d}$, $\Pr(A = d+1 \mid \HH) = 1$ almost surely.
(3) If $p_{0d} > p^* + \epsilon_2$, then $\exists N_{0d} > 0$, when $N_d > N_{0d}$, $\Pr(A = d-1 \mid \HH) = 1$ almost surely.
\end{lemma}

See Appendix \ref{supp:sec:convergence} 
for the proof.
As a result, the convergence or oscillation results of dose allocations (Theorem \ref{thm:convergence} and Corollary \ref{corollary:convergence}) also hold for interval-based and curve-free POD designs.

For point-based designs or curve-based designs,
the consistency and convergence results require additional assumptions.
We direct the readers to \cite{cheung1999sequential} for an example under the TITE-CRM setting.

\subsection{Coherence Principles}
\label{sec:coherence}

The coherence principles are another quality criterion for dose-finding designs motivated by ethical concerns in trial conduct.
\cite{cheung2005coherence} introduced a coherence condition for time-to-event designs, which states that a time-to-event design should not de-escalate from time $\tau$ to $\tau+ \tau'$ if no toxicity occurs during $[\tau, \tau+ \tau')$, and  it should not escalate from time $\tau$ to $\tau+\tau'$ if a toxicity occurs within $[\tau, \tau+ \tau')$ (for $\tau' \rightarrow 0^+$).
The formal definition is given below.

\begin{definition}[\citealp{cheung2005coherence}]
\label{def:coherence}
A time-to-event design $\A$ is \emph{coherent} if 
 (1) for any $\tau, \tau' > 0$,
\begin{align*}
\text{Pr}_{\A} \big\{ \A[\HH(\tau + \tau')]  < \A[\HH(\tau)] \mid \tY_i(\tau + \tau') - \tY_i(\tau) = 0 \text{ for all $i$} \big\} = 0;
\end{align*}
and (2)  for any $\tau > 0$,
\begin{align*}
\lim_{\tau' \rightarrow 0^+} \text{Pr}_{\A} \big\{ \A[\HH(\tau + \tau')]  > \A[\HH(\tau)] \mid \tY_i(\tau + \tau') - \tY_i(\tau) = 1 \text{ for some $i$} \big\} = 0.
\end{align*}
\end{definition}

\cite{cheung2005coherence} showed that the TITE-CRM design is coherent if the weight $\rho(v_i \mid z_i, \bxi)$ is continuous and nondecreasing in $v_i$, which is automatically satisfied under the proposed construction of $\rho$ (see Equation \ref{eq:rho}).
In contrast, interval-based and curve-free designs only use observations at the current dose to make dose-finding decisions thus may be incoherent in the sense of Definition \ref{def:coherence}. For example, 
consider target DLT rate $p^* = 0.2$.
Assume for two adjacent patients, the sequences of dose assignments $\bz =  (2, 1)$ and DLT outcomes $\by = (1, 0)$. 
Using the BOIN or TITE-BOIN design with default hyperparameters, the 3rd patient is assigned to dose level 2. Suppose by the time the 4th patient is enrolled, the 3rd patient has finished DLT assessment with no event. However, since the empirical DLT rate at dose 2 is 0.5, the 4th patient would be assigned to dose 1. In other words, no toxicity occurs after the enrollment of patient 3, but the dose level de-escalates from 2 to 1, which is incoherent.   See Figure \ref{fig:coherence}. 
This is because information at different dose levels is used to make the dose assignments for patients 3 and 4.

\begin{figure}[h!]
\begin{center}
\includegraphics[width = 0.72\textwidth]{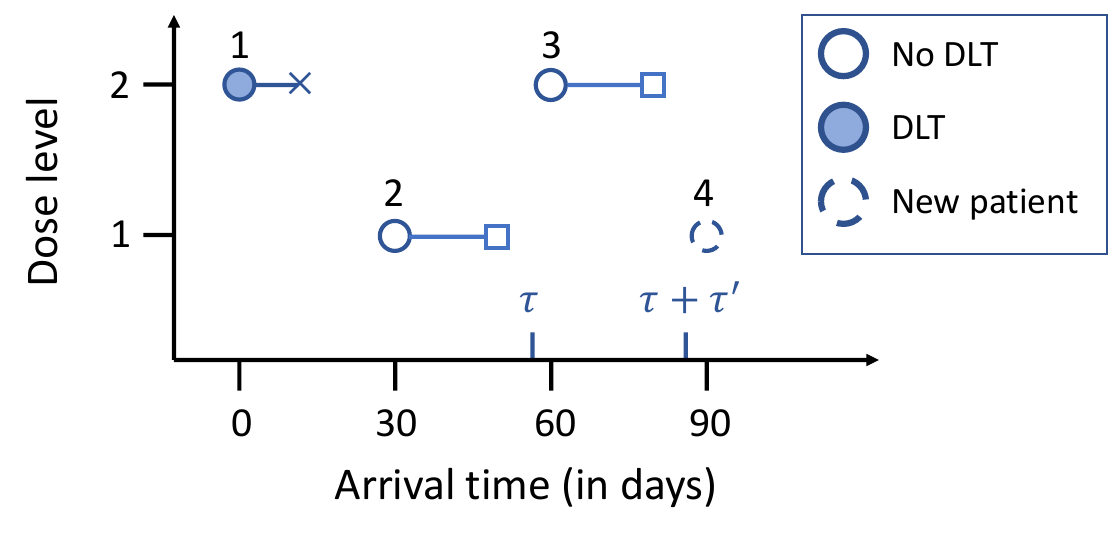} 
\end{center}
\caption{ Example of an incoherent de-escalation made by an interval-based, curve-free design according to Definition 5.1. No toxicity event occurs between times $\tau$ and $\tau + \tau'$, but the dose-assignment decision changes from $\A[\HH(\tau)] = 2$ to $\A[\HH(\tau + \tau')] = 1$. }
\label{fig:coherence}
\end{figure}

To avoid incoherent dose-finding decisions for interval-based and curve-free designs, one may impose coherence as an ad-hoc safety rule.
On the other hand, one may still think such decisions are reasonable and consider alternative coherence conditions for interval-based and curve-free designs, such as the condition given below.
\begin{definition}[Interval coherence]
\label{supp:def:coherence_interval}
An interval-based and curve-free time-to-event design is interval coherent if (1) for any $\tau, \tau' > 0$, if the currently-administrated doses just prior to $\tau$ and $\tau + \tau'$ are the same (denoted by $d$), then
\begin{multline*}
\text{Pr}_{\A} \big\{ \A[\HH(\tau + \tau')]  < \A[\HH(\tau)] \mid \\ 
\tilde{Y}_i(\tau + \tau') - \tilde{Y}_i(\tau) = 0 \text{ for all $i$ s.t. $Z_i = d$} \big\} = 0;
\end{multline*}
and (2) for any $\tau > 0$, suppose the currently-administrated doses just prior to $\tau$ is $d$, then
\begin{multline*}
\lim_{\tau' \rightarrow 0^+} \text{Pr}_{\A} \big\{ \A[\HH(\tau + \tau')]  > \A[\HH(\tau)] \mid \\
 \tY_i(\tau + \tau') - \tY_i(\tau) = 1 \text{ for some $i$ s.t. $Z_i = d$} \big\} = 0,
\end{multline*}
if the  currently-administrated doses just prior to $\tau + \tau'$ is also $d$.
\end{definition}

 According to Definition \ref{supp:def:coherence_interval}, the decision at time $\tau + \tau'$ in Figure \ref{fig:coherence} will not be considered incoherent because of the change of the currently-administrated dose from $\tau$ to $\tau + \tau'$. 
In Appendix \ref{supp:sec:coherence}, 
we show an example that an interval-based and curve-free TITE design is interval coherent in the sense of Definition \ref{supp:def:coherence_interval}.

\cite{liu2015bayesian} defined another coherence condition for dose-finding designs, which states that 
a dose-finding design is \emph{long-term memory coherent} if
 it does not de-escalate (or escalate) when the observed toxicity rate in the accumulative cohorts at the current dose is lower (or higher) than the target toxicity rate.
In other words, suppose the current dose is $d$, then a design is long-term memory coherent if it does not de-escalate (or escalate) when $n_{d} / (n_d + m_d) < p^*$ (or $> p^*$).
Under this definition, time-to-event designs may be incoherent in escalation because the pending outcomes may contribute additional evidence to counteract the toxic outcomes. If we think such an escalation is reasonable, we may define that an escalation is incoherent only if $n_{d} / (n_d + m_d + r_d)  > p^*$.
Alternatively, as in \cite{lin2018time}, we may also assign each pending outcome a weight $\rho$ and calculate an adjusted toxicity rate $\tilde{p}_d = n_{d} / [n_d + m_d + \sum_{i = 1}^N \rho_i \bone(z_i = d, b_i = 0)]$. For example, $\rho_i = v_i / W$. De-escalation is incoherent if $\tilde{p}_d < p^*$, and escalation is incoherent if $\tilde{p}_d > p^*$. However, the specification of the weight can be arbitrary.

Lastly, \cite{cheung2005coherence} noted that ad-hoc rules such as those in Section \ref{sec:practical} may affect the coherence of a design. For example, suppose that a cohort-based enrollment rule is imposed, and a design enrolls six patients to dose 1 as a group. Further, assume that the target DLT rate is $p^* = 0.3$, patients 1--5 have no toxicity events, and patient 6 has DLT. 
Then, if the design recommends escalation to dose 2 right after the toxicity event of patient 6, this is incoherent. In such situations, our recommendation is to use the coherence (or interval coherence) principle to assess the original, unmodified version of a design to ensure that at least,  the design makes sense before implementing the ad-hoc rules.

\subsection{Overdosing Decisions and Incompatible Decisions}
\label{sec:incompatible}

To better understand the decision rules in TITE and POD designs, we introduce the concept of overdosing decisions and incompatible decisions. 
A design's frequency of making such decisions measures the safety of this design.

\paragraph{Overdosing decisions.}
We call a dose-finding decision $\A(\HH)$ an \emph{overdosing decision} if $\A(\HH) > d^*$, where $d^*$ denotes the MTD.
Similarly, we call $\A(\HH) < d^*$ an underdosing decision.
For example, consider a trial with 2 doses,  target DLT rate $p^* = 0.2$, and true DLT probabilities $(p_1, p_2) = (0.2, 0.5)$. Then, $d^* = 1$ is the MTD, and any decision that allocates a patient to dose 2 is an overdosing decision.
Since the true DLT probabilities are unknown in practice, we cannot check whether a decision is an overdosing/underdosing decision in real-world trials.

The decision rules in TITE designs can be viewed as minimizing various  loss functions associated with overdosing/underdosing decisions.
For example, in the TITE-SPM design (Section \ref{sec:tite_new}), recall that $\gamma$ represents the location of the MTD.
Let $\ell(\gamma, \hat{\gamma} )$ denote the loss of allocating the new patient to dose $\hat{\gamma}$  if $\gamma$ is the true MTD, and consider the 0-1 loss $\ell(\gamma, \hat{\gamma} ) = \bone(\gamma \neq \hat{\gamma})$.
That is, there is a loss for an overdosing/underdosing decision.
Then, the decision of allocating the new patient to $\hat{\gamma} = \arg \max_{\gamma} \pi(\gamma  \mid \HH )$ minimizes the posterior expected loss.

Overdosing decisions apply to both complete-data and time-to-event designs and are inevitable due to random sampling.
As in Figure \ref{fig:safety_concern}(a),
suppose when the 4th patient is enrolled, patients 1--3 have finished DLT assessment (assuming $W = 28$ days) with no event. Then, the 4th patient might be assigned to dose 2, which is actually overly toxic.
One may justify such a decision by sampling error.


\begin{figure}[t]
\center
\begin{tabular}{cc}
\includegraphics[width = 0.49\textwidth]{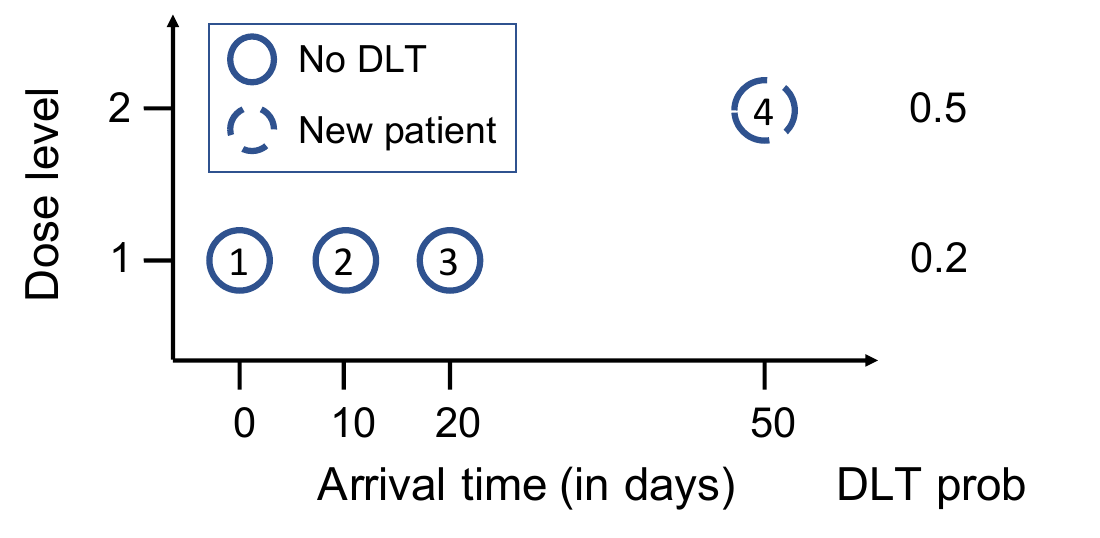} 
&
\includegraphics[width = 0.49\textwidth]{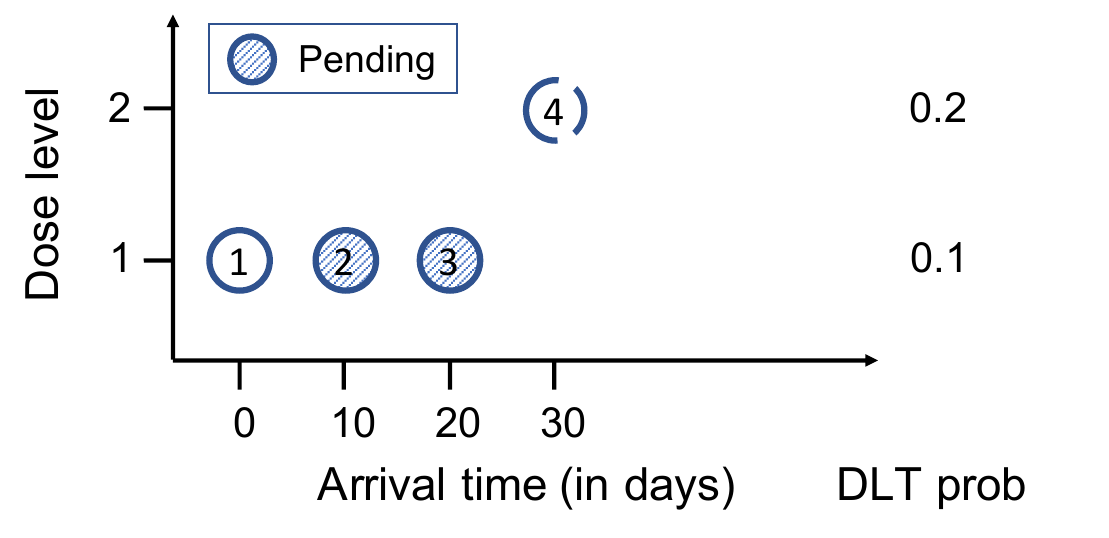} \\
(a) Overdosing decision & (b) Incompatible decision
\end{tabular}
\caption{Examples for (a) an overdosing allocation and (b) an incompatible decision (suppose patient 2 or 3 later experiences DLT within the assessment window).}
\label{fig:safety_concern}
\end{figure}

\paragraph{Incompatible decisions.}
We say that a time-to-event dose-finding decision $\A[\HH(\tau)]$ is \emph{incompatible} with a complete-data decision $\A^*[\HH_{N(\tau)}^*]$ if $\A[\HH(\tau)] \neq \A^*[\HH_{N(\tau)}^*]$. Here, $\HH(\tau) = \{ (\tilde{Y}_i(\tau) , $ $V_i(\tau), $ $Z_i) : i \leq N(\tau) \}$ represents the available time-to-event information at time $\tau$, and $\HH_{N(\tau)}^* = \{ (Y_i, Z_i) :  i \leq N(\tau) \}$ represents the complete toxicity information for the first $N(\tau)$ patients that would have been observed if these patients had completed their follow-up.
For example, as in Figure \ref{fig:safety_concern}(b), consider a trial with 2 doses,  target DLT rate $p^* = 0.2$, true DLT probabilities $(p_1, p_2) = (0.1, 0.2)$ and DLT assessment window $W = 28$ days.
Suppose when the 4th patient arrives, patient 1 have finished DLT assessment with no event, and patients 2 and 3 are still being followed without definitive outcomes. Using a time-to-event design, the 4th patient might be treated at dose 2. However, it is possible that patient 2 or 3 could later experience DLT, making the dose escalation for patient 4 incompatible with a decision that uses complete data of patients 1--3 if the decision making was conducted after patients 2 and 3 completed their follow-up.
In practice, we can check whether a time-to-event decision at time $\tau$ is incompatible with a complete-data decision after all patients enrolled before $\tau$ have finished their DLT assessment.

The decision rules in POD designs can be viewed as minimizing a loss function associated with incompatible decisions.
For example, in the POD-TPI design (Section \ref{sec:pod_example}), recall that $A$ denotes the random mTPI-2 decision in the presence of pending outcomes.
Let $\ell(A, \hat{A} )$ denote the loss of making decision $\hat{A}$ if $A$ is the true complete-data decision, and consider the 0-1 loss $\ell(A, \hat{A} ) = \bone(A \neq \hat{A})$.
Then, the decision $a^* = \argmax_a \Pr(A = a \mid \HH)$ minimizes the posterior expected loss.

Incompatible decisions only apply to time-to-event designs and can be avoided by following patients for the full length of the assessment window.
We note that an incompatible decision is not necessarily an overdosing/underdosing decision.
In the Figure \ref{fig:safety_concern}(b) example, the incompatible decision actually allocates patient 4 to the MTD. 
Still, since the true DLT probabilities are unknown in practice, such a decision cannot be justified based on the observed data, and the safety review boards should express concerns regarding the decision.
Incompatible and overly aggressive decisions ($\A[\HH(\tau)] > \A^*[\HH_{N(\tau)}^*]$) are a major concern for drug companies and regulatory agencies to use and approve time-to-event designs.
Nevertheless, we next show that such decisions may be eliminated in POD designs through a suspension rule.

\section{Practical Considerations}
\label{sec:practical}

\subsection{Safety Rules}
\label{sec:safety_rule}

In addition to statistical modeling, safety rules play an important role in dose-finding designs.  For example, when a dose is deemed overly toxic, future dose assignment to this dose or higher doses should be prohibited due to ethical concerns. If the lowest dose is too toxic, the trial should be terminated and redesigned using lower doses.
From a Bayesian perspective, toxicity can be quantified using posterior probability. Similar to \cite{ji2010modified} and \cite{yuan2016bayesian}, we consider the following safety rules.
\begin{description}
[noitemsep,leftmargin=*]
\item[Safety Rule 1 (Dose Exclusion)] At any moment in the trial, if  $n_z + m_z \geq 3$ and $\Pr(p_z > p^* \mid \data) > \nu$ for a pre-specified threshold $\nu$ close to 1, exclude dose $z$ and higher doses from the trial;
\item[Safety Rule 2 (Early Termination)] At any moment in the trial, if $n_1 + m_1 \geq 3$ and $\Pr(p_1 > p^* \mid \data) > \nu$ for a pre-specified threshold $\nu$ close to 1, terminate the trial due to excessive toxicity.
\end{description}
From a frequentist perspective,
a dose $z$ can be considered overly toxic if the lower one-sided $\nu$ confidence interval of $p_z$ does not cover $p^*$.

It is possible that a dose is considered overly toxic when some toxicity outcomes at this dose are still pending. In this case, we allow this dose and upper doses to be re-opened if the pending outcomes turn out to be safe and suggest this dose is no longer overly toxic.
If the lowest dose is considered overly toxic with some pending outcomes, temporarily suspend the trial. If later pending outcomes are observed and suggest the lowest dose is no longer overly toxic, we resume the trial; otherwise, the trial is permanently terminated.
There is a positive probability that a dose is excluded even if it is actually safe, or the trial is terminated early even when the lowest dose is safe. This is the type I error associated with the safety rules.

\subsection{Enrollment and Suspension}
\label{sec:suspension}
Many existing designs (e.g., 3+3, mTPI, BOIN, TITE-BOIN, TITE-TPI, and POD-TPI) employ a cohort-based enrollment, in which more than one patient (and typically three patients) are enrolled as a cohort. 
The patients in the same cohort are always treated by the same dose.
Cohort-based enrollment is especially helpful for
complete-data designs, as the
trial can take exceedingly long to complete if each patient has to wait for the completion of the DLT assessment for the previous patient.
Cohort-based enrollment can also help avoid overly fast dose escalation and can be more convenient for trial administration.
For most existing designs, a common cohort size is 2 to 4.


For time-to-event designs, the dose-assignment decision for a new patient may be uncertain and risky if the toxicity outcomes of many previous patients have not yet been observed.
For example, three patients have been treated at the lowest dose and have been followed up for a while, but none of them have completed the DLT assessment. In this case, it might be too conservative to treat the fourth patient at the lowest dose as it might be subtherapeutic, but it is also too risky to treat the fourth patient at a higher dose as no safe outcome has been observed. Therefore, it may be more sensible to temporarily suspend the trial until at least one outcome at the current dose has been observed.
Reflecting this situation, a possible suspension rule is the following.
\begin{description}
[noitemsep,leftmargin=*]
\item[Suspension Rule 1] Suppose the current dose is $d$. If $n_d + m_d = 0$, suspend enrollment.
\end{description}

Other types of suspension rules have been considered in the literature. For example,  \cite{yuan2018time} and \cite{guo2019rtpi} proposed suspension rules based on the number of pending outcomes.
\begin{description}
[noitemsep,leftmargin=*]
\item[Suspension Rule 2 (Fixed Suspension Rule)] Suppose the current dose is $d$. If $r_d > C$ for a pre-specified threshold $C$, suspend enrollment.
\end{description}
Here, $C$ can be a fixed number (e.g. $C = 3$, \citealp{guo2019rtpi}) or a portion of the total number of patients at the current dose (e.g. $C = N_d / 2$, \citealp{yuan2018time}).
Alternatively, as in \cite{zhou2019pod}, the POD in Equation \eqref{eq:post_decision} can be used to calibrate the trade-off between the reduction of trial suspension and the reduction of patient risk.
\begin{description}
[noitemsep,leftmargin=*]
\item[Suspension Rule 3 (Probability Suspension Rule)] Let $a^*$ denote the dose-assignment decision made by a time-to-event design in the presence of pending outcomes. If $\Pr(A < a^* \mid \HH) > q_{a^*}$ for a pre-specified threshold $q_{a^*}$, suspend enrollment.
\end{description}
 Here, $a^*$ can be based on inference about $\bp$ or $A$ for TITE or POD designs, respectively. 
For example, for POD-TPI, $a^* = \argmax_a \Pr(A = a \mid \HH)$.
If the posterior probability that a more conservative decision than $a^*$ should be made, $\Pr(A < a^* \mid \HH)$, is higher than some threshold $q_{a^*}$,  enrollment is suspended.
The Probability Suspension Rule guarantees that in the presence of pending outcomes, the chance of making an aggressive incompatible decision is upper bounded.
The threshold $q_{a^*}$ can be different for different $a^*$.

We note that the speed of the trial is solely determined by the suspension rule, i.e., how many times an eligible patient is turned away. If no pending outcome is allowed for making the dose-assignment decision ($C = 0$), the time-to-event design becomes a design using complete outcomes. On the other hand, if a trial never suspends, all eligible patients are enrolled and treated immediately, and the trial achieves its optimal speed.
The Probability Suspension Rule allows a meaningful calibration between trial speed and safety of the design. This will be clear later in our numerical examples.

 In summary, safety and suspension rules reduce the chance of overdosing and incompatible decisions, respectively. The rules discussed in this section are applicable to both TITE and POD designs. Although the decisions in POD designs are based on POD, the posterior distribution of $\bp$ is also obtained (see Figure \ref{fig:illustration}), which can be used to calculate the posterior probability of overdosing required by the safety rules. Similarly, although the decisions in TITE designs are based on the inference about $\bp$, one could calculate the POD and use the Probability Suspension Rule to decide when to suspend a trial. 

After patient enrollment is terminated and all DLT assessments are finished, the trial completes, and the next step is to recommend an MTD. We summarize several methods for MTD selection in Appendix \ref{supp:sec:sel_mtd}.

\section{Simulation Studies}
\label{sec:simulation}

\subsection{Simulation Set-up} 
We conduct simulation studies to compare the operating characteristics of some TITE and POD designs that we have discussed in the previous sections.
We consider 18 dose-toxicity scenarios with target DLT probability $p^* = 0.2$ or $0.3$ and $J = 7$ dose levels. The 18 representative scenarios consist of the 16 scenarios reported in \cite{yuan2018time} and 2 additional scenarios that cover various MTD locations. The scenarios are summarized in Appendix Table \ref{tbl:simu_DLT_prob}.

We assume the DLT assessment window $W = 28$ days and use a maximum sample size of $N^* = 36$ patients.
The time-to-toxicity for each patient, $T_i$, is generated from a Weibull distribution with shape $\zeta_{1z}$ and scale $\zeta_{2z}$, given the patient is treated by dose $z$. That is,
\begin{align*}
(T_i \mid Z_i = z) \sim \Weibull(\zeta_{1z}, \zeta_{2z}).
\end{align*}
The parameters $\zeta_{1z}$ and $\zeta_{2z}$ are chosen
such that $\Pr(T_i \leq W \mid Z_i = z) = p_z$ and $\Pr(T_i \leq W^* \mid Z_i = z) = (1 - q)p_z$. 
Here $ W^* \in (0, W)$ and $q \in (0, 1)$ are arbitrary numbers, meaning if a toxic outcome occurs within the assessment window, with probability $q$ it occurs within the interval $(W^*, W]$. 
The time between the accrual of two consecutive patients is generated from an exponential distribution with rate $\delta$, which means the average wait time between two consecutive patients is $1 / \delta$.
We consider the following three settings with different time-to-toxicity and accrual rate profiles:
\begin{description}
[noitemsep,nolistsep,leftmargin=*]
\item[Setting 1 (default)] inter-arrival time is 10 days, and 50\% of DLTs occur in the second half of the assessment window. This corresponds to $\delta = 0.1$, $W^* = W/2$ and $q = 0.5$;
\item[Setting 2 (more late-onset DLTs)] inter-arrival time is 10 days, and 80\% of DLTs occur in the last quarter of the assessment window. This corresponds to $\delta = 0.1$, $W^* = 3W/4$ and $q = 0.8$;
\item[Setting 3 (faster accrual)] inter-arrival time is 5 days, and 50\% of DLTs occur in the second half of the assessment window. This corresponds to $\delta = 0.2$, $W^* = W/2$ and $q = 0.5$.
\end{description}
We do not consider another setting with a longer DLT assessment window, as it would be equivalent to faster accrual after rescaling the time.

\subsection{Design Specifications}

We consider the TITE-TPI, TITE-CRM, TITE-BOIN (Section \ref{sec:tite_design}) and POD-TPI (Section \ref{sec:pod_design}) designs as examples of different types of time-to-event designs.
In addition, we include the mTPI-2, CRM, and BOIN designs in the comparison as representative complete-data designs.
We require that all designs start from the lowest dose. For TITE-CRM and CRM, dose skipping is prohibited in escalations (i.e., dosage could not increase by more than one level at a time), but there is no restriction on de-escalations \citep{goodman1995some}.
Full details about the design specifications are reported in Appendix \ref{supp:sec:design_spec}.

For a fair comparison, in the default setting, we use the same time-to-toxicity model and safety, enrollment, and suspension rules for all time-to-event designs.
These include (1) a uniform distribution model for $f_{T \mid Z, Y}(t \mid z, Y = 1)$, (2) Safety Rules 1 and 2 in Section \ref{sec:safety_rule} with $\nu = 0.95$, (3) cohort-based enrollment in groups of 3, and (4) Suspension Rules 1 and 2 in Section \ref{sec:suspension}, which suspend enrollment if the number of pending patients at the current dose $r_d > N_d / 2$ \citep{yuan2018time}.
For the complete-data designs, we (1) impose Safety Rules 1 and 2 with $\nu = 0.95$, (2) enroll patients in cohorts of 3, and (3) suspended enrollment whenever some toxicity outcomes are  pending (i.e., $r_d > 0$). During trial suspension, the available patients are turned away.

For additional comparison and sensitivity analysis, we consider the following variations of some designs:
\begin{enumerate}[noitemsep,nolistsep,leftmargin=*]
\item mTPI-2-LA: in mTPI-2, if the pending outcomes have no impact on final decisions, the decisions are executed right away even if $r_d > 0$.  This is called a look-ahead design and is a special case of POD-TPI in which no suspension is necessary if the probability of a decision is 1.
\item TITE-TPI-NS and POD-TPI-NS: discarding Suspension Rule 2 in the default setting of TITE-TPI and POD-TPI.
We still keep Suspension Rule 1 to avoid hasty decisions when no toxicity outcomes have ever been observed at the current dose.
\item TITE-TPI-PSR and POD-TPI-PSR:
replacing Suspension Rule 2 with Suspension Rule 3 (Probability Suspension Rule) in the default setting of TITE-TPI and POD-TPI.
We use $q_{a^*} \equiv 0.25$ for all $a^*$.
Recall that $a^*$ denotes the dose-assignment decision made by a time-to-event design in the presence of pending outcomes. If the probability of $a^*$ being overly aggressive is greater than 0.25, enrollment is suspended. This reduces the chance of making aggressive incompatible decisions.
\item TITE-TPI-PSR2 and POD-TPI-PSR2: replacing Suspension Rule 2 with Suspension Rule 3 (Probability Suspension Rule) and setting $q_{a^*} \equiv 0$ for all $a^*$, such that enrollment is suspended if there is a positive probability that $a^*$ is overly aggressive.
This completely eliminates the chance of making aggressive incompatible  decisions.
\end{enumerate}

\subsection{Performance Metrics}

The performance of a design is evaluated based on the following metrics.
\begin{enumerate}[noitemsep,nolistsep,leftmargin=*]
\item \textbf{Allocation}, including 
(1.1) percentage of patients treated at the MTD (percentage of correct allocation, PCA);
(1.2) percentage of patients treated at doses above the MTD (percentage of overdosing allocation, POA); 
(1.3) percentage of patients treated at doses below the MTD (percentage of underdosing allocation, PUA).
\item \textbf{Selection}, including 
(2.1) percentage of correct selection (PCS) of the MTD;
(2.2) percentage of dose selection above the MTD (percentage of overdosing selection, POS); and 
(2.3) percentage of dose selection below the MTD (percentage of underdosing selection, PUS). 
\item \textbf{Risk}, which is measured by the relative frequencies of incompatible dose assignment decisions (see Section \ref{sec:incompatible}).
Recall that an incompatible decision refers to a decision that is different from what would have been made if complete outcomes were observed.
There are six types of incompatible decisions, including decisions that (3.1) should be de-escalation based on complete data but are stay (DS); (3.2) should be de-escalation but are escalation (DE); (3.3) should be stay but are escalation (SE); (3.4) should be stay based on complete data but are de-escalation (SD); (3.5) should be escalation based on complete data but are de-escalation (ED); and (3.6) should be escalation based on complete data but are stay (ES).
We are particularly concerned about the first three types of incompatible decisions, DS, DE and SE, which are overly aggressive.
\item \textbf{Speed}, which is measured by the average trial duration (Dur).
\end{enumerate}
Metrics (1.1)--(1.3) assess a design's performance in assigning patients to appropriate doses.
Metrics (2.1)--(2.3) evaluate a design's performance in selecting the right dose as the MTD.
Metrics (3.1)--(3.6) are about the risk associated with allowing patient enrollment in the presence of pending outcomes. In particular, the incompatible decisions of the time-to-event designs are obtained by comparing with their complete-data counterparts. For example, TITE-CRM is compared with CRM, and POD-TPI is compared with mTPI-2.
Metric (4) is about the speed of the trial.

\subsection{Simulation Results}
For each dose-toxicity scenario in Appendix Table \ref{tbl:simu_DLT_prob}, 
we simulate 1,000 trials with each design.
Table \ref{tbl:simu_result} summarizes the results by averaging over the 18 scenarios.
The scenario-specific results under Setting 1 are reported in Appendix \ref{supp:sec:ss_result}. 
The performances of the designs are generally similar if averaged over scenarios, although they may have a large variation across different scenarios. 
The comparison among mTPI-2, TITE-TPI and POD-TPI illustrates the different behaviors of various extensions of the same complete-data design.

\begin{center}
\begin{longtable}{l@{\extracolsep{4pt}}c@{\extracolsep{4pt}}c@{\extracolsep{4pt}}c@{\extracolsep{4pt}}c@{\extracolsep{4pt}}c@{\extracolsep{4pt}}c@{\extracolsep{4pt}}c@{\extracolsep{4pt}}c@{\extracolsep{4pt}}c@{\extracolsep{4pt}}c@{\extracolsep{4pt}}c@{\extracolsep{4pt}}c@{\extracolsep{4pt}}c}
\caption{Summary of simulation results under 18 dose-toxicity scenarios and 3 time-to-toxicity and accrual rate settings. 
Values shown are averages over simulated trials and scenarios.
Seven designs, mTPI-2, TITE-TPI, POD-TPI, CRM, TITE-CRM, BOIN and TITE-BOIN, are compared. 
Variations of some designs are considered for sensitivity analysis.
PCA, POA, PUA, PCS, POS, and PUS are in \%, DS, DE, SE, SD, ED, and ES are in $1/10^{3}$, and Dur is in days.} \label{tbl:simu_result} \\
\hline \hline
\multicolumn{1}{c}{\multirow{2}{*}{Design}} & \multicolumn{3}{c}{Allocation} & \multicolumn{3}{c}{Selection} & \multicolumn{6}{c}{Risk} & Speed \\
\cline{2-4} \cline{5-7} \cline{8-13} \cline{14-14}
& PCA & POA & PUA & PCS & POS & PUS & DS & DE & SE & SD & ED & ES & Dur \\ \hline 
\endfirsthead

\multicolumn{14}{c}%
{ \tablename\ \thetable{} -- continued from previous page} \\
\hline 
\multicolumn{1}{c}{\multirow{2}{*}{Design}} & \multicolumn{3}{c}{Allocation} & \multicolumn{3}{c}{Selection} & \multicolumn{6}{c}{Risk} & Speed \\
\cline{2-4} \cline{5-7} \cline{8-13} \cline{14-14}
& PCA & POA & PUA & PCS & POS & PUS & DS & DE & SE & SD & ED & ES & Dur \\ \hline
\endhead

\hline \multicolumn{14}{r}{{Continued on next page}} \\
\endfoot

\hline \hline
\endlastfoot

\multicolumn{14}{c}{Setting 1} \\ 
mTPI-2 & 34.3 & 22.1 & 43.7 & 51.3 & 16.4 & 32.3 & 0.0 & 0.0 & 0.0 & 0.0 & 0.0 & 0.0 & 633 \\
TITE-TPI & 32.2 & 21.5 & 46.3 & 50.5 & 16.6 & 32.9 & 12.7 & 3.3 & 21.9 & 55.4 & 1.5 & 29.7 & 436 \\
POD-TPI & 32.5 & 24.5 & 43.0 & 51.4 & 18.4 & 30.2 & 20.5 & 4.1 & 23.0 & 23.8 & 0.1 & 10.3 & 437 \\
CRM & 36.3 & 25.4 & 38.4 & 55.5 & 30.1 & 14.4 & 0.0 & 0.0 & 0.0 & 0.0 & 0.0 & 0.0 & 632 \\
TITE-CRM & 34.7 & 27.1 & 38.2 & 54.5 & 30.3 & 15.2 & 27.3 & 0.4 & 19.8 & 11.8 & 0.0 & 30.7 & 439 \\
BOIN & 34.2 & 22.1 & 43.7 & 54.1 & 22.8 & 23.1 & 0.0 & 0.0 & 0.0 & 0.0 & 0.0 & 0.0 & 634 \\
TITE-BOIN & 32.4 & 21.0 & 46.5 & 54.0 & 22.0 & 24.0 & 10.4 & 2.8 & 18.7 & 55.8 & 0.2 & 29.2 & 435 \\
\hdashline
mTPI-2-LA & 33.6 & 22.6 & 43.8 & 50.2 & 16.4 & 33.3 & 0.0 & 0.0 & 0.0 & 0.0 & 0.0 & 0.0 & 560 \\
TITE-TPI-NS & 32.1 & 21.8 & 46.1 & 51.5 & 16.0 & 32.4 & 13.1 & 7.7 & 27.4 & 65.0 & 1.9 & 34.0 & 400 \\
TITE-TPI-PSR & 31.8 & 19.7 & 48.5 & 50.7 & 14.4 & 34.8 & 6.5 & 3.7 & 14.1 & 69.5 & 2.0 & 40.2 & 422 \\
TITE-TPI-PSR2 & 31.9 & 19.2 & 48.9 & 51.0 & 15.1 & 33.9 & 0.0 & 0.0 & 0.0 & 76.0 & 2.0 & 18.5 & 527 \\
POD-TPI-NS & 32.2 & 24.4 & 43.4 & 51.3 & 18.2 & 30.5 & 20.9 & 8.2 & 29.3 & 39.2 & 1.7 & 12.6 & 401 \\
POD-TPI-PSR & 33.0 & 22.5 & 44.5 & 51.1 & 16.9 & 32.0 & 9.8 & 4.1 & 14.9 & 37.4 & 1.0 & 11.6 & 433 \\
POD-TPI-PSR2 & 32.7 & 21.3 & 46.0 & 50.0 & 16.7 & 33.4 & 0.0 & 0.0 & 0.0 & 37.5 & 1.1 & 4.3 & 541 \\
\hline
\multicolumn{14}{c}{Setting 2} \\ 
mTPI-2 & 34.5 & 21.9 & 43.7 & 51.2 & 15.8 & 33.0 & 0.0 & 0.0 & 0.0 & 0.0 & 0.0 & 0.0 & 646 \\
TITE-TPI & 32.0 & 24.5 & 43.5 & 51.8 & 18.2 & 30.0 & 25.1 & 8.4 & 36.0 & 45.8 & 1.5 & 25.6 & 445 \\
POD-TPI & 31.8 & 27.2 & 41.0 & 52.0 & 19.5 & 28.5 & 38.4 & 10.3 & 38.3 & 20.0 & 0.0 & 9.1 & 449 \\
CRM & 35.6 & 25.9 & 38.6 & 53.8 & 30.9 & 15.3 & 0.0 & 0.0 & 0.0 & 0.0 & 0.0 & 0.0 & 644 \\
TITE-CRM & 34.2 & 29.0 & 36.8 & 55.0 & 29.9 & 15.1 & 49.8 & 1.6 & 34.7 & 8.6 & 0.0 & 27.6 & 449 \\
BOIN & 34.4 & 22.0 & 43.7 & 53.9 & 22.6 & 23.5 & 0.0 & 0.0 & 0.0 & 0.0 & 0.0 & 0.0 & 645 \\
TITE-BOIN & 32.4 & 23.7 & 43.9 & 53.8 & 23.6 & 22.6 & 22.1 & 7.1 & 32.2 & 44.8 & 0.2 & 24.5 & 445 \\
\hdashline
mTPI-2-LA & 33.8 & 22.8 & 43.4 & 51.4 & 16.3 & 32.3 & 0.0 & 0.0 & 0.0 & 0.0 & 0.0 & 0.0 & 581 \\
TITE-TPI-NS & 31.4 & 26.1 & 42.5 & 51.0 & 19.0 & 30.0 & 27.7 & 18.6 & 46.1 & 50.5 & 1.7 & 28.7 & 412 \\
TITE-TPI-PSR & 31.7 & 23.5 & 44.8 & 51.4 & 16.8 & 31.7 & 16.7 & 10.9 & 30.6 & 55.0 & 1.8 & 34.6 & 435 \\
TITE-TPI-PSR2 & 31.7 & 20.3 & 48.0 & 50.4 & 15.5 & 34.1 & 0.0 & 0.0 & 0.0 & 61.1 & 2.1 & 17.3 & 548 \\
POD-TPI-NS & 30.7 & 29.8 & 39.5 & 51.4 & 20.9 & 27.7 & 41.8 & 21.9 & 50.0 & 27.8 & 1.5 & 10.4 & 414 \\
POD-TPI-PSR & 31.7 & 26.5 & 41.8 & 50.5 & 19.5 & 29.9 & 26.6 & 12.3 & 33.5 & 26.9 & 1.2 & 10.3 & 446 \\
POD-TPI-PSR2 & 32.8 & 21.9 & 45.3 & 51.0 & 16.0 & 33.1 & 0.0 & 0.0 & 0.0 & 26.3 & 1.0 & 4.6 & 560 \\
\hline
 \multicolumn{14}{c}{Setting 3} \\
mTPI-2 & 34.3 & 21.9 & 43.8 & 51.4 & 15.8 & 32.8 & 0.0 & 0.0 & 0.0 & 0.0 & 0.0 & 0.0 & 472 \\
TITE-TPI & 31.4 & 21.1 & 47.5 & 50.5 & 15.8 & 33.7 & 15.8 & 3.6 & 25.5 & 64.1 & 4.0 & 44.6 & 290 \\
POD-TPI & 31.7 & 23.3 & 45.0 & 51.5 & 17.4 & 31.0 & 18.8 & 4.0 & 25.8 & 42.8 & 0.1 & 26.4 & 294 \\
CRM & 35.6 & 25.2 & 39.2 & 54.6 & 30.0 & 15.4 & 0.0 & 0.0 & 0.0 & 0.0 & 0.0 & 0.0 & 471 \\
TITE-CRM & 34.3 & 27.1 & 38.6 & 54.5 & 30.9 & 14.6 & 35.0 & 0.6 & 20.7 & 12.7 & 0.0 & 49.7 & 298 \\
BOIN & 34.2 & 21.9 & 43.9 & 54.1 & 22.4 & 23.4 & 0.0 & 0.0 & 0.0 & 0.0 & 0.0 & 0.0 & 471 \\
TITE-BOIN & 31.7 & 20.6 & 47.7 & 53.5 & 21.5 & 25.0 & 13.1 & 3.2 & 21.5 & 63.2 & 0.9 & 47.8 & 291 \\
\hdashline
mTPI-2-LA & 33.8 & 22.8 & 43.3 & 51.4 & 16.1 & 32.5 & 0.0 & 0.0 & 0.0 & 0.0 & 0.0 & 0.0 & 408 \\
TITE-TPI-NS & 30.9 & 20.8 & 48.3 & 50.3 & 16.0 & 33.7 & 14.7 & 6.6 & 29.3 & 79.2 & 7.1 & 51.1 & 257 \\
TITE-TPI-PSR & 30.1 & 19.2 & 50.7 & 49.5 & 14.5 & 36.0 & 8.0 & 3.3 & 14.9 & 83.8 & 7.6 & 63.7 & 272 \\
TITE-TPI-PSR2 & 30.7 & 18.3 & 51.0 & 50.2 & 14.3 & 35.4 & 0.0 & 0.0 & 0.0 & 94.9 & 6.5 & 28.5 & 370 \\
POD-TPI-NS & 31.3 & 23.0 & 45.8 & 51.6 & 17.2 & 31.2 & 17.2 & 7.8 & 30.6 & 66.0 & 3.5 & 31.7 & 258 \\
POD-TPI-PSR & 31.2 & 21.5 & 47.3 & 50.5 & 16.2 & 33.3 & 11.4 & 3.6 & 16.4 & 65.1 & 2.9 & 29.7 & 279 \\
POD-TPI-PSR2 & 31.6 & 20.6 & 47.8 & 50.0 & 16.0 & 34.0 & 0.0 & 0.0 & 0.0 & 64.5 & 0.3 & 12.5 & 379 \\
\end{longtable}
\end{center}

On average, the trial duration is shortened by about 180--200 days using the time-to-event designs (under the default setting with Suspension Rules 1 and 2). 
Even when compared with a look-ahead design (mTPI-2-LA), the time-to-event designs still lead to a time saving of around 110--130 days.
This is a major benefit for drug development.
The trial durations under TITE-TPI,  POD-TPI,  TITE-CRM and TITE-BOIN are highly similar, because the same suspension rule is imposed,  and the speed of a design is solely determined by the suspension rule. 
In the presence of pending outcomes, time-to-event designs may result in incompatible assignments. For example, in Table \ref{tbl:simu_result}, the DS, DE, SE, SD, ED, and ES of the time-to-event designs are non-zero.
The relative frequency of incompatible decisions (sum over the six types) made by POD-TPI is lower than that of TITE-TPI, because the decision rules in POD designs are directly associated with incompatible decisions.
Although the time-to-event designs make incompatible decisions, their PCA are not much lower, as the incompatible decisions may both be aggressive and conservative, and an incompatible decision is not necessarily an overdosing/underdosing decision.
The PCS of the time-to-event designs are comparable to the corresponding complete-data designs. 
This is not surprising, because we always use complete outcomes to make the final selection of MTD.
With more late-onset toxicities (Setting 2) or faster patient accrual (Setting 3), the performances of the time-to-event designs are slightly decreased.
In particular, the frequencies of incompatible decisions under Settings 2 and 3 are generally increased compared to the results under Setting 1.
Lastly, there is always a trade-off among the different performance metrics. For example, TITE-CRM has about the highest PCA and PCS, but it also has the highest POA and POS due to the more aggressive decision rules and MTD selection.

When the pending outcomes do not affect final decisions, mTPI-2-LA allows the decisions to be executed immediately. In this way, mTPI-2-LA achieves some time saving compared to the standard mTPI-2 without making incompatible decisions.
Essentially, mTPI-2-LA is a special case of POD-TPI, in which no suspension is necessary if the probability of a decision is 1.
While the pending outcomes in mTPI-2-LA have no impact on final decisions, they may contribute to the safety rules.
For example, while a pending toxicity does not alter a de-escalation decision, it may inform whether a dose exclusion or an early termination is needed for safety.
As a result, the performance metrics of mTPI-2-LA and mTPI-2 are slightly different.
Abandoning Suspension Rule 2, TITE-TPI-NS and POD-TPI-NS attain the fastest speed at the cost of lower PCA and more frequent incompatible decisions.
Using Suspension Rule 3 with a threshold of $q_{a^*} = 0.25$, TITE-TPI-PSR and POD-TPI-PSR have similar time saving compared to the default TITE-TPI and POD-TPI while making less aggressive incompatible decisions.
This is a major advantage of the Probability Suspension Rule, which controls the risk posed to patients when decisions are made in the presence of pending outcomes.
Furthermore, the threshold in the Probability Suspension Rule can be tuned to achieve a desirable trade-off between the reduction of trial suspension and the reduction of patient risk.
For example, if no aggressive incompatible decisions are ever acceptable, one can use a stringent threshold value of $q_{a^*} = 0$.
With this choice, TITE-TPI-PSR2 and POD-TPI-PSR2 never make aggressive incompatible decisions (with zero DS, DE, and SE). 
Of course, the time saving is less with such a conservative choice, but it reflects the best one can do if one wants to eliminate the chance of making aggressive incompatible decisions.

Lastly, we perform additional sensitivity analyses with alternative maximum sample sizes ($N^* = 48$ and 60) and time-to-toxicity models (piecewise uniform, discrete hazard, and piecewise constant hazard models). The results are reported in Appendices
\ref{supp:sec:sens_sample_size} and \ref{supp:sec:t_model_spec}.

\section{Concluding Remarks}
\label{sec:discussion}

We have reviewed and summarized statistical frameworks for time-to-event modeling in dose-finding trials.
Two classes of time-to-event designs, TITE designs and POD designs, can be built upon these frameworks. We have demonstrated that existing time-to-event designs (such as TITE-CRM, TITE-BOIN, TITE-keyboard, and TITE-TPI) fall within these frameworks, and new time-to-event designs (such as TITE-SPM and POD-CRM) can be developed under these frameworks. The frameworks open the way to a deeper study of time-to-event designs.


An essential component in a time-to-event design is the modeling of time-to-toxicity data. Existing methods consider $f_{T \mid Z = z}$ as the product of $p_z$ and $f_{T \mid Z = z, Y = 1}$ and then specify $f_{T \mid Z = z, Y = 1}$ (see Equation \ref{eq:model_t}). As a result, models for $\bp$ in complete-data designs can be adopted, inference and assumptions on $\bp$ are explicit, and by assuming  that $f_{T \mid Z = z, Y = 1}$ does not depend on $\bp$, the resulting time-to-toxicity model is relatively simple.
An alternative approach is to specify $f_{T \mid Z}$ directly using, say, a proportional hazard model. Inference on $\bp$ can then be obtained by computing $p_z = \Pr(T \leq W \mid Z = z)$. The performance of the latter approach, as well as the choice between these two modeling approaches, is worth further investigation and discussion. 

We have discussed and studied theoretical properties of time-to-event designs with an emphasis on interval-based and curve-free designs. 
A future direction is to investigate more on these theoretical properties, especially for point-based or curve-based designs.


We have evaluated the operating characteristics of several designs through computer simulations.
As we have seen, there is no single design that outperforms the other designs in all aspects.
To choose one specific design to use in practice, we may run large-scale simulations and consider a loss function of the form
\begin{multline*}
\ell = -\ell_1 \text{PCA} + \ell_2 \text{POA} + \ell_3 \text{PUA} - \ell_4 \text{PCS}  + \ell_5 \text{POS} +
\ell_6 \text{PUS} + \\ 
\ell_7 \text{DS} + \ell_8 \text{DE} + \ell_9 \text{SE} +
\ell_{10} \text{SD} + \ell_{11} \text{ED} + \ell_{12} \text{ES} +
\ell_{13} \text{Dur}, 
\end{multline*}
where $\ell_1, \ldots, \ell_{13} \geq 0$. The design that minimizes the loss function can be selected as the winner. Usually, safety is the major concern. For example, the safety review boards may express concerns regarding aggressive incompatible decisions. If the probability of such decisions needs to be controlled, we recommend using the Probability Suspension Rule with a calibrated threshold $q_{a^*}$.

Lastly, the discussed frameworks can be extended to accommodate non-binary endpoints. For example, ordinal endpoints that account for multiple toxicity grades \citep{van2012dose, mu2019gboin}. Another direction for further exploration is to apply the framework to drug combination trials \citep{wages2011dose, liu2013bayesian2} or phase I/II trials that simultaneously consider toxicity and efficacy \citep{jin2014using}.

\bibliographystyle{apalike}
\bibliography{ref-TITE-TPI}

\clearpage

\appendix

\section{Modeling Time-to-Toxicity Data}
\label{supp:sec:model_tite}

We introduce several examples for the specification of $f_{T \mid Z, Y}(t \mid z, Y = 1, \bxi)$. That is, the model for the conditional distribution $[T \mid Z, Y = 1]$.
Recall that the event $Y = 1$ is equivalent to $T \leq W$. 

\subsection{Uniform Distribution} 
\label{supp:sec:model_tite_unif}

The simplest choice for the conditional distribution of $T$ is a uniform distribution,
\begin{align*}
T \mid Z, Y = 1 \sim \Unif(0, W). 
\end{align*}
This leads to the conditional pdf,
\begin{align*}
f_{T \mid Z, Y}(t \mid z, Y = 1) = 1 / W,
\end{align*}
where no additional parameter $\bxi$ is involved.
The weight function in this case is $\rho(t \mid z) = t / W$. The uniform distribution, albeit simple, has been shown to be sufficient in many cases \citep{cheung2000sequential}. It is the default choice in \cite{cheung2000sequential}, \cite{yuan2018time} and \cite{lin2018time}.

\subsection{Piecewise Uniform Distribution} 
\label{supp:sec:model_tite_pwunif}

A more flexible specification for the conditional distribution of $T$ is a piecewise uniform distribution. The interval $(0, W]$ is first partitioned into $K$ sub-intervals $\{ (h_{k-1}, h_k], k = 1, \ldots, K\}$, where $0 = h_0 < h_1 < \cdots < h_K = W$. Next, each sub-interval is assigned a weight $\omega_k$, $\sum_{k = 1}^K \omega_k = 1$. 
Conditional on $Y = 1$, $T$ falls into $(h_{k-1}, h_k]$ with probability $\omega_k$ and follows a uniform distribution within this interval.
The conditional pdf of $T$ is thus
\begin{align}
f_{T \mid Z, Y} (t \mid z, Y = 1, \bxi) = \omega_k \cdot \frac{1}{h_k - h_{k-1}}, \quad \text{for} \; h_{k-1} < t \leq h_k.
\label{eq:piecewise_uniform}
\end{align}
The weight function is
\begin{align*}
\rho(t \mid z, \bxi) = \sum_{k = 1}^K \omega_k \beta(t, k),
\end{align*}
where
\begin{align}
\beta(t, k) = 
\begin{cases} 
1, & \text{if }  t > h_k; \\
\frac{t - h_{k-1}}{h_k - h_{k-1}},  & \text{if } t \in (h_{k-1}, h_k], k = 1, \ldots, K; \\
0,  & \text{otherwise.}  
\end{cases}
\label{eq:beta_fun}
\end{align}
The parameters $\bxi = (K, h_1, \ldots, h_{K-1}, \omega_1, \ldots, \omega_K)$.
The number of sub-intervals $K$ and the interval boundaries $h_k$'s are usually fixed. For example, $K = 3$ and $h_k = k W / K$ for $k = 1, \ldots, K$, representing three sub-intervals with equal lengths. 
Alternatively, let $n$ denote the number of observed DLTs, and let $0 < t_{(1)} \leq \cdots \leq t_{(n)} \leq W$ denote the ordered observed DLT times. One can set $K = n+1$, $h_k = t_{(k)}$ for $k = 1, \ldots, n$ and $h_K = W$.
The weights of the sub-intervals $\omega_k$'s can be fixed if prior information is available or can be estimated otherwise. A Dirichlet distribution can be used as the prior for $(\omega_1, \ldots, \omega_K)$.
The piecewise uniform distribution with equal-length sub-intervals and fixed sub-interval weights is considered in \cite{yuan2018time} and \cite{lin2018time}.
The piecewise uniform distribution with empirically calibrated sub-intervals and same sub-interval weights $\omega_k = 1 / K$ is equivalent to the adaptive weighting scheme in \cite{cheung2000sequential}.

\subsection{Discrete Hazard Model} 
\label{supp:sec:model_tite_dhm}

The conditional distribution of $T$ can be constructed from a discrete hazard model.
Let $0 < t_{(1)} \leq \cdots \leq t_{(n)} \leq W$ denote the ordered observed DLT times, let $h_1 < \cdots < h_{K-1}$ represent the distinct DLT times that are strictly less than $W$, and let $h_K = W$. 
We assume $T$ can only take discrete values at the $h_k$'s given $Y = 1$.
Let $\omega_k = \Pr(T = h_k \mid T \geq h_k, Y = 1)$ represent the discrete hazard at time $h_k$, with $\omega_K = 1$. The conditional probability mass function of $T$ is
\begin{align*}
\Pr(T = h_k \mid Z, Y = 1, \bxi) = \omega_k \prod_{j = 1}^{k-1} (1 - \omega_{j}), \quad k = 1, \ldots, K,
\end{align*}
and the weight function is
\begin{align*}
\rho(t \mid z, \bxi) = \Pr(T \leq t \mid Z, Y = 1, \bxi) = 1 - \prod_{k: h_k \leq t} (1 - \omega_k).
\end{align*}
The discrete hazard model is used in \cite{yuan2011robust}.

\subsection{Piecewise Constant Hazard Model} 
\label{supp:sec:model_tite_pchm}

The conditional distribution of $T$ can also be constructed from a piecewise constant hazard model.
Again, consider a partition of $(0, W]$ into $K$ sub-intervals $\{ (h_{k-1}, h_k], k = 1, \ldots, K\}$, where $0 = h_0 < h_1 < \cdots < h_K = W$. 
Given $Y = 1$, we assume the hazard of toxicity is $\omega_k$ in the interval $(h_{k-1}, h_k]$. This leads to the conditional pdf,
\begin{align}
f_{T \mid Z, Y} (t \mid z, Y = 1, \bxi) = \omega_{k}^{\bone(t \in (h_{k-1}, h_k])} \exp \left[ - \sum_{k = 1}^{K} \omega_{k} (h_k - h_{k-1}) \beta(t, k) \right],
\label{eq:simple_hazard}
\end{align}
where $\beta(t, k)$ is the same as in Equation \eqref{eq:beta_fun}.
The weight function is
\begin{align*}
\rho(t \mid z, \bxi) = 1 - \exp \left[ - \sum_{k = 1}^{K} \omega_{k} \beta(t, k) \right].
\end{align*}
This model specification is used in \cite{liu2013bayesian} and \cite{jin2014using}. We note that although this specification can facilitate inference, it has a potential pitfall: $\int_0^W f_{T \mid Z, Y} (t \mid z, Y = 1, \bxi) \dd t = \rho(W \mid z, \bxi)  < 1$, which means \eqref{eq:simple_hazard} is not a proper pdf in $(0, W]$.

\subsection{Rescaled Beta Distribution} 

Another possible specification for the conditional distribution of $T$ is a rescaled beta distribution. Specifically,
\begin{align*}
T/W  \mid Z, Y = 1, \xi_1, \xi_2  \sim \Beta(\xi_1, \xi_2),
\end{align*}
and
\begin{align*}
f_{T \mid Z, Y}(t \mid z, Y = 1, \bxi) = \frac{1}{\text{B}(\xi_1, \xi_2)} \cdot \frac{t^{\xi_1 - 1} (W - t)^{\xi_2 - 1}}{ W^{\xi_1 + \xi_2 - 1}},
\end{align*}
where $\text{B}(\cdot, \cdot)$ is the beta function. The rescaled beta distribution is considered in \cite{lin2018time}. Gamma distribution priors can be put on $\xi_1$ and $\xi_2$.

In all the examples above, the conditional distribution $f_{T \mid Z, Y}(t \mid z, Y = 1, \bxi)$ does not depend on the dose, which implies $T$ and $Z$ are conditionally independent given $Y = 1$.
This allows pooling of time-to-event information across doses.
If desired, it is easy to let $f_{T \mid Z, Y}(t \mid z, Y = 1, \bxi)$ vary across doses. For example, in Equation \eqref{eq:piecewise_uniform} the parameters $\omega_k$'s can be changed to dose-specific parameters $\omega_{zk}$'s.
Since the time-to-event data are usually sparse in dose-finding trials, information borrowing is recommended for estimating dose-specific parameters using, e.g., hierarchical priors.

\clearpage

\section{Inference on the Toxicity Probabilities}
\label{supp:sec:inference_p}

\subsection{Inference Based on the Survival Likelihood}
\label{supp:sec:inference_sur}

With the survival likelihood \eqref{eq:likelihood_tite}, one can  proceed with inference on $\bp$. 
In this section, we propose some general strategies to conduct such inference.

\paragraph{Independent Metropolis-Hastings Algorithm.}

From a Bayesian perspective, prior distributions $\pi_0(\bp)$ and $\pi_0(\bxi)$ are specified for $\bp$ and $\bxi$. Inference on $\bp$ is summarized in the posterior distribution, 
\begin{align*}
\pi(\bp, \bxi \mid \HH) \propto \pi_0(\bp) \pi_0(\bxi) L ( \bp, \bxi \mid \HH ).
\end{align*}
In general, the posterior is not available in closed form, and Monte Carlo simulation is needed to approximate the posterior.
When conjugate priors are used, the independent Metropolis-Hastings (IMH) algorithm (\citealt{robert2004monte}, Section 7.4) can be employed.
Let $\tilde{L}(\bp, \bxi \mid \HH)$ denote the complete case likelihood (i.e. the likelihood for the patients with complete outcomes),
\begin{align*}
\tilde{L}(\bp, \bxi \mid \HH)
= \prod_{i = 1}^N \Big[ p_{z_i}^{\bone(\ty_i = 1)}  (1 - p_{z_{i}})^{\bone(\ty_i = 0, v_i = W)} 
 f_{T \mid Z, Y}(v_i \mid z_i, Y = 1, \bxi)^{\bone(\ty_i = 1)} \Big].
\end{align*}
The complete case likelihood can be factorized as $\tilde{L} = \tilde{L}_1 (\bp \mid \HH ) \tilde{L}_2 (\bxi \mid \HH )$, where
\begin{align*}
\tilde{L}_1 (\bp \mid \HH ) &= \prod_{i = 1}^N \Big[ p_{z_i}^{\bone(\ty_i = 1)}  (1 - p_{z_{i}})^{\bone(\ty_i = 0, v_i = W)}  \Big], \text{ and} \\
 \tilde{L}_2 (\bxi \mid \HH ) &= \prod_{i = 1}^N \Big[ f_{T \mid Z, Y}(v_i \mid z_i, Y = 1, \bxi)^{\bone(\ty_i = 1)} \Big].
\end{align*}
To implement the IMH algorithm, we first randomly initialize $\bp^{(0)}$ and $\bxi^{(0)}$. At iteration $j$ ($j = 1, 2, \ldots$), we propose $\tilde{\bp}$ and $\tilde{\bxi}$ from the complete case posteriors,
\begin{align*}
\tilde{\pi}(\bp \mid \HH) \propto \pi_0(\bp) \tilde{L}_1 (\bp \mid \HH ), \;\; \text{and} \;\; \tilde{\pi}(\bxi \mid \HH) \propto \pi_0(\bxi) \tilde{L}_2 (\bxi \mid \HH ).
\end{align*}
If conjugate priors are specified for $\bp$ and $\bxi$ (e.g., a beta distribution prior for $p_z$), the complete case posteriors are available in closed form. The proposals are accepted with probability
\begin{align*}
q_{\text{acc}} \left( \tilde{\bp}, \tilde{\bxi}; \bp^{(j-1)}, \bxi^{(j-1)} \right) = 1 \wedge \prod_{i: \ty_i = 0, v_i < W} \frac{ 1 - \rho(v_i \mid z_i, \tilde{\bxi}) \tilde{p}_{z_i}  }{ 1 - \rho(v_i \mid z_i, \bxi^{(j-1)}) p_{z_i}^{(j-1)}  },
\end{align*}
and otherwise, $\bp^{(j-1)}$ and $\bxi^{(j-1)}$ are retained. Under standard regularity conditions \citep{robert2004monte}, the sequence $\{ \bp^{(j)}, j = 1, 2, \ldots \}$ has a stationary distribution $\pi(\bp \mid \HH)$.

\paragraph{Partial Derivatives of the Log-Likelihood.}

From a frequentist perspective, the MLE $\hat{\bp}$ can be used as an estimate for $\bp$. 
We can calculate $(\hat{\bp}, \hat{\bxi}) = \argmax_{\bp, \bxi} L( \bp, \bxi \mid \HH )$ by taking partial derivatives of the log-likelihood with respect to all the parameters. 
It suffices to solve the following equations,
\begin{align*}
\frac{\partial \log L} {\partial p_z} = \frac{n_z}{p_z} - \frac{m_z}{1 - p_z} - \sum_{i : \ty_i = 0, v_i < W, z_i = z} \frac{\rho(v_i \mid z_i, \bxi)}{1 - \rho(v_i \mid z_i, \bxi) p_z } = 0,
\end{align*}
and 
\begin{multline*}
\frac{\partial \log L} {\partial \xi_k} = -\sum_{i : \ty_i = 0, v_i < W} \frac{p_{z_i}}{1 - \rho(v_i \mid z_i, \bxi) p_z } \frac{\partial \rho(v_i \mid z_i, \bxi)} {\partial \xi_k} + \\
\sum_{i: \ty_i = 1} \frac{1}{f_{T \mid Z, Y}(v_i \mid z_i, Y = 1, \bxi)} \frac{\partial f_{T \mid Z, Y}(v_i \mid z_i, Y = 1, \bxi)} {\partial \xi_k} = 0.
\end{multline*}

\subsection{Equivalence of the Survival Likelihood and the Augmented Likelihood}
\label{supp:sec:eq_sur_mis}

\begin{proof}[Proof of Proposition \ref{prop:eq_mis_sur}]
For notational clarity, let $L_{\mis}(\bp, \bxi \mid \HH )$ denote the derived data likelihood by marginalizing \eqref{eq:likelihood_missing} over $\by_{\mis}$, and let $L_{\text{sur}}(\bp, \bxi \mid \HH )$ denote the survival likelihood \eqref{eq:likelihood_tite}. We have 
\begin{align*}
 L_{\mis} ( \bp, \bxi \mid \HH )
=  {}&{} \sum_{\by_{\mis} \in \{ 0, 1\}^r } L_{\mis}( \bp, \bxi, \by_{\mis} \mid \HH ) \\
= {}&{} \prod_{i : b_i = 1} \Big [ p_{z_i}^{\bone(y_i = 1)} (1 - p_{z_i})^{\bone(y_i = 0)}  f_{T \mid Z, Y}(v_i \mid z_i, Y_i = 1, \bxi)^{\bone(y_i = 1)}  \Big ] \times \\
{}&{} \qquad \prod_{i : b_i = 0} \Big \{ p_{z_i} \left[1 - \rho(v_i \mid z_i, \bxi) \right] + (1 - p_{z_i}) \Big \} \\
= {}&{} \prod_{i = 1}^N \Big \{ p_{z_i}^{\bone(y_i = 1, b_i = 1)} (1 - p_{z_i})^{\bone(y_i = 0, b_i = 1)} \times \\
{}&{} \qquad f_{T \mid Z, Y}(v_i \mid z_i, Y_i = 1, \bxi)^{\bone(y_i = 1, b_i = 1)} \left[1 - \rho(v_i \mid z_i, \bxi) p_{z_i} \right]^{\bone(b_i = 0)} \Big \} \\
= {}&{} L_{\text{sur}} ( \bp, \bxi \mid \HH ).
\end{align*}
Here, $r = \sum_{i = 1}^N \bone(b_i = 0) = \sum_{z = 1}^J r_z$ is the number of missing outcomes.
The last equation holds because $\{ \tY_i = 1\}$ is equivalent to $\{ Y_i = 1, B_i = 1\}$, and $\{ \tY_i = 0 \}$ contains $\{ Y_i = 0, B_i = 1 \}$ and $\{ B_i = 0 \}$.
\end{proof}

\subsection{Inference Based on the Augmented Likelihood}
\label{supp:sec:inference_mis}

With the augmented likelihood \eqref{eq:likelihood_missing}, one can  proceed with inference on $\bp$. 
Specifically, \eqref{eq:likelihood_missing} can be factorized into
$L( \bp, \bxi, \by_{\mis} \mid \HH )  =  L_1(\bp, \by_{\mis} \mid \HH ) L_2 (\bxi, \by_{\mis} \mid \HH)$, where
\begin{align*}
L_1(\bp, \by_{\mis} \mid \HH ) &= \prod_{i = 1}^N \Big[ p_{z_i}^{\bone(y_i = 1)}  (1 - p_{z_{i}})^{\bone(y_i = 0)}  \Big], \text{ and} \\
L_2 (\bxi, \by_{\mis} \mid \HH) &= \prod_{i = 1}^N \Big \{ f_{T \mid Z, Y}(v_i \mid z_i, Y = 1, \bxi)^{\bone(y_i = 1, b_i = 1)} \left[1 - \rho(v_i \mid z_i, \bxi) \right]^{\bone(y_i = 1, b_i = 0)} \Big \}.
\end{align*}
This factorization facilitates inference.

\paragraph{Data Augmentation.}

From a Bayesian perspective, 
the posterior distribution $\pi(\bp \mid \HH)$ can be simulated using the data augmentation (DA) method \citep{tanner1987calculation, higdon1998auxiliary}. 
We first randomly initialize $\bp^{(0)}$, $\bxi^{(0)}$ and $\by_{\mis}^{(0)}$. At iteration $j$ ($j = 1, 2, \ldots$), implement the following procedures until the stationary distribution is reached and the desired number of posterior samples is obtained.

\noindent (1) Imputation step. Sample $\by_{\mis}^{(j)} \mid \HH, \bp^{(j-1)}, \bxi^{(j-1)}$ from \eqref{eq:dist_missing};

\noindent (2) Posterior step. Sample $\bp^{(j)} \mid \HH, \by_{\mis}^{(j)}$ and $\bxi^{(j)} \mid \HH, \by_{\mis}^{(j)}$ from the corresponding posteriors, 
\begin{align*}
&\pi(\bp \mid \HH, \by_{\mis}^{(j)}) \propto \pi_0(\bp) L_1(\bp, \by_{\mis}^{(j)} \mid \HH ), \quad \text{and} \\
&\pi(\bxi \mid \HH, \by_{\mis}^{(j)}) \propto \pi_0(\bxi) L_2 (\bxi, \by_{\mis}^{(j)} \mid \HH).
\end{align*}
If conjugate priors are specified for $\bp$ and $\bxi$ (e.g., a beta distribution prior for $p_z$), the above posteriors are available in closed form.
Through this procedure, we obtain a Markov chain $\{\bp^{(j)},\bxi^{(j)}, \by_{\mis}^{(j)}, j = 1, 2, \ldots \}$, whose stationary distribution is $\pi(\bp, \bxi, \by_{\mis} \mid \HH)$ under standard regularity conditions. The sequence $\{ \bp^{(j)}, j = 1, 2, \ldots \}$ has a marginal stationary distribution $\pi(\bp \mid \HH)$.

\paragraph{Expectation Maximization.}

From a frequentist perspective, the MLE of $\bp$ can be calculated through the expectation maximization (EM) algorithm \citep{dempster1977maximum}. 
We first randomly initialize $\bp^{(0)}$, $\bxi^{(0)}$ and $\by_{\mis}^{(0)}$. At iteration $j$ ($j = 1, 2, \ldots$), implement the following procedures until the desired convergence criteria is met.

\noindent (1) Expectation step. Set $\by_{\mis}^{(j)}$ at the expected values as in \eqref{eq:dist_missing} given the current parameter estimates $\bp^{(j-1)}$ and $\bxi^{(j-1)}$;

\noindent (2) Maximization step. Set $\bp^{(j)}$ and $\bxi^{(j)}$ at the corresponding MLEs of \eqref{eq:likelihood_missing}, using $(\by_{\obs}, \by_{\mis}^{(j)})$ as the full data. That is,
\begin{align*}
&\bp^{(j)} = \argmax_{\bp} L_1(\bp, \by_{\mis}^{(j)} \mid \HH ), \quad \text{and} \\
&\bxi^{(j)} = \argmax_{\bxi} L_2 (\bxi, \by_{\mis}^{(j)} \mid \HH).
\end{align*}
Here, $\by_{\mis}^{(j)}$ can be a fraction.
Under standard regularity conditions, the sequence $\{\bp^{(j)},$ $\bxi^{(j)},$ $j = 1, 2, \ldots \}$ converges to $\argmax_{\bp, \bxi} L(\bp, \bxi \mid \HH)$.

\clearpage

\section{Design Properties}
\label{supp:sec:property}

\subsection{Large-Sample Convergence Properties}
\label{supp:sec:convergence}

\begin{proof}[Proof of Lemma \ref{lem:consistency}]
\noindent (1) 
In the likelihood \eqref{eq:likelihood_tite},
we first treat $\bxi$ as a fixed parameter.
Conditional on $\bxi$, $p_z$ is independent of the information at the other doses, which simplifies the problem.
In the end, we will integrate out $\bxi$ and show the lemma also holds unconditional on $\bxi$.
For notation simplicity, we suppress the subscript $z$ in the following proof and only consider the patients at dose $z$.
That is, $N_z$, $n_z$, $m_z$, $r_z$, $p_z$ and $p_{0z}$ are simplified as $N$, $n$, $m$, $r$, $p$ and $p_0$. Also, $y_i$, $\ty_i$, $v_i$ and $b_i$ now only refer to the DLT outcome, current toxicity status, follow-up time and observed outcome indicator for a patient $i$ at dose $z$, $i = 1, \ldots, N_z$.
Without loss of generality, assume $0 < p_0 - \epsilon  < p_0 < p_0 + \epsilon < 1$, thus $\log (p_0 - \epsilon)$, $\log p_0$ and $\log (1 - p_0 - \epsilon)$ are finite.
The likelihood of $p$ is
\begin{align*}
L ( p \mid \HH, \bxi) &= \prod_{i = 1}^N \Big\{ p^{\bone(\ty_i = 1)}  (1 - p)^{\bone(\ty_i = 0, v_i = W)} 
\left[ 1 - \rho(v_i \mid z, \bxi) p \right]^{\bone(\ty_i = 0, v_i < W)} 
  \Big\}\\
&= \prod_{i = 1}^N \Big\{ p^{\bone(y_i = 1, b_i = 0)}  (1 - p)^{\bone(y_i = 0, b_i = 0)} 
\left[ 1 - \rho(v_i \mid z, \bxi) p \right]^{\bone(b_i = 1)} 
  \Big\}.
\end{align*}
Define 
\begin{align}
\eta_N (p; \HH, \bxi)
= \frac{1}{N} \log \frac{L ( p_0 \mid \HH, \bxi)}{L ( p \mid \HH, \bxi)}.
\label{supp:eq_eta_fn}
\end{align}
We have 
\begin{multline*}
\eta_N (p; \HH, \bxi) = \frac{1}{N} \Big \{ n (\log p_0 - \log p) + m [\log (1 - p_0) - \log (1 - p) ] + \\
\sum_{i: b_i = 1} [ \log (1 - \rho_i p_0) - \log (1 - \rho_i p) ] \Big \},
\end{multline*}
where $\rho_i = \rho(v_i \mid z, \bxi)$, and $0 \leq \rho_i \leq 1$. Thus,
\begin{align*}
\underline{\eta}_N (p; \HH) \leq \eta_N (p; \HH, \bxi) \leq \overline{\eta}_N (p; \HH),
\end{align*}
where
\begin{align*}
\underline{\eta}_N (p; \HH) &= \frac{1}{N}  \left[ n (\log p_0 - \log p) + (m + r) \log (1 - p_0) - m \log(1 - p) \right], \\
\overline{\eta}_N (p; \HH) &= \frac{1}{N}  \left[ n (\log p_0 - \log p) + m \log (1 - p_0) - (m + r) \log(1 - p) \right].
\end{align*}
By taking derivatives, we know $\underline{\eta}_N (p; \HH)$ monotonically decreases on $[0, n/(n+m))$, reaches the minimum at $n/(n+m)$ and monotonically increases on $(n/(n+m), 1]$; 
$\overline{\eta}_N (p; \HH)$ monotonically decreases on $[0, n/N)$, reaches the minimum at $n/N$ and monotonically increases on $(n/N, 1]$

Furthermore, let
\begin{align*}
\eta^* (p) = p_0 (\log p_0 - \log p) + (1 - p_0) [\log (1 - p_0) - \log (1 - p)].
\end{align*}
Similarly, $\eta^* (p)$ monotonically decreases on $[0, p_0)$, reaches the minimum at $p_0$ and monotonically increases on $(p_0, 1]$. We have $\eta^* (p_0) = 0$. Since $\eta^* (p)$ is continuous, $\exists \delta > 0$ and $0 < \epsilon_0 < \epsilon$, such that $\min \{ \eta^* (p_0 - \epsilon), \eta^* (p_0 + \epsilon) \} > 5 \delta$ and $\max \{ \eta^* (p_0 - \epsilon_0), \eta^* (p_0 + \epsilon_0) \} < \delta$.

Let $\bar{\C}_{\varepsilon}$ denote the complement of $\C_{\varepsilon}$.
Recall that $\C_{\varepsilon} = \{ p: | p - p_{0} | < \varepsilon \}$.
 Consider the posterior odds,
\begin{align}
\frac{\int_{\C_{\varepsilon}} \pi(p \mid \HH, \bxi) \dd p}{\int_{\bar{\C}_{\varepsilon}} \pi(p \mid \HH, \bxi) \dd p} = 
&\frac{\int_{\C_{\varepsilon}} \pi_0(p) L(p \mid \HH, \bxi ) \dd p}{\int_{\bar{\C}_{\varepsilon}} \pi_0(p) L(p \mid \HH, \bxi ) \dd p} \nonumber \\ 
= &\frac{\int_{\C_{\varepsilon}} \pi_0(p) \exp [-N \eta_N (p; \HH, \bxi)] \dd p}{\int_{\bar{\C}_{\varepsilon}} \pi_0(p) \exp [-N \eta_N (p; \HH, \bxi)]  \dd p}.
\label{supp:eq:post_odd}
\end{align}

According to the strong law of large numbers, and since $r = o(N)$, $\exists N_0$ (note that $N_0$ does not depend on $\bxi$), $\forall N > N_0$, 
\begin{multline*}
\max \left\{ \left| \frac{n}{n+m} - p_0 \right|, \left| \frac{n}{N} - p_0 \right|, \left| \frac{n+r}{N} - p_0 \right| \right\} < \\
\min \left\{ \epsilon_0, \frac{1}{4} \left| \frac{\delta}{\log(p_0 - \epsilon)}\right|, \frac{1}{4} \left| \frac{\delta}{\log(1 - p_0 - \epsilon)}\right| \right\},
\end{multline*}
almost surely (a.s.).
We have 
\begin{multline*}
| \underline{\eta}_N (p_0 - \epsilon; \HH) - \eta^* (p_0 - \epsilon) | \leq \left| \frac{n}{N} - p_0 \right| \cdot \left| \log p_0 \right| + \left| \frac{n}{N} - p_0 \right| \cdot \left | \log (p_0 - \epsilon) \right | + \\
\left| \frac{n}{N} - p_0 \right| \cdot \left| \log (1 - p_0) \right| + \left| \frac{n+r}{N} - p_0 \right| \cdot \left| \log (1 - p_0 + \epsilon) \right| \leq \delta.
\end{multline*}

Thus, for all $p \leq p_0 - \epsilon$,
\begin{align*}
\eta_N (p; \HH, \bxi) &\geq \underline{\eta}_N (p; \HH) \geq \underline{\eta}_N (p_0 - \epsilon; \HH) \\
&\geq \eta^* (p_0 - \epsilon) - | \underline{\eta}_N (p_0 - \epsilon; \HH) - \eta^* (p_0 - \epsilon) | > 4 \delta.
\end{align*}
Similarly, for all $p \geq p_0 + \epsilon$, we have $\eta_N (p; \HH, \bxi) > 4 \delta$.

Next, consider $p_0 - \epsilon_0 \leq p \leq p_0 + \epsilon_0$. We have 
\begin{multline*}
| \overline{\eta}_N (p_0 - \epsilon_0; \HH) - \eta^* (p_0 - \epsilon_0) | \leq \left| \frac{n}{N} - p_0 \right| \cdot \left| \log p_0 \right| + \left| \frac{n}{N} - p_0 \right| \cdot \left | \log (p_0 - \epsilon_0) \right | + \\
\left| \frac{n+r}{N} - p_0 \right| \cdot \left| \log (1 - p_0) \right| + \left| \frac{n}{N} - p_0 \right| \cdot \left| \log (1 - p_0 + \epsilon_0) \right| \leq \delta.
\end{multline*}
Thus, for all $p_0 - \epsilon_0 \leq p \leq n/N$,
\begin{align*}
\eta_N (p; \HH, \bxi) &\leq \overline{\eta}_N (p; \HH) \leq \underline{\eta}_N (p_0 - \epsilon_0; \HH) \\
&\leq \eta^* (p_0 - \epsilon_0) + | \overline{\eta}_N (p_0 - \epsilon_0; \HH) - \eta^* (p_0 - \epsilon_0) | < 2 \delta.
\end{align*}
Similarly, for all $n/N \leq p \leq p_0 + \epsilon_0$, we have $\eta_N (p; \HH, \bxi) < 2 \delta$.

As a result, in Equation \eqref{supp:eq:post_odd}, the numerator and the denominator satisfy
\begin{align*}
&\int_{\C_{\varepsilon}} \pi_0(p) \exp [-N \eta_N (p; \HH, \bxi)] \dd p \geq \exp (- 2 N \delta) \pi_0(p \in \C_{\epsilon_0}), \\
&\int_{\bar{\C}_{\varepsilon}} \pi_0(p) \exp [-N \eta_N (p; \HH, \bxi)] \dd p \leq \exp (- 4 N \delta),
\end{align*}
respectively, for all $N > N_0$ a.s. 
That is, for some $\delta > 0$ and $0 < \epsilon_0 < \epsilon$, $\exists N_0 > 0$, $\forall N  > N_0$, 
$\eqref{supp:eq:post_odd} \geq \exp (2 N \delta) \pi_0(p \in \C_{\epsilon_0})$ a.s. 

Since the numerator and the denominator of \eqref{supp:eq:post_odd} add up to 1, through simple algebra we have $\pi(p \in \C_{\varepsilon} \mid \HH, \bxi) \geq 1 - [1 + \exp (2 N \delta) \pi_0(p \in \C_{\epsilon_0}) ]^{-1}$. 
Notice that none of the terms on the right hand side depend on $\bxi$.
Thus, 
\begin{align*}
\pi(p \in \C_{\varepsilon} \mid \HH) &= \int_{\bxi} \pi(p \in \C_{\varepsilon} \mid \HH, \bxi) \pi(\bxi \mid \HH) \dd \bxi \\
&\geq \int_{\bxi} \left\{ 1 - [1 + \exp (2 N \delta) \pi_0(p \in \C_{\epsilon_0}) ]^{-1} \right\} \pi(\bxi \mid \HH) \dd \bxi \\
&= 1 - [1 + \exp (2 N \delta) \pi_0(p \in \C_{\epsilon_0}) ]^{-1}
\end{align*}
and goes to 1 a.s. as $N \rightarrow \infty$.

\noindent (2) In (1) we have already proved $\forall \epsilon > 0$, $\exists N_0$, $\forall N > N_0$, $\exists \delta > 0$ such that $\eta_N (p; \HH, \bxi) > 4 \delta$ for all $p \in \bar{\C}_{\varepsilon}$ and for all $\bxi$ a.s. Since the log-likelihood at the MLE $\hat{p}$ must be greater than or equal to the log-likelihood at any other place (including $p_0$), we have $\eta_N (\hat{p}; \HH, \hat{\bxi}) \leq 0$. As a result, $\hat{p} \in \bar{\C}_{\varepsilon}$ would cause a contradiction. Therefore,  $\forall \epsilon > 0$, $\exists N_0$, $\forall N > N_0$, $\hat{p} \in \C_{\varepsilon}$ a.s., which means $\hat{p} \rightarrow p_0$.
\end{proof}

\begin{proof}[Proof of Theorem \ref{thm:convergence}]
The proof follows \cite{oron2011dose} as a consequence of Lemma \ref{lem:consistency}.
Define the random set 
\begin{align*}
\Z = \{ z \in \{ 1, \ldots, J\}: N_z \rightarrow \infty \text{ as }  N \rightarrow \infty \},
\end{align*}
where $N_z$ is the number of patients assigned to dose $z$. Obviously, $\Z$ is nonempty and is composed of consecutive dose levels, $\Z = \{ Z_1, \ldots, Z_2 \}$. Here $Z_1, \ldots, Z_2$ are consecutively ordered integers ($Z_1 \leq Z_2$). 
For a particular dose $d^*$, the possible configurations of $\Z$ can be partitioned into three subspaces: $A = \{ Z_1 = Z_2 = d^* \}$, $B = \{ Z_1 < Z_2 \text{ and } d^* \in \Z \}$ and $C = \{ d^* \notin \Z \}$. Almost sure convergence to $d^*$ is equivalent to $\Pr(A) = 1$.

Suppose the dose $d^*$ satisfies $p_{0d^*} \in (p^* - \epsilon_1, p^* + \epsilon_2)$, and $d^*$ is also the only dose such that $p_{0d^*} \in [p^* - \epsilon_1, p^* + \epsilon_2]$.
We first prove $\Pr(C) = 0$ by contradiction.
With out loss of generality, we assume there is some specific dose $z_1  > d^*$ for which $\Pr(Z_1 = z_1) > 0$. From the theorem's condition, $p_{0 z_1} > p^* + \epsilon_2$. When the event $Z_1 = z_1$ happens, $N_{z_1} \rightarrow \infty$ as $N \rightarrow \infty$.
Thus $\forall \epsilon > 0$, for $N$ large enough, $\pi[p_{z_1} \in (p_{0 z_1} - \epsilon, p_{0 z_1} + \epsilon) \mid \HH] \rightarrow 1$, meaning $\arg \max_{j} \Pr(p_{z_1} \in I_j  \mid \HH) \in \{ \textD_1, \ldots, \textD_{K_2} \}$. Similarly, for $N$ large enough, the MLE $\hat{p}_{z_1} > p^* + \epsilon_2$. According to the transition rule of interval-based designs, the next lower dose level $z_1 - 1$ will be assigned a.s. following each allocation to $z_1$. Thus $z_1 - 1 \in \Z$, reaching a contradiction.
As a result, for all $z_1 > d^*$, $\Pr(Z_1 = z_1) = 0$.
Based on similar reasoning,  for all $z_2 < d^*$, $\Pr(Z_2 = z_2) = 0$. This means with probability 1, $Z_1 \leq d^*$ and $Z_2 \geq d^*$, i.e., $d^* \in \Z$.

Since $d^* \in \Z$, $N_{d^*} \rightarrow \infty$ as $N \rightarrow \infty$. Thus $\forall \epsilon > 0$, for $N$ large enough, $\pi[p_{d^*} \in (p_{0 d^*} - \epsilon, p_{0 d^*} + \epsilon) \mid \HH] \rightarrow 1$, meaning $\Pr(p_{d^*} \in I_{\textS}  \mid \HH ) \rightarrow 1$. Similarly, for $N$ large enough, the MLE $\hat{p}_{d^*} \in I_{\textS}$. According to the transition rule of interval-based designs, the same dose level $d^*$ will be retained a.s. following each assignment to $d^*$. Thus, there is no other level in $\Z$ with probability 1, and the dose allocation converges a.s. to $d^*$.
\end{proof}

\begin{proof}[Proof of Lemma \ref{lem:consistency_POD}]
For notation simplicity, we suppress the subscript $d$ in the following proof and only consider the patients with $Z_i = d$, as interval-based nonparametric designs (defined in Section \ref{sec:complete_data_design}) only use information at the current dose.
Suppose $\A^*(\by)$ is the dose decision function of an interval-based nonparametric complete-data design, where $\by$ denotes the vector of outcomes at the current dose only.
We have 
\begin{align*}
\Pr(A = a \mid \HH) = \sum_{\by_{\mis}: \A^*(\by_{\obs}, \by_{\mis}) = a} \Pr(\bY_{\mis} = \by_{\mis} \mid \HH).
\end{align*}
Again, $\by_{\obs}$ and $\bY_{\mis}$ now only refer to the observed outcomes and unknown outcomes for patients at the current dose.
Let $\mathcal{Y}_{\mis} = \{ 0, 1 \}^r$ denote the support of $\bY_{\mis}$, where $r$ is the number of pending patients at the current dose. 

\noindent (1) Suppose $p_0 \in (p^* - \epsilon_1, p^* + \epsilon_2)$. It suffices to show $\forall \by_{\mis} \in \mathcal{Y}_{\mis}$, $\exists N_0 > 0$,  when $N > N_0$,
$\A^*(\by_{\obs}, \by_{\mis}) = d$ a.s.
For a complete-data design with outcomes $(\by_{\obs}, \by_{\mis})$, the likelihood of $p$ is
\begin{align*}
L( p \mid \by_{\obs}, \by_{\mis}) = \prod_{i = 1}^N p^{y_i} (1-p)^{1 - y_i} = p^{n + s} (1-p)^{m+r-s},
\end{align*}
where $s = \sum_{l = 1}^r \bone(y_{\mis, l} = 1)$ counts the number of DLTs in $\by_{\mis}$, $0 \leq s \leq r$.
Similar to Equation \eqref{supp:eq_eta_fn}, consider
\begin{multline*}
\eta_N (p; \by_{\obs}, \by_{\mis})
= \frac{1}{N} \log \frac{L(p_0 \mid \by_{\obs}, \by_{\mis})}{L(p \mid \by_{\obs}, \by_{\mis})} \\
= \frac{1}{N} \Big\{  (n+s) (\log p_0 - \log p) + (m+r+s) [\log(1 - p_0) - \log (1 - p)] \Big\}.
\end{multline*}
The function $\eta_N (p; \by_{\obs}, \by_{\mis})$ monotonically decreases on $[0, (n+s) / N)$, reaches the minimum at $(n+s) / N$ and monotonically increases on $((n+s) / N, 1]$.



Let $\epsilon_0 = \min \{ \epsilon_1, \epsilon_2 \}$.
Similar to the proof of Lemma \ref{lem:consistency}, by the strong law of large numbers, and since $r = o(N)$, we can show $\exists N_0 > 0$, when $N > N_0$, $\pi(p \in \C_{\epsilon_0} \mid \by_{\obs}, \by_{\mis}) \geq 0.99$ a.s., for any $0 \leq s \leq r$.
Also, when $N > N_0$, $| \hat{p} - p_0 | < \epsilon_0 $ a.s. 
If $p_{0} \in (p^* - \epsilon_1, p^* + \epsilon_2)$, then for $N > N_0$ and for any $\by_{\mis}$, $\arg \max_{j} \Pr(p \in I_j  \mid \by_{\obs}, \by_{\mis})  = \textS$ and $\hat{p} \in I_{\textS}$ a.s.
That is, $\A^*(\by_{\obs}, \by_{\mis}) = d$ for any $\by_{\mis}$.
As a result, 
\begin{align*}
\Pr(A = d \mid \HH) = \sum_{\by_{\mis} \in \mathcal{Y}_{\mis}} \Pr(\bY_{\mis} = \by_{\mis} \mid \HH) = 1,
\end{align*}
a.s. for $N > N_0$.

\noindent (2, 3) For the same reason, if $p_{0} < p^* - \epsilon_1$, then $\exists N_0$, when $N > N_0$, $\forall \by_{\mis} \in \mathcal{Y}_{\mis}$, $\A^*(\by_{\obs}, \by_{\mis}) = d + 1$ a.s. Thus $\Pr(A = d + 1 \mid \HH) = 1$ a.s.
If $p_{0} > p^* + \epsilon_2$, then $\exists N_0$, when $N > N_0$, $\forall \by_{\mis} \in \mathcal{Y}_{\mis}$, $\A^*(\by_{\obs}, \by_{\mis}) = d - 1$ a.s. Thus $\Pr(A = d-1 \mid \HH) = 1$ a.s.
\end{proof}

\subsection{Coherence Principles}
\label{supp:sec:coherence}
                                                                                         
We prove below for an example that an interval-based TITE design is interval coherent in the sense of Definition \ref{supp:def:coherence_interval}.

Consider an interval-based TITE design that makes dose-finding decisions based on the MLE (see Section \ref{sec:complete_data_design}). We establish its interval coherence in de-escalation. 
First, consider a function
\begin{align*}
\ell(p, \rho_1, \ldots, \rho_r; n, m) \triangleq \frac{n}{p} - \frac{m}{1 - p} - \sum_{i = 1}^r \frac{\rho_i}{1 - \rho_i p },
\end{align*}
where $n$, $m$ and $r$ are some non-negative integers.
We have 
\begin{align*}
\frac{\partial \ell}{\partial p} &= - \frac{n}{p^2} - \frac{m}{(1-p)^2} - \sum_{i = 1}^r \frac{\rho_i^2}{(1 - \rho_i p)^2} < 0, \quad \text{and} \\
\frac{\partial \ell}{\partial \rho_i} &= - \frac{1}{(1 - \rho_i p)^2} < 0,
\end{align*}
for all $p$ and $\rho_i$, which mean that $\ell$ monotonically decreases with $p$ and $\rho_i$.

Next, suppose the currently-administrated doses just prior to $\tau$ and $\tau + \tau'$ are both $d$. 
At dose $d$,
let $n_d$ and $m_d$ denote the numbers patients that have or do not have DLTs just prior to $\tau$,
let $i = 1, \ldots, r_1$ index the patients that are still being followed just prior to $\tau$, 
and let $i = r_1 + 1, \ldots, r_2$ index the patients that are enrolled between $[\tau, \tau + \tau')$.
The MLE of $p_d$ just prior to $\tau$, denoted by $\hat{p}_d(\tau)$, satisfies
\begin{align*}
\ell[\hat{p}_d(\tau), \rho_1(\tau), \ldots, \rho_{r_1 + r_2} (\tau); n_d, m_d] = 0,
\end{align*}
where $\rho_i(\tau) = \rho[v_i(\tau) \mid d, \hat{\bxi}]$ for some $\hat{\bxi}$ for $i = 1, \ldots, r_1$, and $\rho_i(\tau) = 0$ for $i = r_1 + 1, \ldots, r_2$.
On the other hand, suppose no DLT occurs at dose $d$ during $[\tau, \tau + \tau')$. Then, the MLE of $p_d$ just prior to $\tau + \tau'$, denoted by $\hat{p}_d(\tau + \tau')$, satisfies
\begin{align*}
\ell[\hat{p}_d(\tau + \tau'), \rho_1(\tau + \tau'), \ldots, \rho_{r_1 + r_2} (\tau + \tau'); n_d, m_d] = 0,
\end{align*}
where $\rho_i(\tau + \tau') = \rho[v_i(\tau + \tau') \mid d, \hat{\bxi}]$ if patient $i$ is still being followed just prior to $\tau + \tau'$, and $\rho_i(\tau + \tau') = 1$ if patient $i$ has finished followup just prior to $\tau + \tau'$ with no DLT.

By definition, $\rho_i(\tau + \tau') \geq \rho_i(\tau)$.
Therefore, 
\begin{multline*}
\ell[\hat{p}_d(\tau + \tau'), \rho_1(\tau), \ldots, \rho_{r_1 + r_2} (\tau); n_d, m_d] \geq \\ 
\ell[\hat{p}_d(\tau + \tau'), \rho_1(\tau + \tau'), \ldots, \rho_{r_1 + r_2} (\tau + \tau'); n_d, m_d] = 0 = \\
\ell[\hat{p}_d(\tau), \rho_1(\tau), \ldots, \rho_{r_1 + r_2} (\tau); n_d, m_d].
\end{multline*}
Since $\ell$ monotonically decreases with $p$, we have  $\hat{p}_d(\tau + \tau') \leq \hat{p}_d(\tau)$, thus  $\A[\HH(\tau + \tau')] \geq \A[\HH(\tau)]$.
Coherence in escalation can be proved in a similar way.

\clearpage

\section{Selection of the MTD}
\label{supp:sec:sel_mtd}

Except for the 3+3 and R6 designs, the patient enrollment is terminated if the number of enrolled patients reaches the pre-specified maximum sample size $N^*$ or  an early stopping rule (e.g. Safety rule 1) is triggered.
After all patients have finished their DLT assessment, the trial completes, and the next step is to recommend an MTD. 
The selection of MTD does not involve any pending outcomes and is simply a problem of statistical inference about $\bp$ under the likelihood \eqref{eq:likelihood_complete} and the order constraint
 $p_1 \leq p_2 \leq \cdots \leq p_J$.
Usually, the doses with $\Pr(p_z > p^* \mid \data) > \nu$ for a $\nu$ close to 1 and the doses that have never been tried are excluded from the MTD candidates. If the trial is stopped early because the lowest dose is overly toxic, no MTD will be selected.
 
The MTD selection rules for the CRM and SPM are consistent with their dose assignment rules. For CRM, the dose $d^* = \argmin_{z} | \hat{p}_{z} - p^* |$ is selected as the MTD. For SPM, the dose $\hat{\gamma} = \argmax_{\gamma} \pi(\gamma \mid \by, \bz)$ is selected as the MTD. See more details in Section \ref{sec:complete_data_design} of the main manuscript.
On the other hand, the MTD selection rules for the BOIN, mTPI-2, keyboard and i3+3 are different from their dose assignment rules, as their dose assignments only depend on outcomes at the current dose. To impose the order constraint, 
an isotonic regression is performed using the pooled adjacent violators algorithm \citep{ji2007dose}, resulting in estimates $\hat{\bp}$ satisfying $\hat{p}_1 \leq \hat{p}_2 \leq \cdots \leq \hat{p}_J$.
For BOIN and keyboard, the dose $d^* = \argmin_{z} | \hat{p}_{z} - p^* |$ is selected as the MTD. For mTPI-2 and i3+3, the dose $d^*$ with the smallest distance will be selected only if $\hat{p}_{d^*} \leq p^* + \epsilon_2$, and otherwise, the highest dose with DLT probability lower than $p^* + \epsilon_2$ is selected, which is more conservative.
For the time-to-event designs, we can simply apply the MTD selection rules of their complete-data counterparts.

\clearpage

\section{Simulation Details}

\subsection{Dose-Toxicity Scenarios}

We summarize the 18 dose-toxicity scenarios in Table \ref{tbl:simu_DLT_prob}.
We follow \cite{guo2017bayesian} to define the MTD
as the highest dose whose probability of DLT is close to or lower than $p^*$. Specifically, the doses $z$ with $p_z \in [p^* - 0.05, p^* + 0.05]$ are considered MTDs, and if such doses do not exist, the highest dose $z$ with $p_z < p^*$ is considered as the MTD.
We note the definition of MTD may be slightly different in other articles.

\begin{table}[h!]
\caption{True DLT probabilities of the 18 dose-toxicity scenarios.
The target DLT probability is 0.2 for scenarios 1--9 and is 0.3 for scenarios 10--18. The MTDs are marked in bold. }
\label{tbl:simu_DLT_prob}
\begin{center}
\scalebox{1}{
\renewcommand{\arraystretch}{0.8}
\begin{tabular}{cccccccc}
\hline \hline
\multirow{2}{*}{Scn.}  & \multicolumn{7}{c}{Dose levels} \\ \cline{2-8}
  & 1 & 2 & 3 & 4 & 5 & 6 & 7  \\ \hline
  & \multicolumn{7}{c}{Target DLT probability 0.20}  \\
1  &  0.28  &  0.36  &  0.44  &  0.52  &  0.60  &  0.68  & 0.76 \\
2  &  0.05  &  \textbf{0.20}  &  0.46  &  0.50  &  0.60  &  0.70  & 0.80  \\
3  &  0.02  &  0.05  &  \textbf{0.20}  &  0.28  &  0.34  &  0.40  &  0.44  \\
4  &  0.01  &  0.05  &  0.10  &  \textbf{0.20}  &  0.32  &  0.50  &  0.70 \\
5  &  0.01  &  0.04  &  0.07  &  \textbf{0.10}  &  0.50  &  0.70  &  0.90 \\
6  &  0.01  &  0.05  &  0.10  &  0.14  &  \textbf{0.20}  &  0.26  &  0.34 \\
7  &  0.01  &  0.02  &  0.03  &  0.05  &  \textbf{0.20}  &  0.40  &  0.50 \\
8  &  0.01  &  0.04  &  0.07  &  0.10  &  \textbf{0.15}  &  \textbf{0.20}  &  \textbf{0.25} \\
9  &  0.01  &  0.02  &  0.03  &  0.04  &  0.05  &  \textbf{0.20}  &  0.45 \\
\hline
  & \multicolumn{7}{c}{Target DLT probability 0.30}  \\
10  &  0.40  &  0.45  &  0.50  &  0.55  &  0.60  &  0.65  &  0.70 \\
11  &  \textbf{0.30}  &  0.40  &  0.50  &  0.60  &  0.70  &  0.80  &  0.90 \\
12  &  0.14  &  \textbf{0.30}  &  0.39  &  0.48  &  0.56  &  0.64  &  0.70 \\
13  &  0.07  &  \textbf{0.23}  &  0.41  &  0.49  &  0.62  &  0.68  &  0.73 \\
14  &  0.05  &  0.15  &  \textbf{0.30}  &  0.40  &  0.50  &  0.60  &  0.70 \\
15  &  0.05  &  0.12  &  0.20  &  \textbf{0.30}  &  0.38  &  0.49  &  0.56 \\
16  &  0.01  &  0.04  &  0.08  &  0.15  &  \textbf{0.30}  &  0.36  &  0.43 \\
17  &  0.02  &  0.04  &  0.08  &  0.10  &  0.20  &  \textbf{0.30}  &  0.40 \\
18  &  0.01  &  0.03  &  0.05  &  0.07  &  0.09  &  \textbf{0.30}  &  0.50 \\
\hline \hline
\end{tabular}
}
\end{center}
\end{table}

\subsection{Design Specifications}
\label{supp:sec:design_spec}



The specification of each design is as follows.
For TITE-TPI, POD-TPI, and mTPI-2, the equivalence interval is chosen with $\epsilon_1 = \epsilon_2 = 0.05$.
For CRM and TITE-CRM, we use the power model, $p_z = \phi(z, \alpha) = p_{0z}^{\exp(\alpha)}$, with $\alpha \sim \text{N}(0, 1.34^2)$. 
The skeleton $(p_{01}, \ldots, p_{0D})$ is calibrated based on \cite{lee2009model}, with prior guess of MTD being the middle dose 4 and halfwidth of the indifference interval being 0.05.
For BOIN and TITE-BOIN, the decision boundaries are calculated based on $p_{\tL} = 0.6 p^*$ and $p_{\tR} = 1.4 p^*$. 

The MTD selection rules of the time-to-event designs follow their complete-data counterparts. The doses with $\Pr(p_z > p^* \mid \data) > 0.95$ are excluded from the MTD candidates for all designs.

\clearpage

\subsection{Scenario-specific Results}
\label{supp:sec:ss_result}

The figures below show the scenario-specific operating characteristics under Setting 1 for mTPI-2, TITE-TPI, POD-TPI, CRM, TITE-CRM, BOIN, and TITE-BOIN.

\begin{center}
\includegraphics[width=\textwidth]{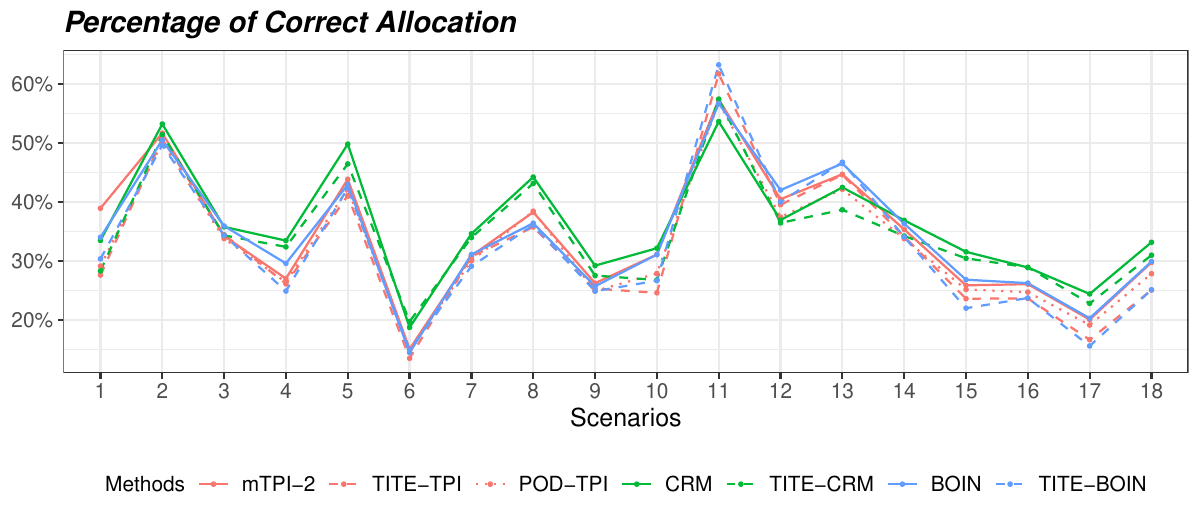}
\includegraphics[width=\textwidth]{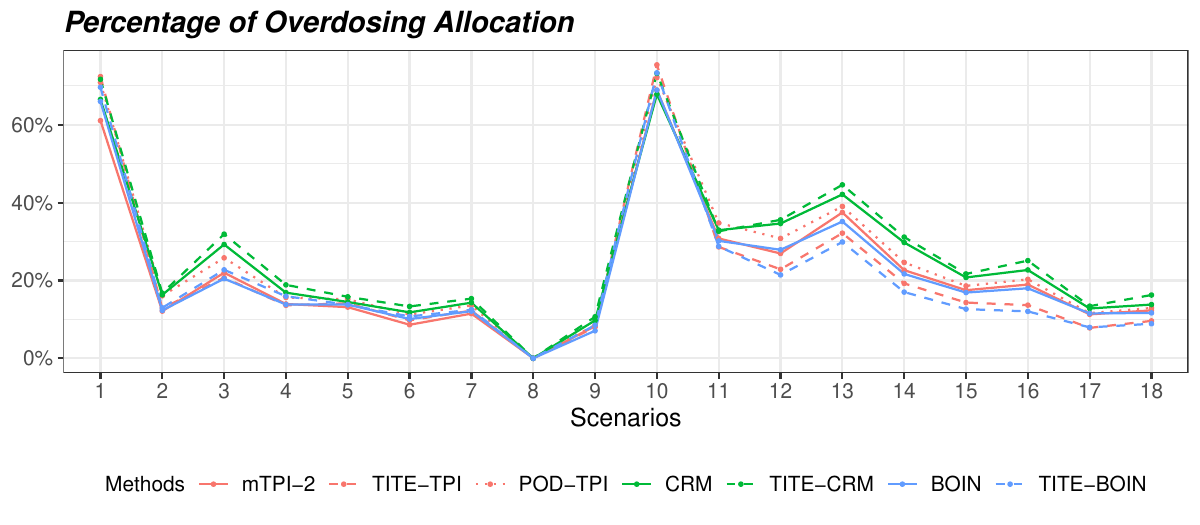}
\includegraphics[width=\textwidth]{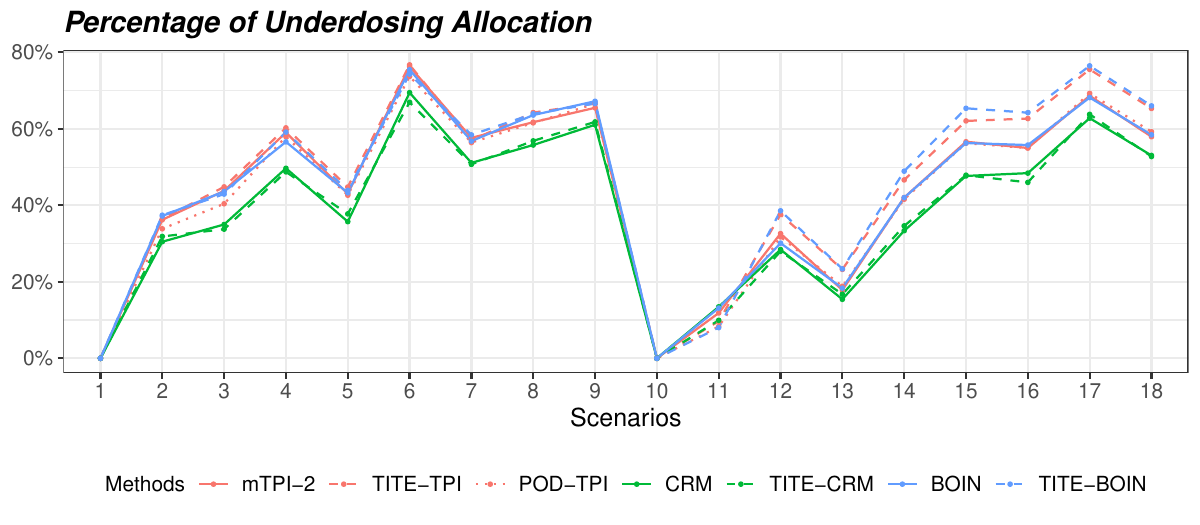}
\includegraphics[width=\textwidth]{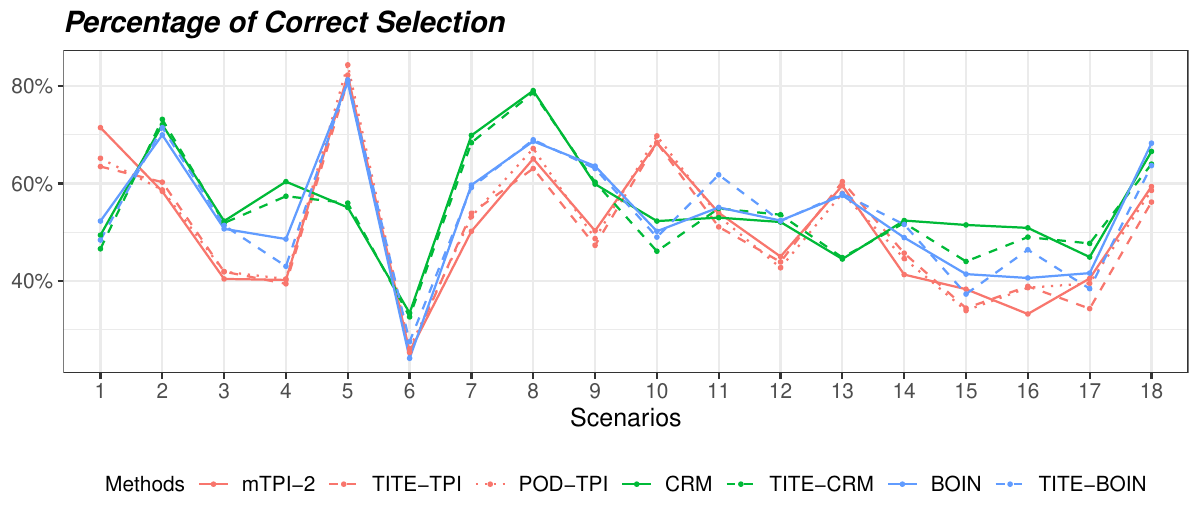}
\includegraphics[width=\textwidth]{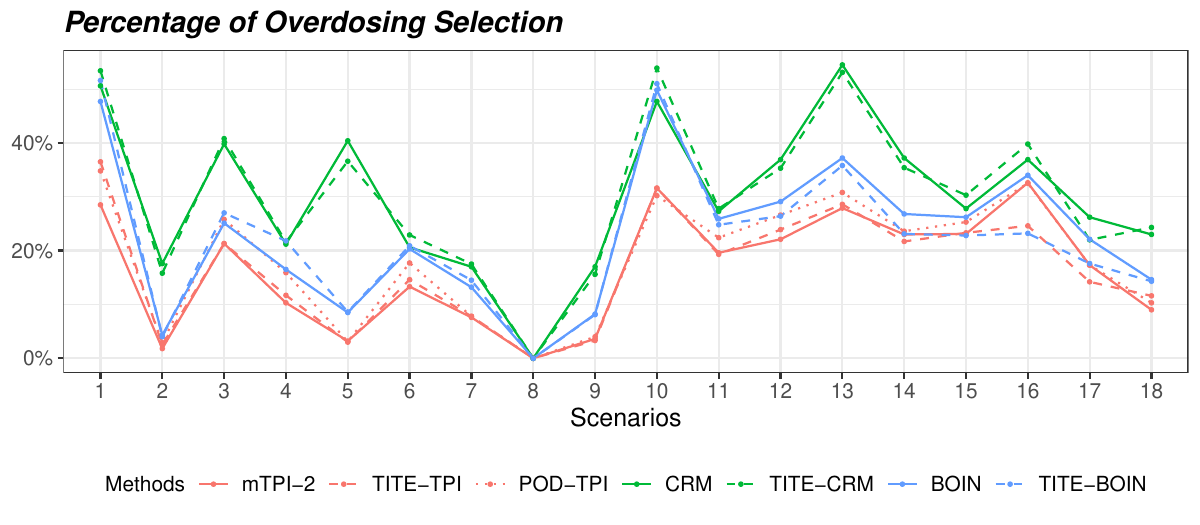}
\includegraphics[width=\textwidth]{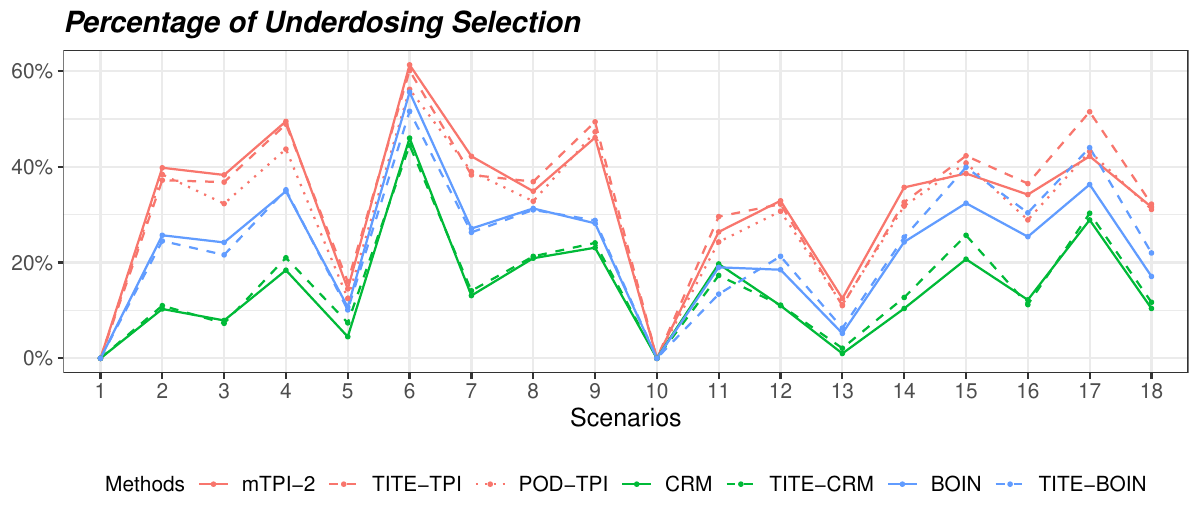}
\includegraphics[width=\textwidth]{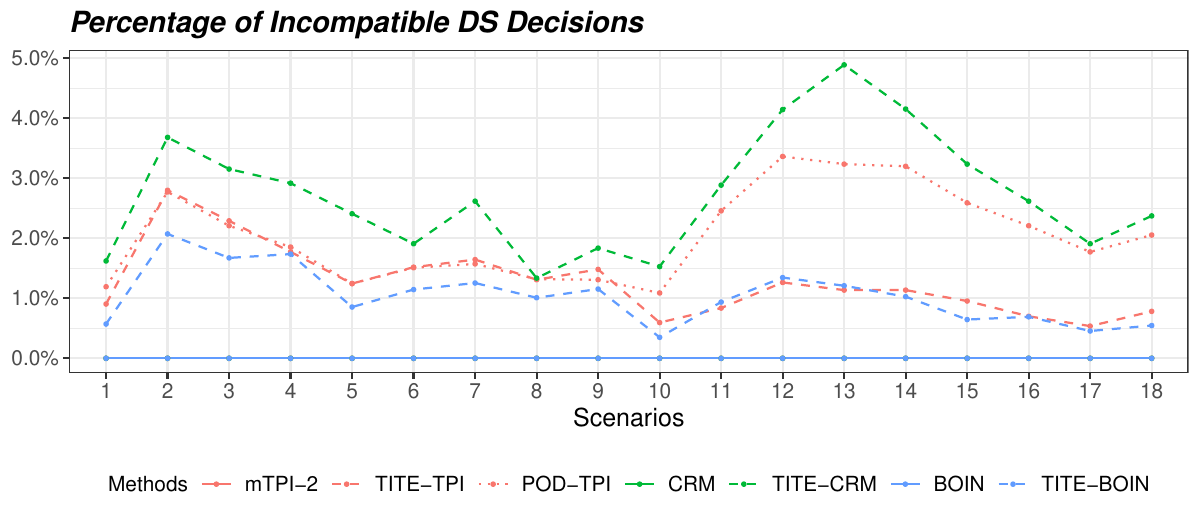}
\includegraphics[width=\textwidth]{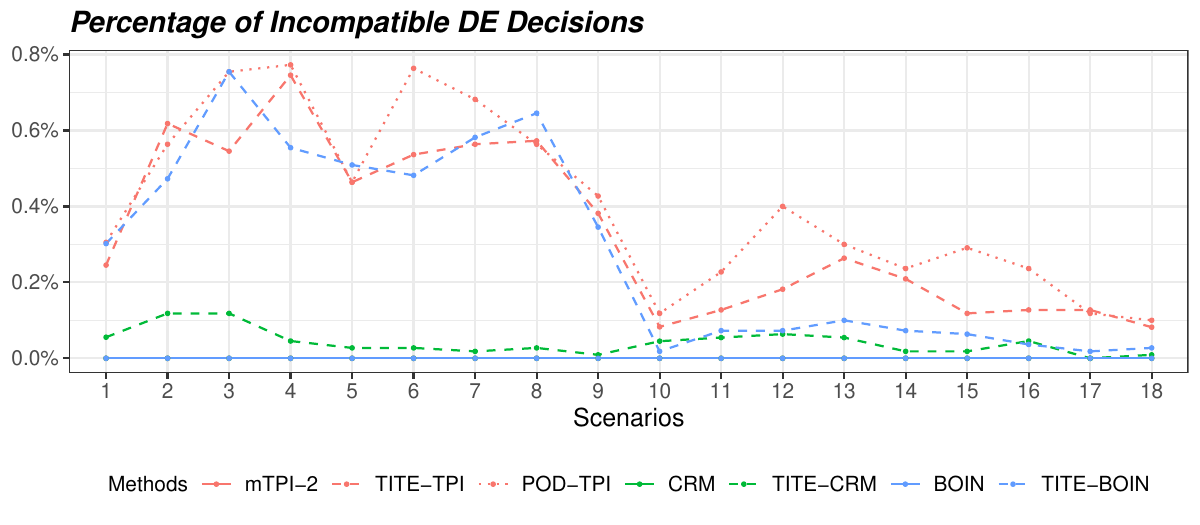}
\includegraphics[width=\textwidth]{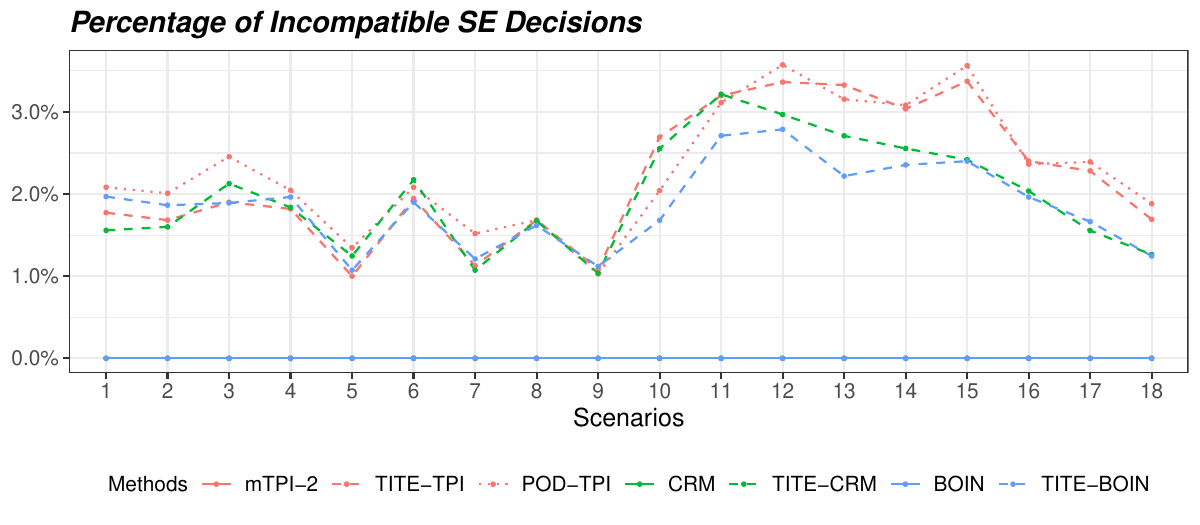}
\includegraphics[width=\textwidth]{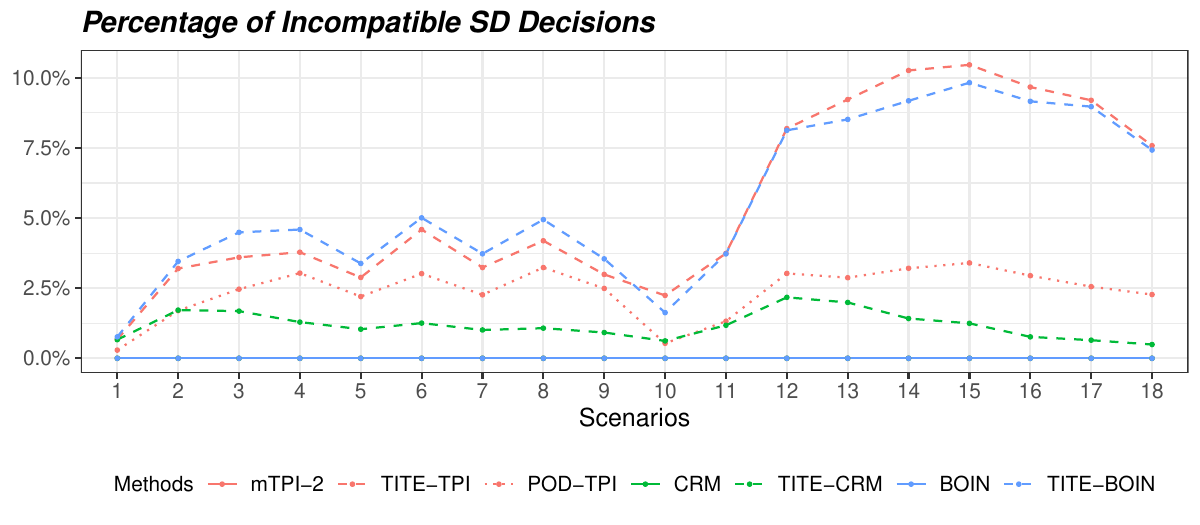}
\includegraphics[width=\textwidth]{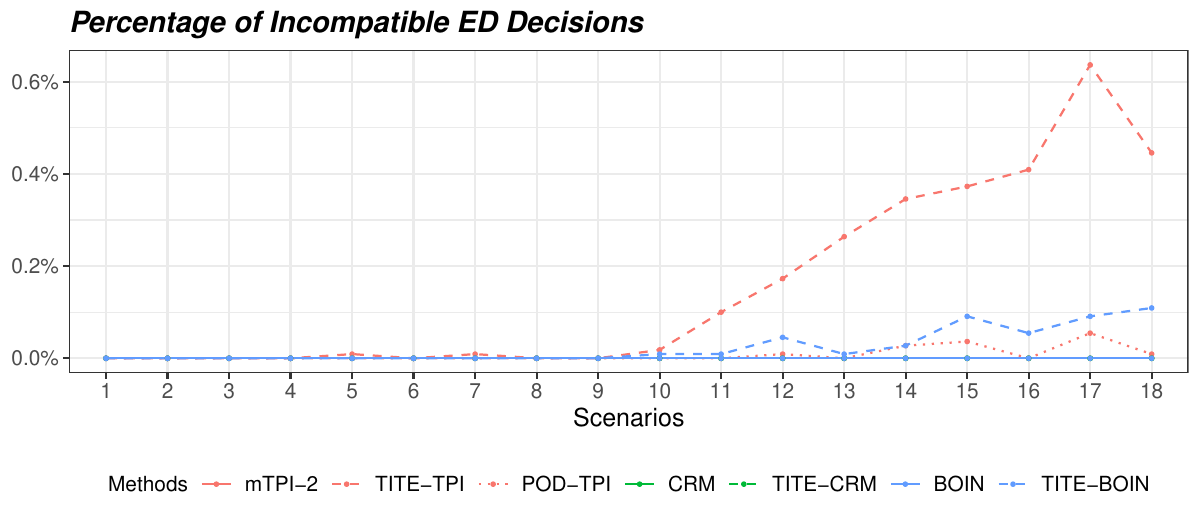}
\includegraphics[width=\textwidth]{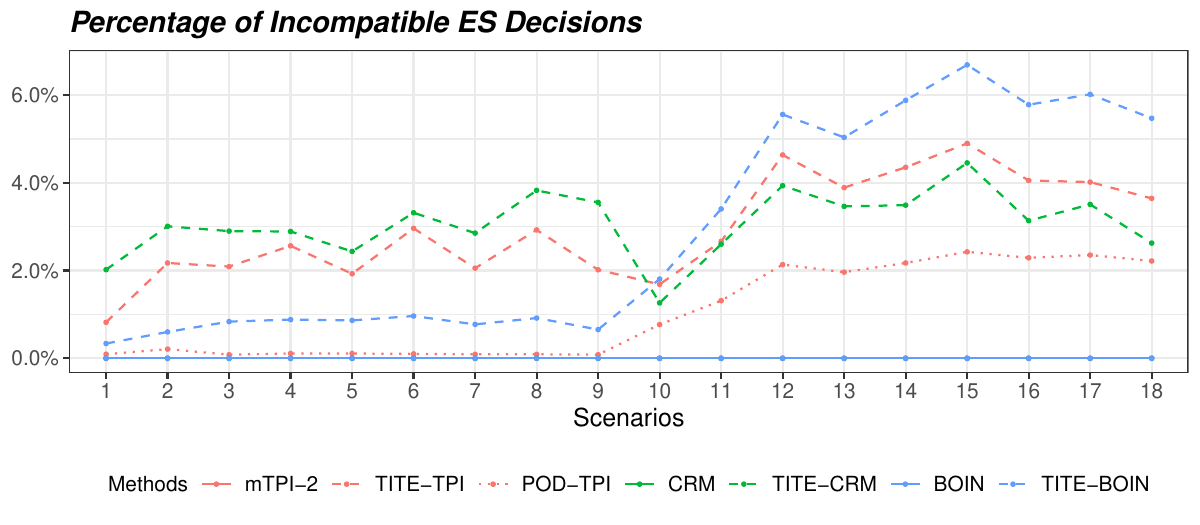}
\includegraphics[width=\textwidth]{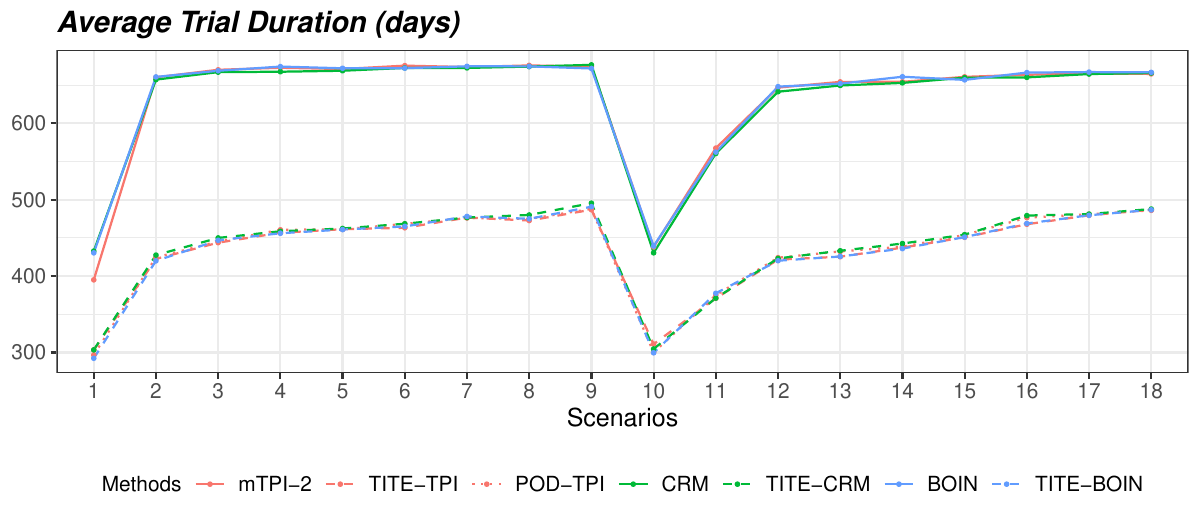}
\end{center}

The figures below show the scenario-specific operating characteristics under Setting 1 for mTPI-2-LA, TITE-TPI-NS, TITE-TPI-PSR, TITE-TPI-PSR2, POD-TPI-NS, POD-TPI-PSR, and POD-TPI-PSR2.

\begin{center}
\includegraphics[width=\textwidth]{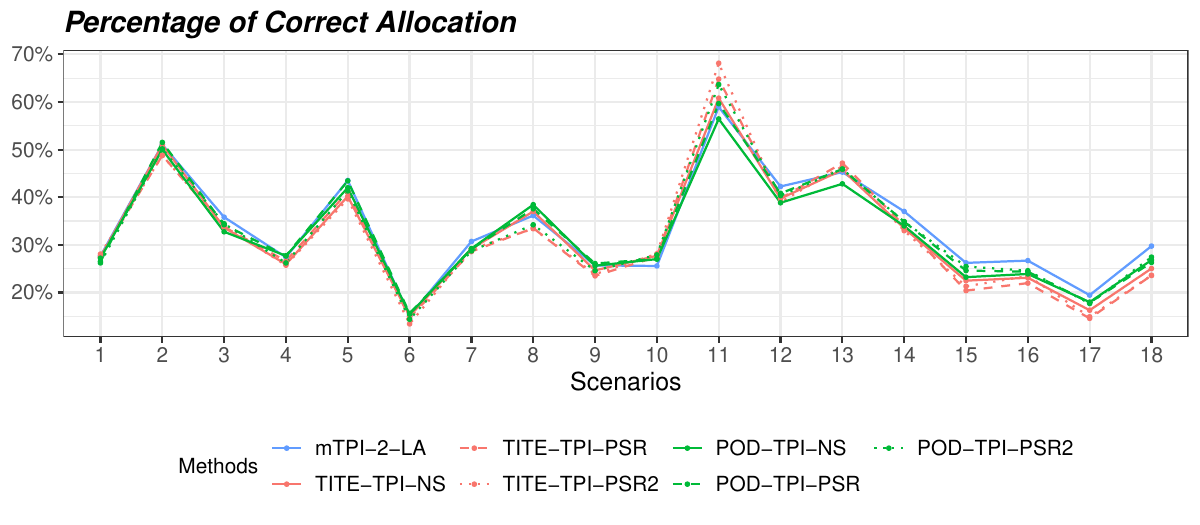}
\includegraphics[width=\textwidth]{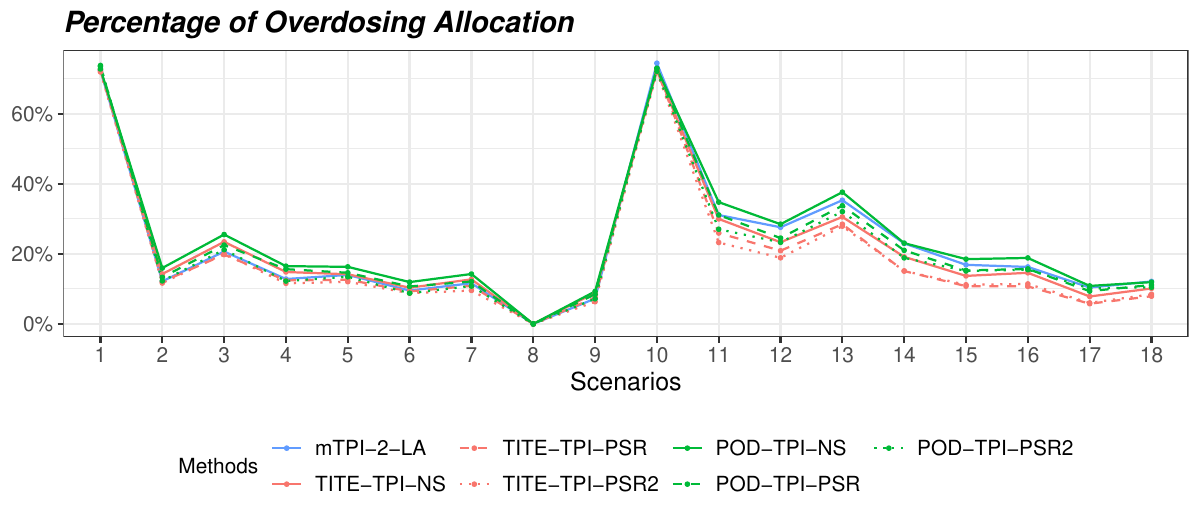}
\includegraphics[width=\textwidth]{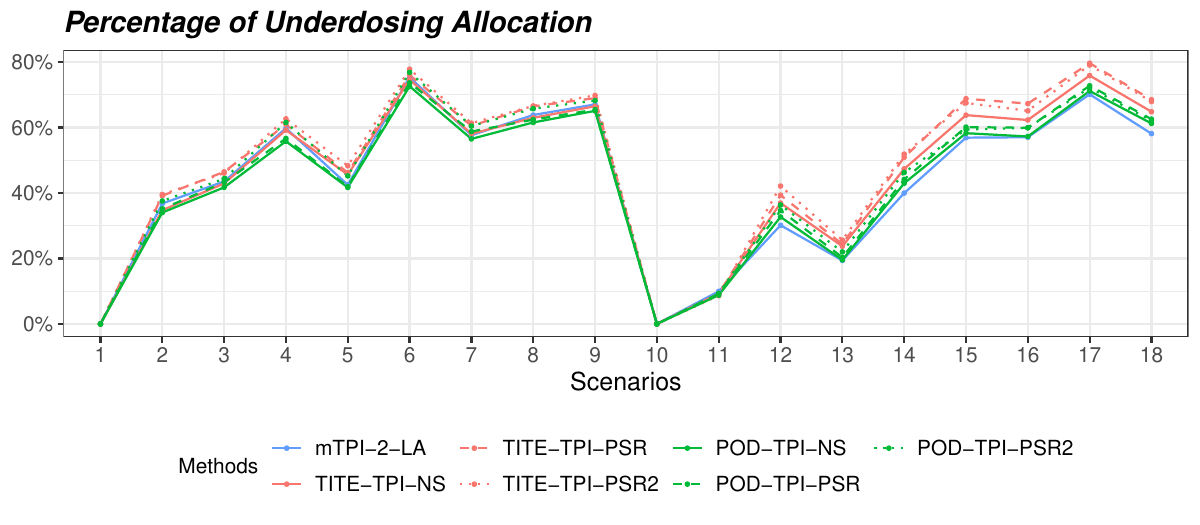}
\includegraphics[width=\textwidth]{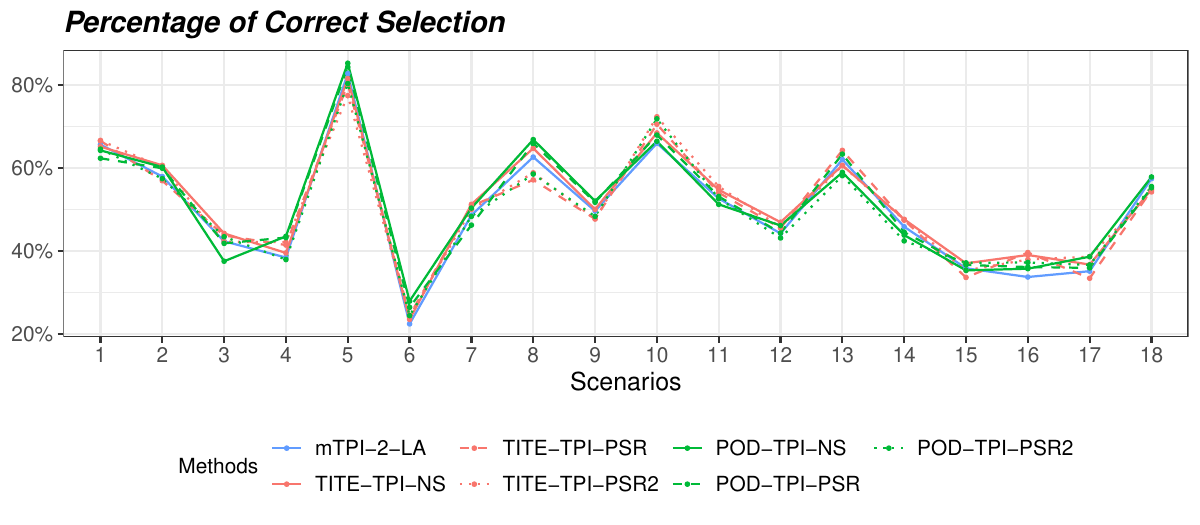}
\includegraphics[width=\textwidth]{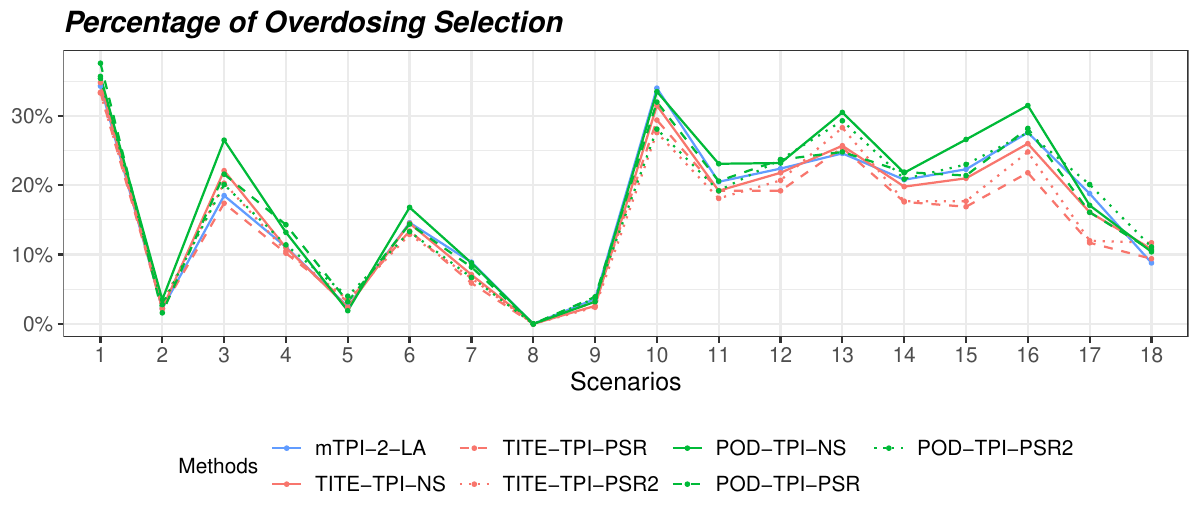}
\includegraphics[width=\textwidth]{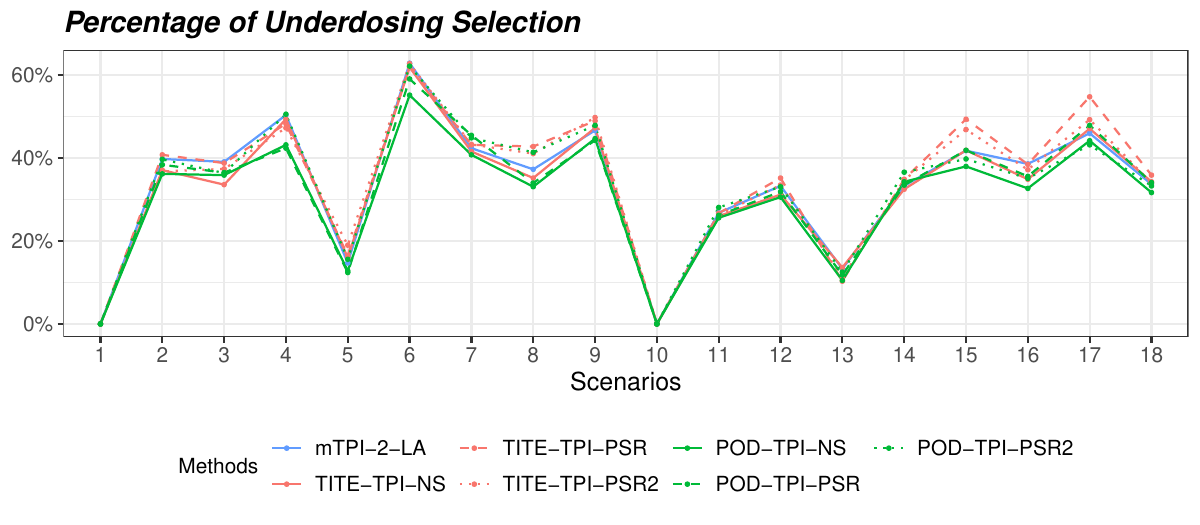}
\includegraphics[width=\textwidth]{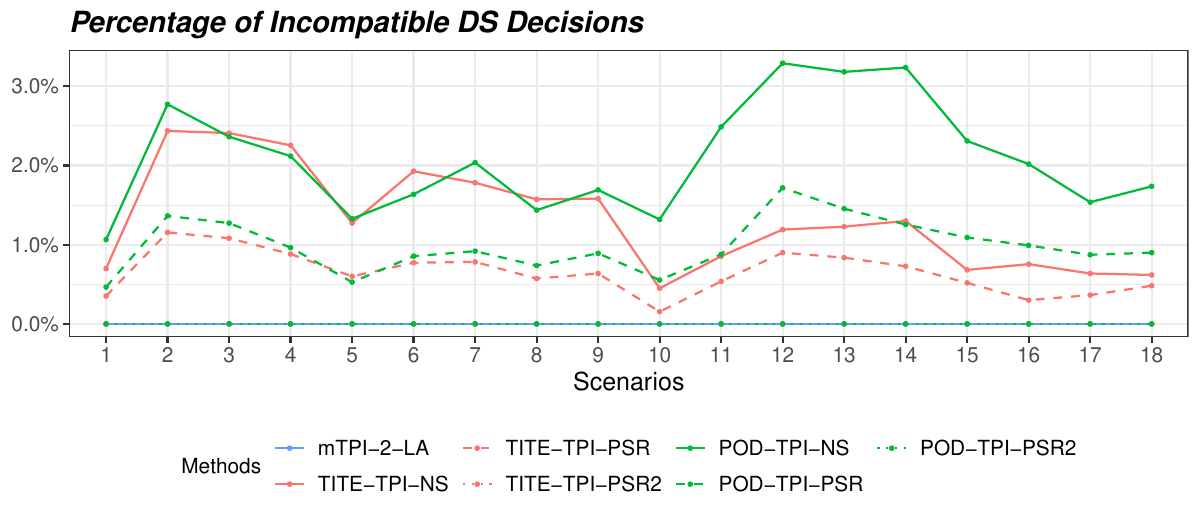}
\includegraphics[width=\textwidth]{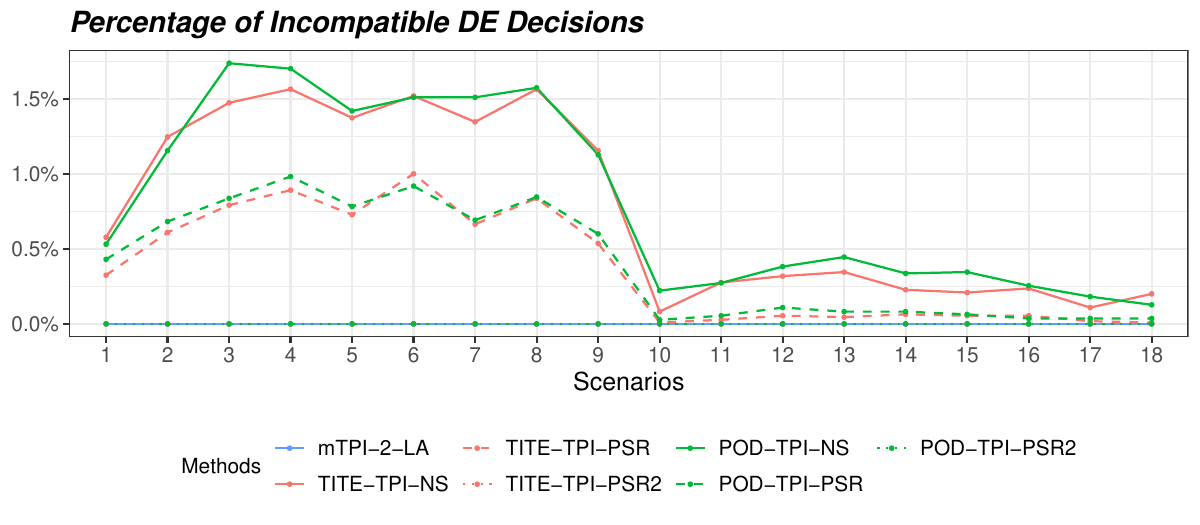}
\includegraphics[width=\textwidth]{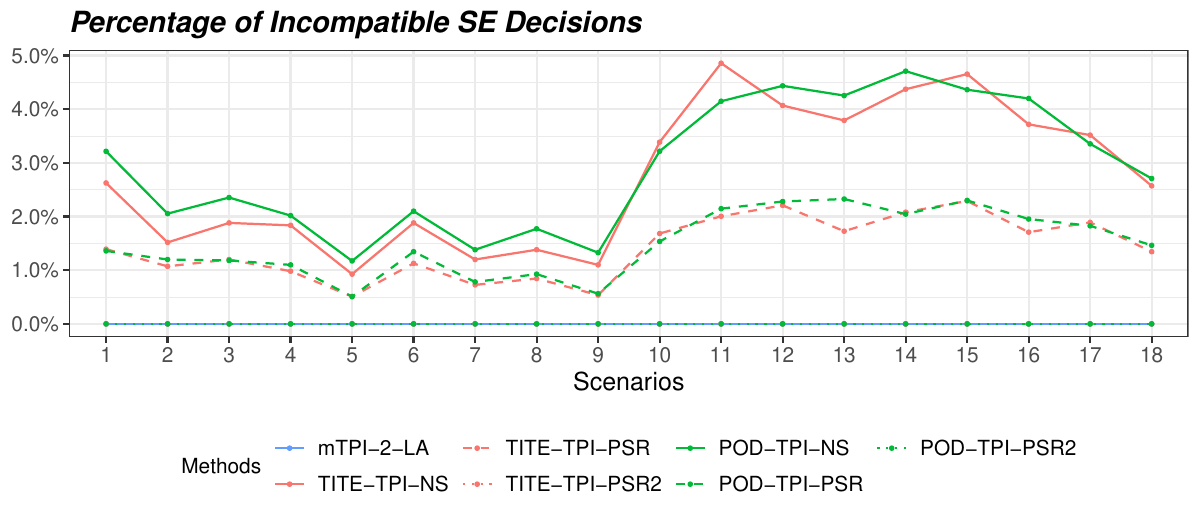}
\includegraphics[width=\textwidth]{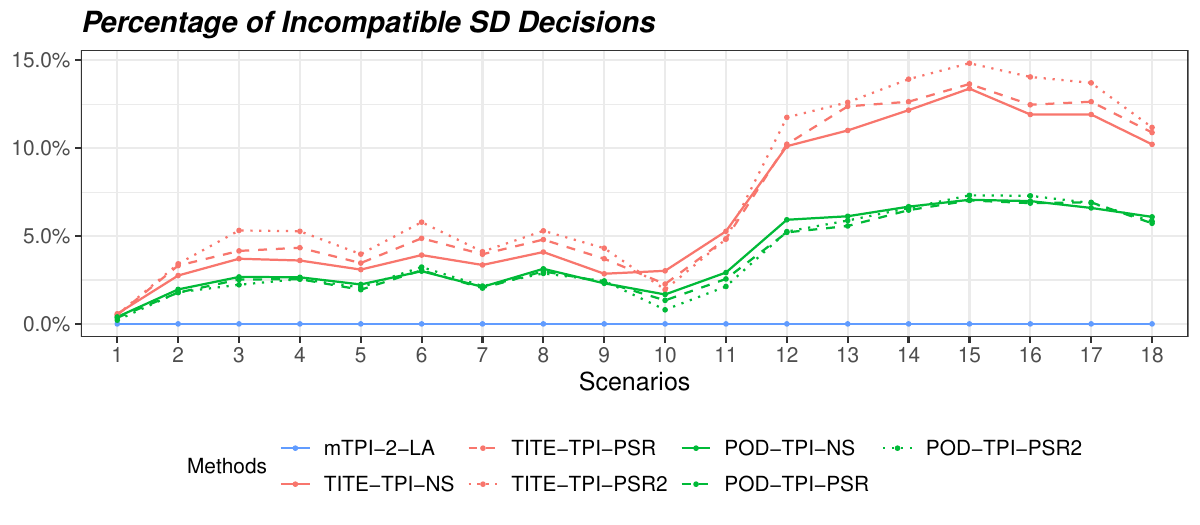}
\includegraphics[width=\textwidth]{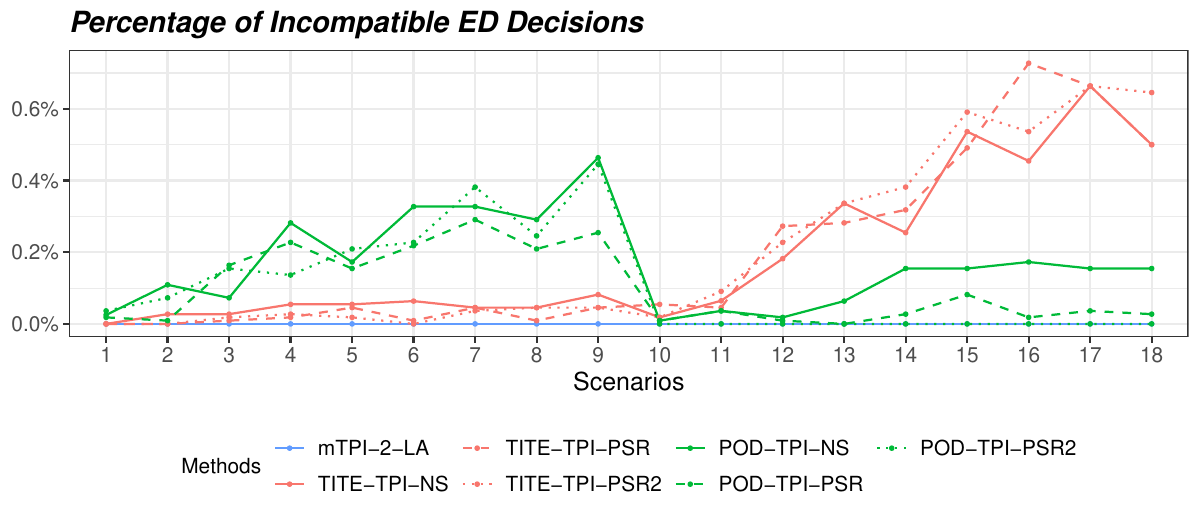}
\includegraphics[width=\textwidth]{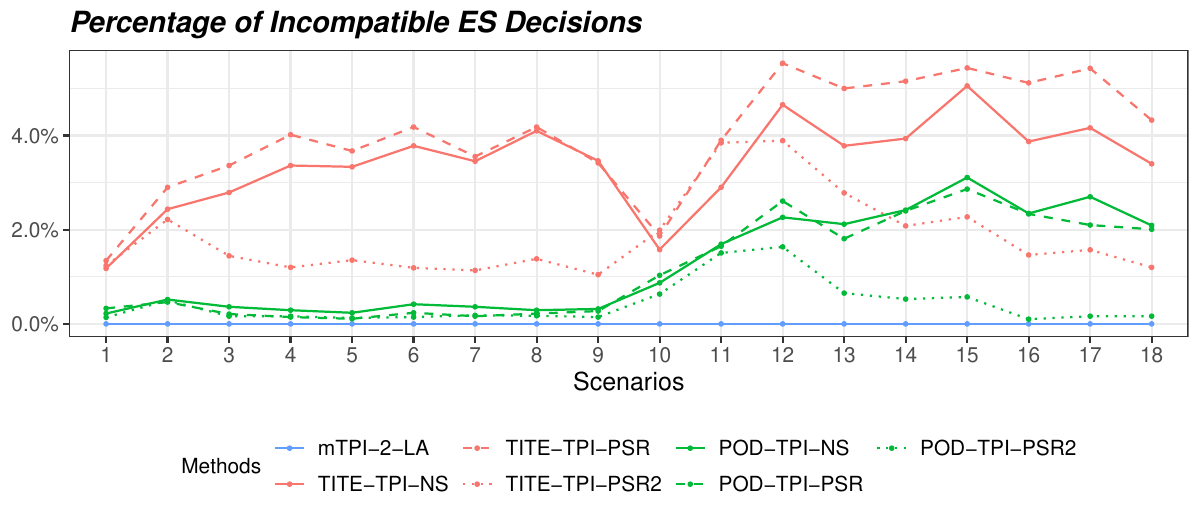}
\includegraphics[width=\textwidth]{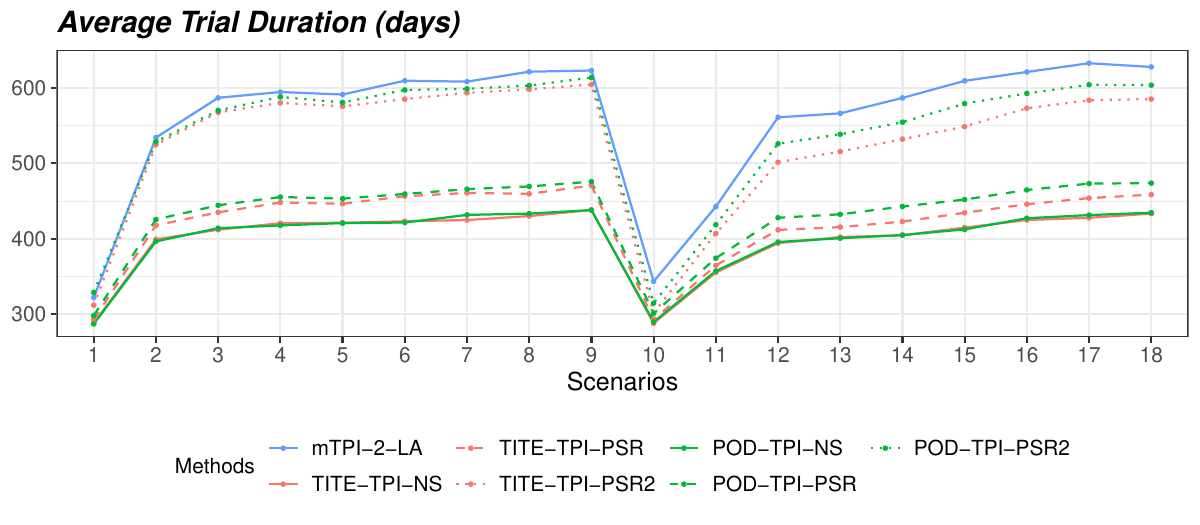}
\end{center}

\clearpage

\subsection{Sensitivity Analyses with Different Maximum Sample Sizes}
\label{supp:sec:sens_sample_size}

To explore the performance of time-to-event designs with different 
we perform additional sensitivity analyses with alternative maximum sample sizes ($N^* = 48$ and $60$). The results are summarized in Tables \ref{supp:tbl:simu_48} and \ref{supp:tbl:simu_60} below. It is clear that the PCA and PCS of all designs increase with the sample size.
However, larger sample sizes are also associated with higher costs and longer trial duration.
Therefore, the investigator needs to strike a balance between sample size and the accuracy of MTD estimation.

\begin{center}
\begin{longtable}{l@{\extracolsep{4pt}}c@{\extracolsep{4pt}}c@{\extracolsep{4pt}}c@{\extracolsep{4pt}}c@{\extracolsep{4pt}}c@{\extracolsep{4pt}}c@{\extracolsep{4pt}}c@{\extracolsep{4pt}}c@{\extracolsep{4pt}}c@{\extracolsep{4pt}}c@{\extracolsep{4pt}}c@{\extracolsep{4pt}}c@{\extracolsep{4pt}}c}
\caption{Summary of simulation results under 18 dose-toxicity scenarios and 3 time-to-toxicity and accrual rate settings. 
Values shown are averages over simulated trials and scenarios.
The maximum sample size is $N^* = 48$. PCA, POA, PUA, PCS, POS, and PUS are in \%, DS, DE, SE, SD, ED, and ES are in $1/10^{3}$, and Dur is in days.} \label{supp:tbl:simu_48} \\
\hline \hline
\multicolumn{1}{c}{\multirow{2}{*}{Design}} & \multicolumn{3}{c}{Allocation} & \multicolumn{3}{c}{Selection} & \multicolumn{6}{c}{Risk} & Speed \\
\cline{2-4} \cline{5-7} \cline{8-13} \cline{14-14}
& PCA & POA & PUA & PCS & POS & PUS & DS & DE & SE & SD & ED & ES & Dur \\ \hline 
\endfirsthead

\multicolumn{14}{c}%
{ \tablename\ \thetable{} -- continued from previous page} \\
\hline 
\multicolumn{1}{c}{\multirow{2}{*}{Design}} & \multicolumn{3}{c}{Allocation} & \multicolumn{3}{c}{Selection} & \multicolumn{6}{c}{Risk} & Speed \\
\cline{2-4} \cline{5-7} \cline{8-13} \cline{14-14}
& PCA & POA & PUA & PCS & POS & PUS & DS & DE & SE & SD & ED & ES & Dur \\ \hline
\endhead

\hline \multicolumn{14}{r}{{Continued on next page}} \\
\endfoot

\hline \hline
\endlastfoot

 \multicolumn{14}{c}{Setting 1} \\ 
mTPI-2 & 39.0 & 22.0 & 39.0 & 55.4 & 13.5 & 31.2 & 0.0 & 0.0 & 0.0 & 0.0 & 0.0 & 0.0 & 839 \\
TITE-TPI & 37.4 & 21.4 & 41.2 & 55.8 & 14.2 & 30.0 & 15.0 & 2.9 & 20.8 & 52.9 & 1.3 & 27.2 & 549 \\
POD-TPI & 37.2 & 24.5 & 38.3 & 55.8 & 15.2 & 29.0 & 23.4 & 3.3 & 21.9 & 22.8 & 0.0 & 8.5 & 551 \\
CRM & 41.2 & 24.4 & 34.4 & 60.4 & 27.8 & 11.7 & 0.0 & 0.0 & 0.0 & 0.0 & 0.0 & 0.0 & 838 \\
TITE-CRM & 40.0 & 26.0 & 34.0 & 59.0 & 28.4 & 12.5 & 27.4 & 0.3 & 16.9 & 11.1 & 0.0 & 27.9 & 552 \\
BOIN & 39.0 & 21.9 & 39.1 & 58.6 & 19.8 & 21.6 & 0.0 & 0.0 & 0.0 & 0.0 & 0.0 & 0.0 & 840 \\
TITE-BOIN & 37.3 & 21.6 & 41.2 & 58.2 & 20.0 & 21.8 & 12.9 & 2.4 & 18.0 & 51.7 & 0.2 & 26.0 & 551 \\
\hdashline
mTPI-2-LA & 38.4 & 22.8 & 38.7 & 55.5 & 14.0 & 30.6 & 0.0 & 0.0 & 0.0 & 0.0 & 0.0 & 0.0 & 718 \\
TITE-TPI-NS & 36.8 & 22.3 & 40.9 & 55.7 & 14.5 & 29.9 & 15.9 & 6.2 & 25.8 & 59.2 & 1.6 & 31.1 & 515 \\
TITE-TPI-PSR & 36.5 & 20.2 & 43.3 & 55.3 & 13.6 & 31.1 & 7.4 & 2.8 & 13.6 & 64.6 & 1.7 & 37.4 & 541 \\
TITE-TPI-PSR2 & 36.5 & 19.4 & 44.1 & 55.1 & 13.1 & 31.8 & 0.0 & 0.0 & 0.0 & 70.3 & 1.8 & 16.9 & 675 \\
POD-TPI-NS & 37.3 & 24.8 & 37.9 & 56.2 & 15.3 & 28.4 & 23.2 & 6.8 & 27.2 & 34.2 & 1.2 & 10.2 & 516 \\
POD-TPI-PSR & 37.3 & 22.8 & 39.8 & 55.6 & 14.4 & 30.0 & 11.4 & 3.0 & 13.9 & 33.1 & 0.9 & 10.2 & 553 \\
POD-TPI-PSR2 & 37.6 & 21.2 & 41.2 & 55.2 & 13.4 & 31.4 & 0.0 & 0.0 & 0.0 & 33.1 & 0.7 & 3.7 & 693 \\
\hline
\multicolumn{14}{c}{Setting 2} \\ 
mTPI-2 & 39.1 & 22.0 & 38.9 & 55.7 & 13.2 & 31.0 & 0.0 & 0.0 & 0.0 & 0.0 & 0.0 & 0.0 & 856 \\
TITE-TPI & 36.8 & 24.7 & 38.6 & 56.0 & 15.7 & 28.2 & 29.3 & 7.4 & 33.4 & 42.7 & 1.1 & 23.6 & 561 \\
POD-TPI & 36.4 & 27.7 & 35.9 & 56.0 & 16.6 & 27.4 & 44.3 & 8.6 & 35.6 & 18.4 & 0.0 & 7.6 & 563 \\
CRM & 40.9 & 24.4 & 34.7 & 60.1 & 27.8 & 12.1 & 0.0 & 0.0 & 0.0 & 0.0 & 0.0 & 0.0 & 856 \\
TITE-CRM & 39.5 & 28.1 & 32.4 & 59.8 & 27.5 & 12.6 & 50.0 & 1.4 & 30.2 & 8.2 & 0.0 & 24.2 & 564 \\
BOIN & 39.2 & 21.9 & 39.0 & 59.1 & 19.5 & 21.5 & 0.0 & 0.0 & 0.0 & 0.0 & 0.0 & 0.0 & 857 \\
TITE-BOIN & 36.8 & 24.3 & 38.8 & 58.0 & 21.6 & 20.4 & 27.0 & 6.2 & 30.0 & 42.3 & 0.1 & 21.7 & 562 \\
\hdashline
mTPI-2-LA & 38.4 & 22.8 & 38.8 & 55.8 & 13.2 & 30.9 & 0.0 & 0.0 & 0.0 & 0.0 & 0.0 & 0.0 & 739 \\
TITE-TPI-NS & 35.9 & 26.6 & 37.5 & 55.8 & 16.0 & 28.2 & 31.9 & 14.9 & 42.6 & 47.2 & 1.3 & 25.0 & 528 \\
TITE-TPI-PSR & 36.9 & 23.5 & 39.6 & 56.6 & 14.5 & 28.9 & 20.3 & 8.5 & 28.2 & 50.8 & 1.6 & 30.5 & 554 \\
TITE-TPI-PSR2 & 36.9 & 20.0 & 43.1 & 54.9 & 13.1 & 32.0 & 0.0 & 0.0 & 0.0 & 57.6 & 1.7 & 15.4 & 700 \\
POD-TPI-NS & 36.0 & 29.5 & 34.5 & 56.4 & 17.0 & 26.6 & 47.6 & 17.1 & 45.3 & 24.9 & 1.0 & 8.6 & 529 \\
POD-TPI-PSR & 37.0 & 26.5 & 36.5 & 56.2 & 16.0 & 27.9 & 29.4 & 9.5 & 30.2 & 24.7 & 0.8 & 8.6 & 565 \\
POD-TPI-PSR2 & 37.8 & 22.0 & 40.2 & 56.0 & 13.3 & 30.8 & 0.0 & 0.0 & 0.0 & 23.4 & 0.8 & 3.7 & 718 \\
\hline
\multicolumn{14}{c}{Setting 3} \\
mTPI-2 & 38.9 & 22.2 & 38.9 & 55.5 & 13.9 & 30.6 & 0.0 & 0.0 & 0.0 & 0.0 & 0.0 & 0.0 & 625 \\
TITE-TPI & 36.3 & 21.6 & 42.1 & 55.5 & 14.2 & 30.3 & 18.8 & 3.4 & 25.1 & 61.2 & 3.1 & 41.9 & 352 \\
POD-TPI & 36.3 & 23.8 & 39.9 & 55.7 & 15.2 & 29.1 & 24.1 & 3.6 & 26.1 & 41.9 & 0.1 & 23.5 & 354 \\
CRM & 40.8 & 24.4 & 34.8 & 59.6 & 28.1 & 12.3 & 0.0 & 0.0 & 0.0 & 0.0 & 0.0 & 0.0 & 624 \\
TITE-CRM & 39.5 & 26.0 & 34.5 & 59.4 & 28.1 & 12.5 & 35.6 & 0.5 & 19.4 & 12.0 & 0.0 & 45.5 & 357 \\
BOIN & 39.3 & 21.8 & 38.8 & 59.4 & 19.1 & 21.5 & 0.0 & 0.0 & 0.0 & 0.0 & 0.0 & 0.0 & 625 \\
TITE-BOIN & 36.5 & 20.6 & 42.9 & 58.9 & 19.1 & 22.0 & 15.6 & 2.6 & 21.2 & 61.9 & 0.8 & 44.6 & 351 \\
\hdashline
mTPI-2-LA & 38.4 & 22.9 & 38.7 & 55.8 & 13.7 & 30.4 & 0.0 & 0.0 & 0.0 & 0.0 & 0.0 & 0.0 & 516 \\
TITE-TPI-NS & 35.4 & 21.4 & 43.2 & 55.4 & 14.4 & 30.1 & 19.2 & 5.7 & 28.6 & 74.8 & 6.1 & 49.2 & 316 \\
TITE-TPI-PSR & 35.5 & 18.5 & 46.0 & 55.2 & 12.4 & 32.4 & 9.9 & 2.5 & 14.0 & 80.0 & 6.4 & 59.0 & 333 \\
TITE-TPI-PSR2 & 35.6 & 18.3 & 46.1 & 54.5 & 13.1 & 32.3 & 0.0 & 0.0 & 0.0 & 91.0 & 5.2 & 24.6 & 471 \\
POD-TPI-NS & 35.9 & 23.2 & 40.9 & 56.4 & 15.0 & 28.6 & 22.4 & 6.5 & 30.1 & 60.8 & 3.0 & 28.8 & 316 \\
POD-TPI-PSR & 36.3 & 21.4 & 42.3 & 55.1 & 14.1 & 30.8 & 14.3 & 3.0 & 16.1 & 59.3 & 2.2 & 26.8 & 343 \\
POD-TPI-PSR2 & 36.8 & 20.5 & 42.7 & 54.9 & 14.0 & 31.2 & 0.0 & 0.0 & 0.0 & 58.2 & 0.2 & 10.3 & 481 \\
\end{longtable}
\end{center}

\begin{center}
\begin{longtable}{l@{\extracolsep{4pt}}c@{\extracolsep{4pt}}c@{\extracolsep{4pt}}c@{\extracolsep{4pt}}c@{\extracolsep{4pt}}c@{\extracolsep{4pt}}c@{\extracolsep{4pt}}c@{\extracolsep{4pt}}c@{\extracolsep{4pt}}c@{\extracolsep{4pt}}c@{\extracolsep{4pt}}c@{\extracolsep{4pt}}c@{\extracolsep{4pt}}c}
\caption{Summary of simulation results under 18 dose-toxicity scenarios and 3 time-to-toxicity and accrual rate settings. 
Values shown are averages over simulated trials and scenarios.
The maximum sample size is $N^* = 60$. PCA, POA, PUA, PCS, POS, and PUS are in \%, DS, DE, SE, SD, ED, and ES are in $1/10^{3}$, and Dur is in days.} \label{supp:tbl:simu_60} \\
\hline \hline
\multicolumn{1}{c}{\multirow{2}{*}{Design}} & \multicolumn{3}{c}{Allocation} & \multicolumn{3}{c}{Selection} & \multicolumn{6}{c}{Risk} & Speed \\
\cline{2-4} \cline{5-7} \cline{8-13} \cline{14-14}
& PCA & POA & PUA & PCS & POS & PUS & DS & DE & SE & SD & ED & ES & Dur \\ \hline 
\endfirsthead

\multicolumn{14}{c}%
{ \tablename\ \thetable{} -- continued from previous page} \\
\hline 
\multicolumn{1}{c}{\multirow{2}{*}{Design}} & \multicolumn{3}{c}{Allocation} & \multicolumn{3}{c}{Selection} & \multicolumn{6}{c}{Risk} & Speed \\
\cline{2-4} \cline{5-7} \cline{8-13} \cline{14-14}
& PCA & POA & PUA & PCS & POS & PUS & DS & DE & SE & SD & ED & ES & Dur \\ \hline
\endhead

\hline \multicolumn{14}{r}{{Continued on next page}} \\
\endfoot

\hline \hline
\endlastfoot

 \multicolumn{14}{c}{Setting 1} \\ 
mTPI-2 & 42.8 & 21.4 & 35.8 & 58.6 & 11.8 & 29.6 & 0.0 & 0.0 & 0.0 & 0.0 & 0.0 & 0.0 & 1046 \\
TITE-TPI & 40.6 & 21.0 & 38.4 & 58.4 & 12.1 & 29.5 & 15.7 & 2.3 & 19.9 & 47.5 & 1.1 & 25.1 & 662 \\
POD-TPI & 40.8 & 23.7 & 35.4 & 59.0 & 12.5 & 28.5 & 24.1 & 2.8 & 20.6 & 21.0 & 0.0 & 7.4 & 664 \\
CRM & 45.0 & 23.0 & 32.0 & 63.8 & 26.1 & 10.1 & 0.0 & 0.0 & 0.0 & 0.0 & 0.0 & 0.0 & 1044 \\
TITE-CRM & 43.9 & 24.9 & 31.2 & 63.5 & 26.4 & 10.0 & 25.6 & 0.3 & 15.0 & 10.3 & 0.0 & 25.1 & 664 \\
BOIN & 42.6 & 20.8 & 36.7 & 61.9 & 16.2 & 21.9 & 0.0 & 0.0 & 0.0 & 0.0 & 0.0 & 0.0 & 1045 \\
TITE-BOIN & 41.2 & 20.9 & 37.9 & 61.7 & 17.3 & 21.0 & 13.7 & 2.0 & 17.4 & 47.0 & 0.1 & 23.8 & 662 \\
\hdashline
mTPI-2-LA & 42.1 & 21.9 & 36.0 & 58.3 & 11.7 & 30.0 & 0.0 & 0.0 & 0.0 & 0.0 & 0.0 & 0.0 & 861 \\
TITE-TPI-NS & 40.4 & 21.6 & 38.0 & 58.4 & 12.3 & 29.2 & 16.6 & 5.1 & 23.8 & 53.9 & 1.4 & 27.6 & 627 \\
TITE-TPI-PSR & 40.6 & 19.6 & 39.8 & 57.8 & 11.4 & 30.8 & 8.4 & 2.4 & 12.4 & 57.1 & 1.3 & 32.2 & 654 \\
TITE-TPI-PSR2 & 40.4 & 18.5 & 41.1 & 57.4 & 11.3 & 31.3 & 0.0 & 0.0 & 0.0 & 63.6 & 1.4 & 15.6 & 815 \\
POD-TPI-NS & 40.8 & 24.2 & 35.0 & 58.7 & 13.1 & 28.3 & 24.4 & 5.7 & 24.3 & 30.8 & 1.0 & 8.8 & 626 \\
POD-TPI-PSR & 41.5 & 22.2 & 36.3 & 58.8 & 12.3 & 28.8 & 11.9 & 2.6 & 13.0 & 29.1 & 0.8 & 8.4 & 671 \\
POD-TPI-PSR2 & 41.4 & 20.5 & 38.1 & 58.0 & 11.3 & 30.7 & 0.0 & 0.0 & 0.0 & 28.6 & 0.6 & 3.2 & 835 \\
\hline
\multicolumn{14}{c}{Setting 2} \\ 
mTPI-2 & 43.0 & 20.7 & 36.3 & 58.7 & 11.0 & 30.2 & 0.0 & 0.0 & 0.0 & 0.0 & 0.0 & 0.0 & 1061 \\
TITE-TPI & 40.1 & 24.1 & 35.8 & 58.5 & 13.1 & 28.4 & 31.1 & 6.0 & 31.1 & 38.6 & 0.9 & 20.1 & 675 \\
POD-TPI & 39.8 & 27.2 & 33.0 & 58.9 & 14.0 & 27.1 & 45.8 & 7.2 & 32.5 & 16.9 & 0.0 & 6.5 & 678 \\
CRM & 44.8 & 22.9 & 32.2 & 62.8 & 26.8 & 10.4 & 0.0 & 0.0 & 0.0 & 0.0 & 0.0 & 0.0 & 1066 \\
TITE-CRM & 43.3 & 26.8 & 29.9 & 63.9 & 25.8 & 10.3 & 46.7 & 0.9 & 26.5 & 7.8 & 0.0 & 21.7 & 676 \\
BOIN & 42.7 & 21.0 & 36.3 & 61.6 & 16.8 & 21.6 & 0.0 & 0.0 & 0.0 & 0.0 & 0.0 & 0.0 & 1064 \\
TITE-BOIN & 40.7 & 23.8 & 35.5 & 61.5 & 18.5 & 20.0 & 28.9 & 5.3 & 27.7 & 37.6 & 0.1 & 19.8 & 674 \\
\hdashline
mTPI-2-LA & 41.7 & 22.7 & 35.6 & 58.3 & 11.7 & 30.0 & 0.0 & 0.0 & 0.0 & 0.0 & 0.0 & 0.0 & 890 \\
TITE-TPI-NS & 39.6 & 25.7 & 34.7 & 58.7 & 13.4 & 28.0 & 33.1 & 12.5 & 38.1 & 42.2 & 1.0 & 22.3 & 640 \\
TITE-TPI-PSR & 40.1 & 23.1 & 36.7 & 58.3 & 12.9 & 28.8 & 21.6 & 7.1 & 25.5 & 45.2 & 1.2 & 26.6 & 670 \\
TITE-TPI-PSR2 & 41.0 & 19.3 & 39.7 & 58.2 & 11.2 & 30.6 & 0.0 & 0.0 & 0.0 & 51.2 & 1.3 & 14.4 & 842 \\
POD-TPI-NS & 39.5 & 29.1 & 31.5 & 59.3 & 14.5 & 26.2 & 48.8 & 14.2 & 40.7 & 22.3 & 0.8 & 7.2 & 642 \\
POD-TPI-PSR & 40.3 & 26.0 & 33.7 & 59.2 & 13.1 & 27.7 & 30.6 & 7.9 & 27.6 & 21.7 & 0.7 & 7.4 & 683 \\
POD-TPI-PSR2 & 41.1 & 21.5 & 37.5 & 57.4 & 12.0 & 30.6 & 0.0 & 0.0 & 0.0 & 21.2 & 0.6 & 3.2 & 864 \\
\hline
\multicolumn{14}{c}{Setting 3} \\
mTPI-2 & 42.6 & 21.2 & 36.2 & 58.3 & 11.3 & 30.4 & 0.0 & 0.0 & 0.0 & 0.0 & 0.0 & 0.0 & 776 \\
TITE-TPI & 39.7 & 21.2 & 39.1 & 58.3 & 12.7 & 29.0 & 21.2 & 2.8 & 24.2 & 58.4 & 2.7 & 38.7 & 410 \\
POD-TPI & 40.1 & 23.4 & 36.6 & 58.8 & 13.2 & 28.1 & 25.5 & 3.2 & 24.3 & 39.0 & 0.1 & 21.1 & 412 \\
CRM & 44.7 & 23.1 & 32.2 & 63.1 & 26.7 & 10.2 & 0.0 & 0.0 & 0.0 & 0.0 & 0.0 & 0.0 & 777 \\
TITE-CRM & 43.8 & 24.6 & 31.6 & 63.6 & 26.3 & 10.1 & 34.0 & 0.3 & 17.6 & 11.3 & 0.0 & 41.3 & 415 \\
BOIN & 42.4 & 21.2 & 36.4 & 61.3 & 16.9 & 21.8 & 0.0 & 0.0 & 0.0 & 0.0 & 0.0 & 0.0 & 777 \\
TITE-BOIN & 40.5 & 20.1 & 39.4 & 63.0 & 16.3 & 20.7 & 17.6 & 2.4 & 21.3 & 58.2 & 0.7 & 42.1 & 409 \\
\hdashline
mTPI-2-LA & 41.9 & 22.2 & 35.9 & 58.5 & 11.6 & 29.9 & 0.0 & 0.0 & 0.0 & 0.0 & 0.0 & 0.0 & 613 \\
TITE-TPI-NS & 39.2 & 21.2 & 39.5 & 58.4 & 13.1 & 28.5 & 20.5 & 4.9 & 27.0 & 68.6 & 4.9 & 44.0 & 372 \\
TITE-TPI-PSR & 39.2 & 18.6 & 42.2 & 57.5 & 12.0 & 30.6 & 11.1 & 2.3 & 13.3 & 72.0 & 5.2 & 54.4 & 395 \\
TITE-TPI-PSR2 & 39.5 & 17.7 & 42.8 & 57.5 & 11.5 & 31.0 & 0.0 & 0.0 & 0.0 & 84.6 & 4.1 & 22.9 & 563 \\
POD-TPI-NS & 39.8 & 23.1 & 37.2 & 59.3 & 13.3 & 27.5 & 24.5 & 5.8 & 28.8 & 55.7 & 2.4 & 25.6 & 373 \\
POD-TPI-PSR & 40.2 & 21.1 & 38.7 & 58.8 & 12.1 & 29.1 & 15.2 & 2.4 & 15.3 & 54.1 & 2.0 & 24.0 & 404 \\
POD-TPI-PSR2 & 40.2 & 20.0 & 39.8 & 58.0 & 11.3 & 30.7 & 0.0 & 0.0 & 0.0 & 51.8 & 0.2 & 8.8 & 574 \\
\end{longtable}
\end{center}

\clearpage

\subsection{Sensitivity Analyses with Different Time-to-Toxicity Models}
\label{supp:sec:t_model_spec}

We have listed several possible specifications of the time-to-toxicity model in Section \ref{sec:model_tite} and Supplementary Section \ref{supp:sec:model_tite}.
To explore how these specifications can affect the operating characteristics of a design, we conduct additional simulation studies using TITE-TPI and POD-TPI as examples. 
We consider five different time-to-toxicity models: (1) uniform distribution (U, default); (2) piecewise uniform distribution with 3 sub-intervals (PU3); (3) piecewise uniform distribution with 9 sub-intervals (PU9); (4) discrete hazard model (DH); and (5) piecewise constant hazard model with 3 sub-intervals (PCH3). Recall that the true distribution of $[T \mid Z, Y = 1]$ is a truncated Weibull distribution.  The maximum sample size here is $N^* = 36$.

Specifically, for the PU3 model, we consider 3 equal-length sub-intervals, $h_k = kW/K$ for $K = 3$ and $k = 1, 2, 3$.
A $\Dir(1, 1, 1)$ prior is assumed for the sub-interval weights $(\omega_1, \omega_2, \omega_3)$ (see Section \ref{supp:sec:model_tite_pwunif}). 
For the PU9 model, we consider 9 equal-length sub-intervals, $h_k = kW/K$ for $K = 9$ and $k = 1, \ldots, 9$.
A $\Dir(1, \ldots, 1)$ prior is assumed for the sub-interval weights $(\omega_1, \ldots, \omega_9)$. 
For the DH model, we assume a $\Beta(0.5, 0.5)$ prior for $\omega_k$, which is the discrete hazard at time $h_k$ (see Section \ref{supp:sec:model_tite_dhm}).
For the PCH3 model, we consider 3 equal-length sub-intervals, $h_k = kW/K$ for $K = 3$.
We follow \cite{liu2013bayesian} and assume a $\text{Gamma}(K / [2W(K - k + 0.5)], 1/2)$ prior for $\omega_k$, which is the hazard in the $k$th sub-interval (see Section \ref{supp:sec:model_tite_pchm}).

Table \ref{supp:tbl:simu_result_t_model} summarizes the simulation results.
To better understand the effect of the time-to-toxicity model, we also report the percentage of patients who have experienced toxicity (POT).
The performances of TITE-TPI and POD-TPI with different time-to-toxicity models are generally similar. 
The average number of DLTs in the trial is $N^* \times \text{POT} \approx 36 \times 20\% = 7.2$. As a result, there is very limited information for estimating the true time-to-toxicity distribution, and the specification of the time-to-toxicity model matters little.
Using the discrete hazard model or the piecewise constant hazard model, when many DLTs are late-onset (e.g., under Setting 2), the pending patients are weighted less, making the designs safer in such situations.
For example, under Setting 2, the designs with the DH and PCH3 models have lower POA and POT and make less frequent aggressive incompatible decisions (DS, DE, and SE).
This is consistent with the results reported in \cite{yuan2011robust} and \cite{liu2013bayesian}.

\begin{center}
\begin{longtable}{l@{\extracolsep{4pt}}c@{\extracolsep{4pt}}c@{\extracolsep{4pt}}c@{\extracolsep{4pt}}c@{\extracolsep{4pt}}c@{\extracolsep{4pt}}c@{\extracolsep{4pt}}c@{\extracolsep{4pt}}c@{\extracolsep{4pt}}c@{\extracolsep{4pt}}c@{\extracolsep{4pt}}c@{\extracolsep{4pt}}c@{\extracolsep{4pt}}c@{\extracolsep{4pt}}c}
\caption{Summary of simulation results using TITE-TPI and POD-TPI with 5 different time-to-toxicity models. 
Values shown are averages over simulated trials and scenarios. 
PCA, POA, PUA, POT, PCS, POS, and PUS are in \%, DS, DE, SE, SD, ED, and ES are in $1/10^{3}$, and Dur is in days.} \label{supp:tbl:simu_result_t_model} \\
\hline \hline
\multicolumn{1}{c}{\multirow{2}{*}{$T$ Model}} & \multicolumn{4}{c}{Allocation} & \multicolumn{3}{c}{Selection} & \multicolumn{6}{c}{Risk} & Speed \\
\cline{2-5} \cline{6-8} \cline{9-14} \cline{15-15}
& PCA & POA & PUA & POT & PCS & POS & PUS & DS & DE & SE & SD & ED & ES & Dur \\ \hline 
\endfirsthead

\multicolumn{15}{c}%
{ \tablename\ \thetable{} -- continued from previous page} \\
\hline 
\multicolumn{1}{c}{\multirow{2}{*}{$T$ Model}} & \multicolumn{4}{c}{Allocation} & \multicolumn{3}{c}{Selection} & \multicolumn{6}{c}{Risk} & Speed \\
\cline{2-5} \cline{6-8} \cline{9-14} \cline{15-15}
& PCA & POA & PUA & POT & PCS & POS & PUS & DS & DE & SE & SD & ED & ES & Dur \\ \hline
\endhead

\hline \multicolumn{15}{r}{{Continued on next page}} \\
\endfoot

\hline \hline
\endlastfoot

\multicolumn{15}{c}{Setting 1, TITE-TPI} \\ 
U & 32.2 & 21.5 & 46.3 & 18.8 & 50.5 & 16.6 & 32.9 & 12.7 & 3.3 & 21.9 & 55.4 & 1.5 & 29.7 & 436 \\
PU3 & 32.2 & 21.5 & 46.3 & 18.7 & 51.3 & 16.4 & 32.3 & 12.7 & 3.6 & 20.8 & 55.5 & 1.6 & 29.7 & 435 \\
PU9 & 32.3 & 21.5 & 46.1 & 18.8 & 51.8 & 15.7 & 32.5 & 13.2 & 3.3 & 21.4 & 54.8 & 1.6 & 29.6 & 436 \\
DH & 32.2 & 21.3 & 46.5 & 18.7 & 51.2 & 16.0 & 32.7 & 13.1 & 3.2 & 21.2 & 55.6 & 2.1 & 32.0 & 435 \\
PCH3 & 32.3 & 21.3 & 46.4 & 18.6 & 51.5 & 16.3 & 32.2 & 12.6 & 3.2 & 20.9 & 56.8 & 2.1 & 31.5 & 433 \\
\hdashline
\multicolumn{15}{c}{Setting 1, POD-TPI} \\ 
U & 32.5 & 24.5 & 43.0 & 19.8 & 51.4 & 18.4 & 30.2 & 20.5 & 4.1 & 23.0 & 23.8 & 0.1 & 10.3 & 437 \\
PU3 & 32.4 & 24.3 & 43.3 & 19.8 & 51.1 & 18.2 & 30.7 & 21.5 & 3.7 & 23.3 & 23.6 & 0.1 & 10.8 & 437 \\
PU9 & 32.5 & 24.2 & 43.3 & 19.8 & 51.2 & 18.3 & 30.5 & 21.3 & 3.9 & 23.0 & 24.5 & 0.1 & 10.9 & 437 \\
DH & 32.3 & 24.4 & 43.3 & 19.7 & 51.2 & 18.6 & 30.2 & 21.8 & 4.1 & 22.6 & 26.5 & 0.2 & 13.9 & 437 \\
PCH3 & 32.9 & 24.0 & 43.1 & 19.7 & 51.8 & 17.9 & 30.2 & 20.5 & 4.1 & 22.5 & 26.3 & 0.2 & 13.6 & 436 \\
\hline
\multicolumn{15}{c}{Setting 2, TITE-TPI} \\ 
U & 32.0 & 24.5 & 43.5 & 19.8 & 51.8 & 18.2 & 30.0 & 25.1 & 8.4 & 36.0 & 45.8 & 1.5 & 25.6 & 445 \\
PU3 & 31.5 & 23.9 & 44.6 & 19.7 & 51.1 & 17.5 & 31.4 & 22.0 & 7.8 & 33.0 & 49.5 & 2.4 & 28.1 & 445 \\
PU9 & 32.0 & 24.1 & 44.0 & 19.7 & 51.7 & 17.9 & 30.4 & 22.6 & 7.8 & 34.5 & 48.1 & 1.9 & 26.5 & 445 \\
DH & 31.4 & 22.6 & 46.0 & 19.1 & 51.1 & 16.7 & 32.2 & 17.4 & 6.3 & 29.5 & 56.1 & 4.4 & 35.2 & 444 \\
PCH3 & 31.5 & 22.8 & 45.8 & 19.1 & 51.0 & 16.8 & 32.1 & 18.8 & 6.8 & 30.5 & 55.1 & 3.7 & 33.8 & 443 \\
\hdashline
\multicolumn{15}{c}{Setting 2, POD-TPI} \\ 
U & 31.8 & 27.2 & 41.0 & 20.9 & 52.0 & 19.5 & 28.5 & 38.4 & 10.3 & 38.3 & 20.0 & 0.0 & 9.1 & 449 \\
PU3 & 31.9 & 25.8 & 42.2 & 20.4 & 51.8 & 18.8 & 29.4 & 31.1 & 8.6 & 33.6 & 27.5 & 0.2 & 15.0 & 446 \\
PU9 & 32.0 & 26.3 & 41.8 & 20.5 & 51.7 & 19.5 & 28.8 & 33.4 & 8.8 & 35.6 & 25.5 & 0.1 & 12.9 & 448 \\
DH & 31.6 & 24.5 & 43.9 & 19.9 & 51.6 & 18.0 & 30.4 & 24.3 & 7.1 & 29.4 & 39.5 & 0.6 & 26.7 & 446 \\
PCH3 & 32.0 & 24.9 & 43.1 & 20.0 & 51.4 & 18.7 & 29.9 & 26.5 & 7.4 & 30.0 & 35.4 & 0.6 & 23.2 & 446 \\
\hline
\multicolumn{15}{c}{Setting 3, TITE-TPI} \\ 
U & 31.4 & 21.1 & 47.5 & 18.6 & 50.5 & 15.8 & 33.7 & 15.8 & 3.6 & 25.5 & 64.1 & 4.0 & 44.6 & 290 \\
PU3 & 31.6 & 21.3 & 47.1 & 18.6 & 51.0 & 16.2 & 32.8 & 16.0 & 3.7 & 25.7 & 63.1 & 3.9 & 44.0 & 291 \\
PU9 & 31.3 & 21.2 & 47.5 & 18.6 & 50.9 & 15.8 & 33.3 & 15.8 & 3.7 & 25.4 & 63.9 & 4.1 & 44.0 & 291 \\
DH & 31.5 & 20.7 & 47.8 & 18.4 & 51.0 & 15.9 & 33.1 & 15.4 & 3.6 & 24.4 & 64.8 & 5.5 & 47.1 & 290 \\
PCH3 & 31.5 & 21.0 & 47.6 & 18.5 & 51.9 & 15.5 & 32.6 & 15.2 & 3.5 & 24.4 & 65.7 & 5.2 & 45.3 & 290 \\
\hdashline
\multicolumn{15}{c}{Setting 3, POD-TPI} \\ 
U & 31.7 & 23.3 & 45.0 & 19.4 & 51.5 & 17.4 & 31.0 & 18.8 & 4.0 & 25.8 & 42.8 & 0.1 & 26.4 & 294 \\
PU3 & 31.8 & 23.1 & 45.0 & 19.4 & 51.4 & 17.4 & 31.1 & 20.2 & 3.8 & 25.7 & 42.1 & 0.2 & 26.5 & 293 \\
PU9 & 31.6 & 23.7 & 44.7 & 19.5 & 51.6 & 17.7 & 30.8 & 20.2 & 3.9 & 26.2 & 42.6 & 0.1 & 27.5 & 294 \\
DH & 31.9 & 23.4 & 44.7 & 19.4 & 51.6 & 17.6 & 30.9 & 22.5 & 4.5 & 25.7 & 39.9 & 0.8 & 30.0 & 293 \\
PCH3 & 31.8 & 22.8 & 45.4 & 19.3 & 50.9 & 16.9 & 32.2 & 20.1 & 4.1 & 24.4 & 42.3 & 0.8 & 31.3 & 293 \\
\end{longtable}
\end{center}

\end{document}